\documentclass[a4paper,12pt,reqno]{book}%
\usepackage{amsmath}%
\usepackage{amsfonts}%
\usepackage{amssymb}%
\usepackage{graphicx}
\usepackage{longtable}
\usepackage[T1]{fontenc}
\usepackage[francais]{babel}

\newtheorem{theorem}{Th\'eor\`eme}

\newtheorem{corollary}[theorem]{Corollaire}

\newtheorem{property}[theorem]{Propri\'et\'e}
\newtheorem{definition}[theorem]{D\'efinition}
\newtheorem{example}[theorem]{Exemple}

\newtheorem{lemma}[theorem]{Lemme}
\newtheorem{notation}[theorem]{Notation}
\newtheorem{problem}[theorem]{R\`egle}
\newtheorem{proposition}[theorem]{Proposition}
\newtheorem{remark}[theorem]{Remarque}

\newenvironment{proof}[1][Preuve]{\textbf{#1.} }{\ \rule{0.5em}{0.5em}}

\begin{document}

\frontmatter
\begin{center}
UNIVERSITE D'ANTANANARIVO
\end{center}
\begin{center}
\textbf{FACULTE DES SCIENCES}
\end{center}
\begin{center}
\texttt{FORMATION DOCTORALE EN PHYSIQUE}
\end{center}
\begin{center}
\texttt{DEPARTEMENT DE PHYSIQUE}
\end{center}
\begin{center}
\textbf{Laboratoire de Rh\'eologie des Suspensions}
\end{center}
\begin{center}
\LARGE{\textbf{MEMOIRE D'HABILITATION A DIRIGER DES RECHERCHES}}
\end{center}
\begin{center}
\textit{option}: M\'ecanique et physique des suspensions
\end{center}
\begin{center}
\textit{sur}:
\end{center}
\begin{center}
\large{\textbf{MATRICES DE COMMUTATION TENSORIELLE}}:\\
\large{\textbf{DE L'EQUATION DE DIRAC VERS UNE APPLICATION EN PHYSIQUE DES PARTICULES}}
\end{center}
\begin{center}
\textit{pr\'esent\'e par}
\end{center}
\begin{center}
\large{RAKOTONIRINA Christian}
\end{center}

\begin{center}
\textit{devant la commission d'examen compos\'ee de:}
\end{center}
\begin{tabbing}
\textit{Directeur de HDR}:\=Monsieur \=RATIARISON Adolphe Andriamanga \quad \quad\= Professeur Titulaire \kill\\

\textit{Pr\a'esident}:     \>Monsieur \>RAKOTOMAHANINA RALAISOA  \quad \quad \> \\[5pt]  
\>\> Emile  \quad \> Professeur \a'em\a'erite \\[8pt]

\textit{Rapporteurs}:\>Madame \>RANDRIAMANANTANY Zely Arivelo\quad \quad\> Professeur titulaire\\[5pt]
                     \>Madame \>RAZANAJATOVO Mariette \quad \quad \> Professeur titulaire\\[5pt]
                      \>Monsieur \>RAKOTOMALALA Jean Lalaina \quad \quad\> Professeur\\[8pt]
 \textit{Examinateurs}:\>Monsieur \>RANAIVO-NOMENJANAHARY Flavien \quad \quad\> Professeur titulaire \\[5pt]  
                       \>Monsieur \>ANDRIAMAMPIANINA José \quad \quad\> Professeur\\[8pt]
	\textit{Directeur de HDR}:\>Monsieur \>RATIARISON Adolphe Andriamanga \quad \quad\> Professeur titulaire
\end{tabbing}

\title{MATRICES DE COMMUTATION TENSORIELLE: DE L'EQUATION DE DIRAC VERS UNE APPLICATION EN PHYSIQUE DES PARTICULES}

\author{RAKOTONIRINA Christian}

\maketitle
\begin{center}
\textbf{Abstract}
\end{center}
We construct two sets of representations of the Dirac equation. They transform themselves from one to another in multiplying by the tensor commutation matrix (TCM) $2\otimes 2$. The Gauss matrix can lead us to the Cholesky decomposition. The TCMs can be used in order to obtain different transforms of some matrix equations to linear matrix equations of the form $AX=B$. A TCM $n\otimes n$ is expressed in terms of the $n\otimes n$-Gell-Mann matrices. In order to generalize this expression to which of TCMs $n\otimes p$, we introduced what we call rectangle Gell-Mann matrices. The electric charge operator (ECO) for eight leptons and quarks of the Standard Model (SM) of a same generation proposed by Zenczykowski is expressed in terms of the TCM $2\otimes 2$. An ECO including all the fermions of the SM is constructed in terms of TCM $4\otimes 4$. Then the eigenvalues of a TCM take sense.\\


\begin{center}
\textbf{R\'esum\'e}
\end{center}
Nous avons construit deux ensembles de repr\'esentations de l'\'equation de Dirac. Ils se transforment l'un en l'autre en multipliant par la matrice de commutation tensorielle  (MCT) $2\otimes2$. La m\'ethode de Gauss matricielle peut nous conduire \`a la d*\'ecomposition de Cholesky. Les MCT peuvent \^etre utilis\'ees afin d'obtenir diff\'erentes transform\'ees de certaines \'equations matricielles en \'equations matricielles lin\'eaires de la forme $AX=B$. Une MCT $n\otimes n$ a \'et\'e exprim\'ee en termes de matrices de Gell-Mann $n\otimes n$. Afin de g\'en\'raliser cette expression \`a celles des MCT $n\otimes p$, nous avons introduit ce que nous appelons matrices de Gell-Mann rectangles. L'op\'erateur charges \'electriques (OCE), pour huit leptons et quarks du mod\`ele standard (MS) d'une m\^eme g\'en\'eration propos\'e par Zenczykowski a \'et\'e exprim\'e en termes de  MCT $2\otimes 2$. Un OCE incluant tous les fermions du MS a \'et\'e construit en termes de MCT $4\otimes 4$. Alors, les valeurs propres d'une MCT ont pris sens.\\

\textit{Mots cl\'es} :   Equation de Dirac, Equation de Dirac-Sidharth, produit de Kronecker, Système d'\'quations lin\'eaires, Cholesky, Matrices de Gell-Mann g\'en\'eralis\'ees,  leptons, quarks.

\chapter*{Avant-propos}
Ce document pr\'{e}sente un ensemble de travaux de recherches que nous avons men\'{e} depuis l'obtention de notre doctorat en 2003. Nous avons rassemblons dans ce m\'{e}moire nos publications pour obtenir certaine coh\'{e}rence. Nous avons d\'{e}lib\'{e}r\'{e}ment choisi de ne pas y inclure notre travail sur les \'{e}quations diff\'{e}rentielles. En effet, l'inclusion de ce travail rendrait difficile, pour ne pas dire impossible, la coh\'{e}rence. Ces activit\'{e}s se sont d\'{e}roul\'{e}es \`a l'Institut Sup\'{e}rieur de T\'{e}chnologie d'Antananarivo, IST-T, o\`u nous enseignons en tant que ma\^itre de conf\'{e}rences et \`a l'universit\'{e} d'Antananarivo, D\'{e}partement de Physique, Laboratoire de Rh\'{e}ologies des Suspensions, LRS, o\`u nous avons travaill\'{e} pour finir ce m\'{e}moire et avons encadr\'{e} un \'{e}tudiant en DEA. Nous avons aussi travaill\'{e} avec Madagascar-Institut National des Sciences et Techniques Nucl\'{e}aires, Madagascar-INSTN, o\`u nous avions travaill\'{e} de 1999 \`a 2003 pour preparer notre th\`ese de doctorat et avons encadr\'{e} un \'{e}tudiant en DEA.\\
Ce travail est la suite logique de notre formation en math\'{e}matiques et notre th\`ese de physique th\'{e}orique. Il est incontestable que les math\'{e}matiques sont des outils que tous les pays peuvent utiliser pour contribuer au d\'eveloppement des sciences. R\'eciproquement, les sciences, de par leurs d\'{e}veloppements, enrichissent les math\'{e}matiques. Comme son titre l'indique, ce travail est un exemple qui met en \'evidence cette dialectique entre sciences, physique et math\'{e}matiques. En effet, notre \'{e}tude math\'{e}matiques de l'\'{e}quation de Dirac-Sidharth, qui est une \'{e}quation relativiste des particules de spin-$\frac{1}{2}$ nous conduit aux matrices que nous appelons matrices de commutation tensorielle. Ces matrices sont tr\`es utiles en math\'{e}matiques, dans les r\'{e}solutions des \'{e}quations matricielles. Elles se g\'{e}n\'{e}ralisent \`a ce que nous appelons matrices de permutation tensorielles qui, \`a leur tour, ont une application en m\'{e}canique quantique.\\
 Nous avons exprim\'{e} ces matrices en termes de matrices de Gell-Mann g\'{e}n\'{e}ralis\'{e}es, qui sont des matrices de la physique des particules, en esp\'{e}rant trouver application de ces matrices dans ce domaine de la physique. Nous pensons que nous nous dirigeons vers l'application de ces matrices en physique des particules lorsque, lors de la recherche que nous avons effectu\'{e} dans le Laboratoire de Rh\'{e}ologies des Suspensions, LRS, Universit\'{e} d'Antananarivo, nous avons exprim\'{e} l'op\'{e}rateur charges \'{e}l\'{e}ctrique des particules fermions du Mod\`ele Standard propos\'{e} par Zenczykowski en termes de matrice de commutation tensorielle. Ainsi, ce travail s'inscrit dans le cadre de l'application de l'alg\`ebre lin\'{e}aire et multilin\'{e}aire en physique.\\
 L'insertion de la subsection sur la m\'{e}thode de Cholesky dans ce m\'{e}moire semble perturber un peu la coh\'{e}rence. Mais comme cette subsection fait partie de l'alg\`ebre lin\'{e}aire et qu'elle figure parmi les fruits de notre recherche p\'{e}dagogique effectu\'{e}e \`a l'Institut Sup\'{e}rieur de Technologie d'Antananarivo, IST-T, nous nous excusons aupr\`es des lecteurs cette incoh\'{e}rence apparente. Nous disons "apparente" parce que pour un lecteur averti ce ne sera plus du tout une incoh\'{e}rence, puisque cette subsection est un pont menant vers l'application des matrices de commutation tensorielle aux \'{e}tudes des \'{e}quations matricielles. Cette subsection, constitue une partie d'un chapitre de notre cours d'Analyse Num\'{e}rique \`a l'Institut Sup\'{e}rieur de Technologie d'Antananarivo, IST-T, est donn\'{e}e en anglais aux \'{e}tudiants, mais expliqu\'{e}e en Malagasy.
\maketitle

\tableofcontents

\mainmatter

\chapter*{INTRODUCTION}

Le produit tensoriel de matrices ou produit de kronecker ou encore produit direct de matrices n'est pas commutatif en g\'en\'eral. Cependant, le produit de certaines matrices avec un produit tensoriel de matrices commute ce produit. Ces matrices sont les matrices de commutation tensorielle (MCT) ou matrices de commutation de kronecker (Cf. par exemple \cite{Lin12,Zhang11,Gilchrist11,Fujii01}). Les MCT ont des applications en m\'ecanique quantique, en particulier en th\'eorie quantique de l'information (Cf. par exemple \cite{Zhang11,Gilchrist11,Fujii01,Chefles00,Wilmott08}). Nous avons d\'ecouvert les MCT quand nous travallions sur l'\'equation de Dirac \cite{Wang01,Rakotonirina03}, qui est l'\'equation quantique relativiste des particules de spin $\frac{1}{2}$ telles qu'un \'electron. Cependant, dans ce m\'emoire c'est l'\'{e}quation de Dirac-Sidharth qui nous conduit \`a ces matrices. C'est une \'{e}quation de Dirac modifi\'ee. Les repr\'esentations qui sont communes \`a ces \'equations sont ce qui nous y conduisent.\\
Nous pensons que ces matrices m\'eritent plus d'attention puisqu'elles ont aussi des applications en math\'ematiques, plus pr\'ecisement pour les r\'esolutions des \'equations matricielles (Cf. par exemple \cite{Lin12,Li13,Zhou12}). Lorsque nous \'etudiions \cite{Zenczykowski07}, l'expression de la MCT $2\otimes2$ en termes  de matrices de Pauli,
\begin{equation}\label{eq5}
 \textbf{U}_{2\otimes2}= \frac{1}{2} I_2 \otimes I_{2} + \frac{1}{2} \sum_{i
   = 1}^3 \sigma^i \otimes \sigma^i 
\end{equation}
 nous a fait remarqu\'e que ce MCT pourrait avoir applications en physique des particules.\\
Les MCT se g\'en\'eralisent  aux matrices de permutation tensorielle (MPT), qui permutent le produit tensoriel de matrices. Les MPT ont aussi des applications en m\'ecanique quantique.\\
Ce m\'emoire est divis\'e en trois chapitres de la mani\`ere suivante. Dans le chapitre I, nous \'etudierons une \'equation de Dirac modifi\'ee, l'\'equation de Dirac-Sidharth. C'est l\`a que la MCT $2\otimes2$  
\begin{equation*}
\textbf{U}_{2\otimes2}=\left(%
\begin{array}{cccc}
  1 & 0 & 0 & 0 \\
  0 & 0 & 1 & 0 \\
  0 & 1 & 0 & 0 \\
  0 & 0 & 0 & 1 \\
\end{array}%
\right) 
\end{equation*}
nous appara\^itra comme quand nous \'etudiions l'\'equation de Dirac. Le chapitre II est l'\'etudes et applications math\'ematiques des MCT et leur g\'en\'eralisation aux MPT. Dans le chapitre III, nous g\'en\'eraliserons l'expression (\ref{eq5}) \`a  celle de la MCT $n\otimes  n$ en termes de matrices de Gell-Mann $n\times n$, qui sont des matrices de la physique des particules. Pour g\'en\'eraliser \`a son tour cette relation \`a celle de MCT $n\otimes p$, nous introduirons ceux que nous appelons  matrices de Gell-Mann rectangles. Nous verrons aussi dans ce chapitre comment exprimer une MPT en termes de matrices de Gell-Mann g\'en\'eralis\'ees. A la fin de ce chapitre, nous exprimerons l'op\'erateur charges \'electriques (OCE) de fermions propos\'e par Zenczykowski 
\begin{equation*}
Q=\frac{1}{2}\sigma_0\otimes\sigma_0\otimes\sigma_3+\frac{1}{6}\left(\sum_{i=1}^3\sigma_i\otimes \sigma_i\right)\otimes \sigma_0
\end{equation*} 
en termes de la MCT $2\otimes2$. Puis nous introduirons les MCT $3\otimes3$ et $4\otimes4$ pour obtenir des op\'erateurs charges \'electriques pour plus de fermions.    

\chapter{MCT A PARTIR DE L'EQUATION DE DIRAC-SIDHARTH}
Ce chapitre est bas\'e sur Ref. \cite{Raoelina11}
\section{EQUATION DE DIRAC-SIDHARTH ET PRODUIT TENSORIEL DE MATRICES    }

Nous allons construire l'\'equation de Dirac-Sidharth, \`a partir de l'hamiltonien de Sidharth, par quantification de l'\'energie et de l'impulsion dans l'alg\`ebre de Pauli. Nous allons r\'esoudre cette \'equation en utilisant le produit tensoriel de matrices.

Selon la relativit\'e restreinte d'Einstein \cite{Einstein1905}, la relation entre l'\'energie et l'impulsion est 
\begin{equation*}
E^2=c^{2}p^{2}+m^{2}c^{4}
\end{equation*}
\`a partir de laquelle nous pouvons d\'eduire l'\'equation de Klein-Gordon et l'\'equation de Dirac. Cette th\'eorie utilise le concept de l'espace-temps continu.  \\
L'espace-temps quantifi\'e a \'et\'e introduit pour la premi\`ere fois par Snyder \cite{Snyder472,Snyder471}, sous le nom de g\'eom\'etrie non commutative de Snyder, \`a cause de la modification sur les relations de commutation. Dans cette th\'eorie les relations de commutation sont \cite{Snyder472,Snyder471}\\
\begin{equation*}
\left[x^{\mu},x^{\nu}\right]=i\alpha\frac{\ell^{2}c^{2}}{\hbar}\left(x^{\mu}p^{\nu}-x^{\nu}p^{\mu}\right),
\end{equation*}
\begin{equation*}
\left[x^{\mu},p_{\nu}\right]=i\hbar\left[\delta^{\mu}_\nu+i\alpha\frac{\ell^{2}c^{2}}{\hbar^{2}}p^{\mu}p_{\nu}\right],
\end{equation*}
\begin{equation*}
\left[p_{\mu},p_{\nu}\right]=0
\end{equation*}

\noindent o\`u  $\ell$ est une \'echelle de longueur quelconque en physique. Par exemple, $\ell=\ell_p=\sqrt{\frac{\hbar G}{c^{3}}}\approx1.6\times10^{-33}cm$ la longueur de Planck , la longueur minimale possible de mesurer en physique, o\`u G est la constante gravitationnelle.
Comme consequence, la relation entre \'energie et implulsion est modifi\'ee et devient (en unit\'e SUN, $c =\hbar=1$)\cite{Sidharth04}
\begin{equation*}
E^2 = p^2 + m^2 + \alpha l^2p^4 
\end{equation*}
 ou (en unit\'e SI)\cite{Glinka101, Glinka102}
\begin{equation}\label{eq1}
E^2=c^{2}p^{2}+m^{2}c^{4}+\alpha\left(\frac{c}{\hbar}\right)^{2}\ell^{2}p^{4}
\end{equation}
o\`u $\alpha$ une constante adimensionnelle.
\begin{equation}\label{eq2}
\epsilon=\frac{\hbar c}{\sqrt{\alpha}\ell} 
\end{equation}
est l'\'energie de Planck d\^ue \`a l'\'echelle de longueur de Planck $\ell=\ell_p$ \cite{Glinka101, Glinka102}.
Alors, \cite{Glinka102}
\begin{equation*}
E^2=c^{2}p^{2}+m^{2}c^{4}+\frac{c^4p^4}{\epsilon^2}
\end{equation*}
 
 \noindent Le r\^ole fondamental de $\epsilon$ est expliqu\'e dans \cite{Glinka101, Glinka102, Glinka11}.\\
\noindent En fait, en appliquant l'hamiltonien de Snyder-Sidharth (\ref{eq1}) Sidharth a construit l'\'equation de Dirac-Sidharth \cite{Sidharth04,Sidharth05}, i.e.  l'\'equation de Dirac modifi\'ee d\^ue \`a la g\'eom\'etrie non commutative de l'espace de phase.\\
Dans la sous-section \ref{DDSE}, nous construirons l'\'equation de Dirac-Sidharth, \`a partir de la relation (\ref{eq1}), par quantification de l'\'energie et l'impulsion. Dans la sous-section \ref{RDSE}, nous r\'esoudrons l'\'equation de Dirac-Sidharth par utilisation du produit tensoriel de matrices. \\

\subsection{Une derivation de l'\'equation de Dirac-Sidharth}\label{DDSE}
Pour \'etablir l'\'equation de Dirac-Sidharth nous allons utiliser la m\'ethode de J.J. Sakurai \cite{Sakurai67} pour la d\'erivation de l'\'equation de Dirac.

La fonction d'onde d'une particule de spin-$\frac{1}{2}$ doit \^etre \`a deux composantes. Ainsi, pour  quantifier la relation  \'energy-impulsion relativiste afin d'obtenir  l'\'equation modifi\'ee de Klein-Gordon \cite{Sidharth04,Sidharth05}, ou \'equation de Klein-Gordon-Sidharth, d'une particule de spin-$\frac{1}{2}$, les operateurs qui prennent part dans la quantification doivent \^etre des  matrices $2\times2$ . Ainsi, prenons comme r\`egles de quantification 

$E   \longrightarrow  i\hbar\sigma^{0}\frac{\partial}{\partial t}=i\hbar\frac{\partial}{\partial t}$

$\vec{p} \longrightarrow  -i\hbar\sigma^{1}\frac{\partial}{\partial x^{1}}-i\hbar\sigma^{2}\frac{\partial}{\partial x^{2}}-i\hbar\sigma^{3}\frac{\partial}{\partial x^{3}}=-i\hbar\vec{\sigma}\vec{\nabla}=\hat{p}_{1}\sigma^{1}+\hat{p}_{2}\sigma^{2}+\hat{p}_{3}\sigma^{3}$

\noindent o\`u \begin{center}$\sigma_1$ = $\left(
\begin{array}{cc}
  0 & 1 \\
  1 & 0 \\
\end{array}
\right)$, $\sigma_2$\;=\;$\left(
\begin{array}{cc}
  0 & -i \\
  i & 0 \\
\end{array}
\right)$, $\sigma_3$\;=\;$\left(
\begin{array}{cc}
  1 & 0 \\
  0 & -1 \\
\end{array}
\right)$
\end{center} sont les matrices de Pauli.
Alors, nous avons d'abord l'\'equation de Klein-Gordon-Sidharth
\begin{equation*}
c^{2}\hbar^{2}\left(\frac{\partial^{2}}{c^{2}\partial t^{2}}-\Delta-m^{2}c^{2}-\alpha\frac{\ell^{2}}{\hbar^{2}}\vec{\nabla}^{4}\right)\phi=0
\end{equation*}
\begin{equation*}
\begin{split}
\left(i\hbar\frac{\partial}{\partial t}+ic\hbar\vec{\sigma}\vec{\nabla}\right)\frac{1}{mc^{2}}\left\{\sum^{+\infty}_{k=0}\left(-1\right)^k\left [\frac{i\sqrt{\alpha}}{mc\hbar}\ell \left(-i\hbar\vec{\sigma}\vec{\nabla}\right)^{2}\right]^k\right\}\times \\ \left(i\hbar\frac{\partial}{\partial t}-ic\hbar\vec{\sigma}\vec{\nabla}\right)\phi
 = \left[mc^{2}+i\sqrt{\alpha}\frac{c}{\hbar}\ell \left(-i\hbar\vec{\sigma}\vec{\nabla}\right)^{2}\right]\phi
 \end{split}
\end{equation*}
\noindent dont l'operateur agit sur la fonction d'onde \`a deux composantes $\phi$, qui est solution de l'\'equation de Klein-Gordon-Sidharth.
Soit
\begin{equation*}
\chi=\frac{1}{mc^{2}}\left\{\sum^{+\infty}_{k=0}\left(-1\right)^k\left [\frac{i\sqrt{\alpha}}{mc\hbar}\ell \left(-i\hbar\vec{\sigma}\vec{\nabla}\right)^{2}\right]^k\right\}\\
  \left(i\hbar\frac{\partial}{\partial t}-ic\hbar\vec{\sigma}\vec{\nabla}\right)\phi
\end{equation*}
\noindent alors, nous avons le syst\`eme d'\'equations aux d\'eriv\'ees partielles suivantes
\begin{equation*}
\left\{\begin{aligned}
i\hbar\frac{\partial}{c\partial t}\chi+i\hbar\vec{\sigma}\vec{\nabla}\chi & = mc\phi+i\sqrt{\alpha}\frac{\ell}{\hbar} \left(i\hbar\vec{\sigma}\vec{\nabla}\right)^{2}\phi\\
i\hbar\frac{\partial}{c\partial t}\phi-i\hbar\vec{\sigma}\vec{\nabla}\phi & = mc\chi-i\sqrt{\alpha}\frac{\ell}{\hbar} \left(i\hbar\vec{\sigma}\vec{\nabla}\right)^{2}\chi
\end{aligned} \right.
\end{equation*}
En additionnant et en soustrayant ces \'equations, et en transformant  les \'equations obtenues sous forme matricielle, nous avons l'\'equation de Dirac-Sidharth
\begin{equation*}
i\hbar\gamma^{\mu}_{D}\partial_{\mu}\psi_{D}-mc\psi_{D}-i\sqrt{\alpha}\ell\hbar\gamma^{5}_{D}\Delta\psi_{D}=0
\end{equation*}
\noindent dans la representation de Dirac (ou "Standard") des $\gamma$-matrices, o\`u \\
$\gamma^{0}_{D}=\begin{pmatrix}
\sigma^{0} & 0 \\
0 & -\sigma^{0}
\end{pmatrix}=\sigma^{3}\otimes\sigma^{0}$, $\gamma^{j}_{D}=\begin{pmatrix}
0 & \sigma^{j} \\
-\sigma^{j} & 0
\end{pmatrix}=i\sigma^{2}\otimes\sigma^{j}$, $j=1, 2, 3$, \\
 $\gamma^{5}_{D}=i\gamma^{0}_{D}\gamma^{1}_{D}\gamma^{2}_{D}\gamma^{3}_{D}=\begin{pmatrix}
0 & \sigma^{0} \\
\sigma^{0} & 0
\end{pmatrix}=\sigma^{1}\otimes\sigma^{0}$, et $\psi_{D}=\begin{pmatrix}
\chi+\phi\\
\chi-\phi
\end{pmatrix}$,\\
 $\Delta=\frac{\partial^{2}}{\partial x^{2}_{1}}+\frac{\partial^{2}}{\partial x^{2}_{2}}+\frac{\partial^{2}}{\partial x^{2}_{3}}$.\\
 
 Les matrices $4\times4$ $\left(\gamma^{\mu}_{D}\right)_{0\leq \mu \leq 3}$ satisfont les relations suivantes
 \begin{equation}\label{eq3prime}
\gamma^{\mu}_{D}\gamma^{\nu}_{D}+\gamma^{\nu}_{D}\gamma^{\mu}_{D}=2g^{\mu\nu}I_4, \; \; \mu, \nu \in\left\{0, 1, 2, 3\right\}
\end{equation}

\noindent o\`u $g^{\mu\nu}=0$ if $\mu\neq\nu$, \; $g^{jj}=-1, j \in\left\{1, 2, 3\right\}$, \; \; $g^{00}=1$

Nous savons que \cite{BjorkenDrell64}
\begin{equation*}
P\gamma^{5}=-\gamma^{5}P
\end{equation*}
Il s'ensuit que l'\'equation de Dirac-Sidharth est non invariant sous l'op\'erateur parit\'e (ou r\'eflexion dans l'espace) \cite{Sidharth09}.\\
L'\'equation
\begin{equation*}
i\hbar\gamma^{\mu}_{W}\partial_{\mu}\psi_{W}-mc\psi_{W}-i\sqrt{\alpha}\ell\hbar\gamma^{5}_{W}\Delta\psi_{W}=0
\end{equation*}
est l'\'equation de Dirac-Sidharth dans la representation de Weyl (ou chiral), o\`u  $\psi_{W}=\begin{pmatrix}
\chi\\
\phi
\end{pmatrix}$.\\
Ainsi, $\chi$ et $\phi$ sont les composantes (ou chiralit\'e) respectivement gauche et droite.
 
\subsection{R\'esolution de l'\'equation de Dirac-Sidharth}\label{RDSE}
Dans cette section, nous cherchons les solutions de l'\'equation de Dirac-Sidharth, en forme d'onde plane en utilisant le produit tensoriel de matrices. Nous avions utilis\'e cette m\'ethode  pour r\'esoudre l'\'equation de Dirac \cite{Rakotonirina03}.\\
Ainsi, cherchons une solution sous la forme
\begin{equation*}
\psi_{D}=U(p)e^{\frac{i}{\hbar}\left(\vec{p}\vec{x}-Et\right)}
\end{equation*}
Soit $\Psi$ un spineur \`a quatre composantes qui est un \'etat propre de $\hat{p}_{j}=-i\hbar\frac{\partial}{\partial x^{j}}$ et de $\hat{E}=i\hbar\frac{\partial}{\partial t}$, $\vec{p}=\begin{pmatrix}
p^{1}\\
p^{2}\\
p^{3}
\end{pmatrix}$, et $\vec{n}=\frac{\vec{p}}{p}=\begin{pmatrix}
n^{1}\\
n^{2}\\
n^{3}
\end{pmatrix}$.\\
L'\'equation de Dirac-Sidharth devient
\begin{multline*}
\sigma^{0}\otimes\sigma^{0}U(p)-\frac{2}{\hbar}cp\sigma^{1}\otimes\left(\frac{\hbar}{2}\vec{\sigma}\vec{n}\right)U(p)-mc^{2}\sigma^{3}\otimes\sigma^{0}U(p)\\
+c\sqrt{\alpha}p^{2}\frac{\ell}{\hbar}\sigma^{2}\otimes\sigma^{0}U(p)=0
\end{multline*}
Prenons $U(p)$ sous la forme
\begin{equation*}
U(p)=\varphi\otimes u
\end{equation*}
o\`u $u$ est un vecteur propre de l'op\'erateur spin $\frac{\hbar}{2}\vec{\sigma}\vec{n}$. $\varphi=\begin{pmatrix}
\varphi^{1}\\
\varphi^{2}
\end{pmatrix}$ et $u$ sont \`a deux componsantes.\\
Comme $u\neq0$, donc
\begin{equation}\label{eq4prime}
\left(\eta cp\sigma^{1}+mc^{2}\sigma^{3}-c\sqrt{\alpha}p^{2}\frac{\ell}{\hbar}\sigma^{2}\right)\varphi= E\varphi
\end{equation}
avec $\eta=
\begin{cases}
+1 & \text{spin haut}\\
-1 & \text{spin bas}
\end{cases}$ \\
En r\'esolvant cette \'equation par rapport \`a $\varphi^{1}$ et $\varphi^{2}$, nous avons 
 \begin{equation*}
\Psi_{+}=\sqrt{\frac{E+mc^{2}}{2E}}\begin{pmatrix}
1\\
\frac{\eta cp-i\frac{c}{\hbar}\sqrt{\alpha}p^{2}\ell}{mc^{2}+E}
\end{pmatrix}\otimes se^{\frac{i}{\hbar}\left(\vec{p}\vec{x}-Et\right)}
\end{equation*}
la solution \`a \'energie positive, o\`u $s=\frac{1}{\sqrt{2\left(1+n^{3}\right)}}\begin{pmatrix}
-n^{1}+in^{2}\\
1+n^{3}
\end{pmatrix}$ spin haut, $s=\frac{1}{\sqrt{2\left(1+n_{3}\right)}}\begin{pmatrix}
1+n^{3}\\
n^{1}+in^{2}
\end{pmatrix}$ spin bas.\\

D'apr\`es l'\'equation (\ref{eq4prime}), cette m\'ethode fait appara\^itre la matrice $h=\eta cp\sigma^{1}-c\sqrt{\alpha}p^{2}\frac{\ell}{\hbar}\sigma^{2}+mc^{2}\sigma^{3}$, ou $h=\eta cp\sigma^{1}-\frac{c^2p^2}{\epsilon}\sigma^{2} + mc^{2}\sigma^{3}$ (si $\ell$ est l'\'echelle de longueur de Planck), dont les valeurs propres sont les \'energies positive et negative. $h$ est comme un vecteur dans l'alg\`ebre de Pauli. Ainsi, l'\'energie d'une particule de spin-$\frac{1}{2}$ peut \^etre consid\'er\'ee comme un vecteur dans l'alg\`ebre de Pauli, dont la longueur ou l'intensit\'e est donn\'ee par la relation  \'energie-impulsion relativiste.
\begin{equation*}
h^{2}=E^{2}
\end{equation*}
\subsection{Representation de l'\'equation de Dirac-Sidharth}
\begin{definition}
Un syst\`eme de matrices $4\times4$ $\left(\gamma^\mu\right)_{0\leq\mu\leq3}$ satisfaisant la relation $(\ref{eq3prime})$, c'est-\`a-dire $\gamma^{\mu}\gamma^{\nu}+\gamma^{\nu}\gamma^{\mu}=2g^{\mu\nu}I_4$, est appel\'e une repr\'esentation de l'\'equation de Dirac-Sidharth. Si de plus, $\left(\gamma^\mu\right)_{0\leq\mu\leq3}$ est un syst\`eme de matrices unitaires  , on dit qu'il est une repr\'esentation unitaire de l'\'equation de Dirac-Sidharth.
\end{definition}
Les th\'eor\`emes suivants concernent les relations entre diff\'erentes repr\'esentations \cite{messiah58}.
\begin{theorem}Th\'eor\`eme Fondamental de Pauli.\\
Pour deux repr\'esentations de l'\'equation de Dirac-Sidharth $\left(\gamma^\mu\right)_{0\leq\mu\leq3}$, $\left(\gamma^{'\mu}\right)_{0\leq\mu\leq3}$, il existe une matrice $S$, d\'efinie \`a une constante multiplicative pr\`es et non singulier $\left(det(S)\neq0\right)$, telle que 
\begin{equation*}
\gamma^\mu=S\gamma^{'\mu} S^{-1}, \mu=0, 1, 2, 3.
\end{equation*} 
\end{theorem}
\begin{corollary}
Pour deux repr\'esentations unitaires de l'\'equation de Dirac-Sidharth $\left(\gamma^\mu\right)_{0\leq\mu\leq3}$, $\left(\gamma^{'\mu}\right)_{0\leq\mu\leq3}$, il existe une matrice unitaire $U$, d\'efinie \`a une phase pr\`es  
\begin{equation*}
\gamma^\mu=U\gamma^{'\mu} U^{-1}, \mu=0, 1, 2, 3.
\end{equation*} 
\end{corollary}

Dans la section \ref{DDSE} nous avons vu comment les matrices gamma dans la representation de Dirac peuvent \^etre exprim\'ees \`a l'aide des matrices de Pauli. De mani\`ere analogue nous pouvons obtenir, en construisant l'\'equation de Dirac-Sidharth six representations unitaires de cette equation, \`a savoir\\

\noindent $\left(\sigma^3\otimes\sigma^0, i\sigma^2\otimes\sigma^1, i\sigma^2\otimes\sigma^2, i\sigma^2\otimes\sigma^3\right)$ Repr\'esentation de Dirac\\

\noindent $\left(\sigma^3\otimes\sigma^0, i\sigma^1\otimes\sigma^1, i\sigma^1\otimes\sigma^2, i\sigma^1\otimes\sigma^3\right)$\\
 
\noindent $\left(\sigma^2\otimes\sigma^0, i\sigma^1\otimes\sigma^1, i\sigma^1\otimes\sigma^2, i\sigma^1\otimes\sigma^3\right)$\\

\noindent $\left(\sigma^2\otimes\sigma^0, i\sigma^3\otimes\sigma^1, i\sigma^3\otimes\sigma^2, i\sigma^3\otimes\sigma^3\right)$\\

\noindent $\left(\sigma^1\otimes\sigma^0, i\sigma^2\otimes\sigma^1, i\sigma^2\otimes\sigma^2, i\sigma^2\otimes\sigma^3\right)$ Repr\'esentation de Weyl\\

\noindent $\left(\sigma^1\otimes\sigma^0, i\sigma^3\otimes\sigma^1, i\sigma^3\otimes\sigma^2, i\sigma^3\otimes\sigma^3\right)$\\

\noindent et six autres repr\'esentations unitaires de l'\'equation de Dirac-Sidharth obtenues en commutant les produits tensoriels de matrices ci-dessus\\

\noindent $\left(\sigma^0\otimes\sigma^3, i\sigma^1\otimes\sigma^2, i\sigma^2\otimes\sigma^2, i\sigma^3\otimes\sigma^2\right)$ \\

\noindent $\left(\sigma^0\otimes\sigma^3, i\sigma^1\otimes\sigma^1, i\sigma^2\otimes\sigma^1, i\sigma^3\otimes\sigma^1\right)$\\
 
\noindent $\left(\sigma^0\otimes\sigma^2, i\sigma^1\otimes\sigma^1, i\sigma^2\otimes\sigma^1, i\sigma^3\otimes\sigma^1\right)$\\

\noindent $\left(\sigma^0\otimes\sigma^2, i\sigma^1\otimes\sigma^3, i\sigma^2\otimes\sigma^3, i\sigma^3\otimes\sigma^3\right)$\\

\noindent $\left(\sigma^0\otimes\sigma^1, i\sigma^1\otimes\sigma^2, i\sigma^2\otimes\sigma^2, i\sigma^3\otimes\sigma^2\right)$ \\

\noindent $\left(\sigma^0\otimes\sigma^1, i\sigma^1\otimes\sigma^3, i\sigma^2\otimes\sigma^3, i\sigma^3\otimes\sigma^3\right)$\\.

La matrice unitaire du corollaire qui transforme la premi\`ere famille de six repr\'esentations unitaires de l'\'equation de Dirac-Sidharth \`a la seconde famille est la matrice unitaire  
\begin{equation*}
\textbf{U}_{2\otimes2}=\left(%
\begin{array}{cccc}
  1 & 0 & 0 & 0 \\
  0 & 0 & 1 & 0 \\
  0 & 1 & 0 & 0 \\
  0 & 0 & 0 & 1 \\
\end{array}%
\right) 
\end{equation*}
puisqu'elle commute le produit tensoriel de deux matrices $2\times2$ quelconques, $\textbf{A}$,
$\textbf{B} \in\mathbb{C}^{2\times 2}$ de la mani\`ere suivante
\begin{equation*}
\textbf{U}_{2\otimes2}\cdot\left(\textbf{A}\otimes\textbf{B}\right)\cdot\textbf{U}_{2\otimes2} = \textbf{B}\otimes\textbf{A}
\end{equation*}
 Nous l'appelons MCT $2\otimes2$. Elle commute aussi le produit tensoriel de deux matrices complexes unicolonnes de la mani\`ere suivante. Pour $\textbf{a}=\begin{pmatrix}
  a^1 \\
  a^2 \\
\end{pmatrix}%
$, $\textbf{b}=\begin{pmatrix}
  b^1 \\
  b^2 \\
\end{pmatrix}\in\mathbb{C}^{2\times 1}$
quelconques 
\begin{equation*}
\textbf{U}_{2\otimes2}.\left(\textbf{a}\otimes \textbf{b}\right)=\textbf{b}\otimes \textbf{a}
\end{equation*}

La MCT $2\otimes2$ est fr\'equemment trouv\'ee en th\'eorie quantique de l'information 
\cite{Fujii01}, \cite{Faddev95}, \cite{Verstraete02}, o\`u on \'ecrit \`a l'aide des matrices de Pauli \cite{Fujii01},\cite{Faddev95}
\begin{equation}\label{eq5prime}
 \textbf{U}_{2\otimes2}= \frac{1}{2} I_2 \otimes I_{2} + \frac{1}{2} \sum_{i
   = 1}^3 \sigma^i \otimes \sigma^i 
\end{equation}
La MCT $3\otimes3$ a \'et\'e \'ecrite par KAZUYUKI FUJII \cite{Fujii01} de la mani\`ere suivante\\

\begin{equation}\textbf{U}_{3\otimes3}=
\begin{pmatrix}\begin{pmatrix}
  1 & 0 & 0 \\
  0 & 0 & 0 \\
  0 & 0 & 0 \\
\end{pmatrix} & 
\begin{pmatrix}
  0 & 0 & 0 \\
  1 & 0 & 0 \\
  0 & 0 & 0 \\
\end{pmatrix}
 & \begin{pmatrix}
  0 & 0 & 0 \\
  0 & 0 & 0 \\
  1 & 0 & 0 \\
\end{pmatrix}\\
  \begin{pmatrix}
  0 & 1 & 0 \\
  0 & 0 & 0 \\
  0 & 0 & 0 \\
\end{pmatrix} & \begin{pmatrix}
  0 & 0 & 0 \\
  0 & 1 & 0 \\
  0 & 0 & 0 \\
\end{pmatrix} & 
\begin{pmatrix}
  0 & 0 & 0 \\
  0 & 0 & 0 \\
  0 & 1 & 0 \\
\end{pmatrix} \\
  \begin{pmatrix}
  0 & 0 & 1 \\
  0 & 0 & 0 \\
  0 & 0 & 0 \\
\end{pmatrix} & \begin{pmatrix}
  0 & 0 & 0 \\
  0 & 0 & 1 \\
  0 & 0& 0 \\
\end{pmatrix} & \begin{pmatrix}
  0 & 0 & 0\\
  0 & 0 & 0 \\
  0 & 0 & 1 \\
\end{pmatrix} \\
\end{pmatrix}
\end{equation}
afin d'obtenir une conjecture pour la forme de la MCT $n\otimes n$, pour tout $n \in \mathbb{N}^\star$.\\

D\'signons par $\textbf{U}_{n\otimes p}$ la MCT $n\otimes p$, $n$, $p\in\mathbb{N}$.
\chapter{MATRICES DE PERMUTATION TENSORIELLE}\label{chap3}
Ce chapitre est bas\'e sur Refs. \cite{Rakotonirina05,Rakotonirina07,Rakotonirina09,Rakotonirina11}

\section{MATRICES DE COMMUTATION TENSORIELLE}
\subsection{Matrices de commutation tensorielle}
\begin{definition}
Pour $p$, $q\in\mathbb{N}$, $p\geq 2$, $q\geq 2$, nous appelons MCT $p\otimes q$ la matrice de permutation  $\textbf{U}_{p\otimes
q}\in \mathbb{C}^{pq\times pq}$, v\'erifiant la propri\'et\'e suivante
\begin{equation}\label{eq6}
\textbf{U}_{p\otimes q}.(\textbf{a}\otimes \textbf{b}) = \textbf{b}\otimes \textbf{a}
\end{equation}
 pour tous $\textbf{a}\in
\mathbb{C}^{p\times 1}$, $\textbf{b}\in
\mathbb{C}^{q\times 1}$.
\end{definition}
 En considerant $\textbf{U}_{p\otimes q}$ comme une matrice d'un tenseur d'ordre deux (Cf. Annexe \ref{appA}), nous pouvons la construire en utilisant la r\`egle suivante.

\begin{problem}\label{rl}
 Commen\c cons par mettre 1 sur la premi\`ere ligne et premi\`ere colonne, puis passons \`a la deuxi\`eme colonne en descendant $p$ lignes suivant cette colonne pla\c cons 1 \`a cette place qui est la $p+1$-i\`eme ligne et deuxi\`eme colonne. Puis passons \`a la troisi\`eme colonne en descendant $p$ lignes suivant cette colonne pla\c cons 1 \`a cette place qui est la $2p+2$-i\`eme ligne et troisi\`eme colonne, et ainsi de suite jusqu'\`a nous n'avons que $p-1$ lignes pour descendre (alors nous avons comme nombre de 1: $q$). Puis passons \`a la colonne suivante qui est la $(q + 1)$-i\`eme colonne, mettre 1 sur la deuxi\`eme ligne de cette colonne (puisque $(p-1)+1=p$) et r\'ep\'etons le processus jusqu'\`a nous n'avons que $p-2$
lignes pour descendre (alors nous avons comme nombre de 1: $2q$). Apr\`es, passons \`a la colonne suivante qui est la $(2q + 1)$ -i\`eme colonne, mettre 1 sur la troisi\`eme ligne de cette colonne (puisque $(p-2)+2=p$) et r\'ep\'etons le processus jusqu'\`a nous n'avons que $p-3$ lignes pour descendre (alors nous avons comme nombre de 1: $3q$). En Continuant de cette mani\`ere nous aurons que l'\'el\'ement sur $p\times q$-i\`eme ligne et $p\times q$-i\`eme colonne est 1.  Les autres \'el\'ements sont 0.
\end{problem}

\begin{proposition}\label{prop23}
  Pour $n, p \in \mathbb{N}$, $n, p \geqslant 2$,
  \[ \textbf{U}_{n \otimes p} = \sum_{( i, j )}^{( p, n )} \textbf{E}_{p \times n}^{( i, j )}
     \otimes \textbf{E}_{p \times n}^{( i, j )^t} = \sum_{( i, j )}^{( p, n )} \textbf{E}_{p
     \times n}^{( i, j )} \otimes \textbf{E}_{n \times p}^{( j, i )} \]
  o\`u $\textbf{E}_{p \times n}^{( i, j )}$ est la matrice $p \times n$ \'el\'ementaire 
  form\'ee par des zeros sauf l'\'el\'ement sur la $i$-i\`eme ligne et $j$-i\`eme colonne qui est \'egal \`a $1$.
  \end{proposition}

\begin{proof}
  Soient $\textbf{a} =  \begin{pmatrix}
    a_1\\
    a_2\\
    \vdots\\
    a_n
  \end{pmatrix} \in \mathbb{C}^{n \times 1}$, $\textbf{b} =
  \begin{pmatrix}
    b_1\\
    b_2\\
    \vdots\\
    b_p
  \end{pmatrix} \in \mathbb{C}^{p \times 1}$
\begin{equation*}
\begin{split}
   \textbf{U}_{n \otimes p} \cdot(\textbf{a} \otimes \textbf{b}) &= \sum_{( i, j )}^{( p, n )}
     \textbf{E}_{p \times n}^{( i, j )} \otimes \textbf{E}_{n \times p}^{( j, i )} \cdot ( \textbf{a}
     \otimes \textbf{b} ) = \sum_{( i, j )}^{( p, n )} ( \textbf{E}_{p \times n}^{( i, j )}
     \cdot \textbf{a} ) \otimes ( \textbf{E}_{n \times p}^{( j, i )} \cdot \textbf{b} )\\
     &= \sum_{( i, j )}^{( p, n )} ( \delta_{ik} a_j )_{1 \leqslant k \leqslant
     p} \otimes ( \delta_{jl} b_i )_{1 \leqslant l \leqslant n}
     = \sum_{( i, j )}^{( p, n )} ( \delta_{ik} b_i )_{1 \leqslant k \leqslant
     p} \otimes ( \delta_{jl} a_j )_{1 \leqslant l \leqslant n}\\
     &= \sum_{( i, j )}^{( p, n )} \begin{pmatrix}
       0\\
       \vdots\\
       0\\
       b_i\\
       0\\
       \vdots\\
       0
     \end{pmatrix} \otimes \begin{pmatrix}
       0\\
       \vdots\\
       0\\
       a_j\\
       0\\
       \vdots\\
       0
     \end{pmatrix} 
     = \textbf{b} \otimes \textbf{a}
\end{split}
\end{equation*}\\
\noindent o\`u $\delta_{ij}$ est le symbole de kronecker.
 \end{proof} 
  \begin{example}
    L'application de la r\`egle \ref{rl} nous donne 

\begin{equation*}
       \textbf{U}_{2 \otimes 3} = \begin{pmatrix}
         1 & 0 & 0 & 0 & 0 & 0  \\
         0 & 0 & 0 & 1 & 0 & 0  \\
         0 & 1 & 0 & 0 & 0 & 0  \\
         0 & 0 & 0 & 0 & 1 & 0  \\
         0 & 0 & 1 & 0 & 0 & 0  \\
         0 & 0 & 0 & 0 & 0 & 1 
       \end{pmatrix} 
       \end{equation*}
       \begin{equation*}
       \begin{split}
    U_{2 \otimes 3} &= \begin{pmatrix}
      1 & 0\\
      0 & 0\\
      0 & 0
    \end{pmatrix} \otimes \begin{pmatrix}
      1 & 0 & 0\\
      0 & 0 & 0
    \end{pmatrix} + \begin{pmatrix}
      0 & 1\\
      0 & 0\\
      0 & 0
    \end{pmatrix} \otimes \begin{pmatrix}
      0 & 0 & 0\\
      1 & 0 & 0
    \end{pmatrix} + \begin{pmatrix}
      0 & 0\\
      1 & 0\\
      0 & 0
    \end{pmatrix} \otimes \begin{pmatrix}
      0 & 1 & 0\\
      0 & 0 & 0
    \end{pmatrix}\\
     &+ \begin{pmatrix}
      0 & 0\\
      0 & 1\\
      0 & 0
    \end{pmatrix} \otimes \begin{pmatrix}
      0 & 0 & 0\\
      0 & 1 & 0
    \end{pmatrix} + \begin{pmatrix}
      0 & 0\\
      0 & 0\\
      1 & 0
    \end{pmatrix} \otimes \begin{pmatrix}
      0 & 0 & 1\\
      0 & 0 & 0
    \end{pmatrix} + \begin{pmatrix}
      0 & 0\\
      0 & 0\\
      0 & 1
    \end{pmatrix} \otimes \begin{pmatrix}
      0 & 0 & 0\\
      0 & 0 & 1
    \end{pmatrix}
    \end{split}
\end{equation*}
  \end{example}
  \begin{remark}\label{rmk}
Consid\'erons la fonction $L$ de l'ensemble de toutes les matrices de dimension finies vers l'ensemble des matrices unicolonnes . Pour une matrice $\textbf{X}=\begin{pmatrix}
  x_{11} & x_{12} & \ldots & x_{1p} \\
  x_{21} & x_{22} & \ldots & x_{2p}\\
  \ldots & \ldots & \ldots & \ldots  \\
  x_{n1} &  x_{n2} & \ldots & x_{np} \\
   \end{pmatrix}$,\\
   $L\left(\textbf{X}\right)=\begin{pmatrix}x_{11}\\
   x_{12}\\
   \vdots\\
   x_{1p}\\
   x_{21}\\
   x_{22}\\
   \vdots\\
   x_{2p}\\
   \vdots\\
   x_{n1}\\
   x_{n2}\\
   \vdots\\
   x_{np}\\
   \end{pmatrix}$.
    On peut obtenir facilement que $\textbf{U}_{n\otimes p}\cdot L\left(\textbf{X}\right)=L\left(\textbf{X}^T\right)$.
   
\end{remark}
\subsection{Expression d'un \'el\'ement d'une matrice de commutation tensorielle}
 Ici $n$ et $p$ sont des \'el\'ements quelconques de $\mathbb{N^{\star}}$, $n, p\geq 2$. Ainsi, il s'agit d'une g\'en\'eralisation de l' expression d'un \'el\'ement d'une MCT $n\otimes n$ pour un $n\;\in\;\mathbb{N^{\star}}$ \cite{Fujii01}. D'abord, \'etudions
 l'exemple ci-dessous pour conjecturer l'expression pour le cas plus g\'en\'eral. Ainsi, nous suivons la m\'ethode de \cite{Fujii01}. Ecrivons alors 
 $\textbf{U}_{3\otimes 5}$ de la fa\c con suivante:\\

$\textbf{U}_{3\otimes 5}=
                        \begin{pmatrix}
                           \begin{pmatrix}
                               1 & 0 & 0 & 0 & 0 \\
                               0 & 0 & 0 & 0 & 0 \\
                               0 & 0 & 0 & 0 & 0 \\
                             \end{pmatrix}
                           & \begin{pmatrix}
                               0 & 0 & 0 & 0 & 0 \\
                               1 & 0 & 0 & 0 & 0 \\
                               0 & 0 & 0 & 0 & 0 \\
                             \end{pmatrix} & \begin{pmatrix}
                               0 & 0 & 0 & 0 & 0 \\
                               0 & 0 & 0 & 0 & 0 \\
                               1 & 0 & 0 & 0 & 0 \\
                             \end{pmatrix} \\
                          \begin{pmatrix}
                               0 & 1 & 0 & 0 & 0 \\
                               0 & 0 & 0 & 0 & 0 \\
                               0 & 0 & 0 & 0 & 0 \\
                             \end{pmatrix} & \begin{pmatrix}
                               0 & 0 & 0 & 0 & 0 \\
                               0 & 1 & 0 & 0 & 0 \\
                               0 & 0 & 0 & 0 & 0 \\
                             \end{pmatrix} & \begin{pmatrix}
                               0 & 0 & 0 & 0 & 0 \\
                               0 & 0 & 0 & 0 & 0 \\
                               0 & 1 & 0 & 0 & 0 \\
                             \end{pmatrix} \\
                          \begin{pmatrix}
                               0 & 0 & 1 & 0 & 0 \\
                               0 & 0 & 0 & 0 & 0 \\
                               0 & 0 & 0 & 0 & 0 \\
                             \end{pmatrix} & \begin{pmatrix}
                               0 & 0 & 0 & 0 & 0 \\
                               0 & 0 & 1 & 0 & 0 \\
                               0 & 0 & 0 & 0 & 0 \\
                             \end{pmatrix} & \begin{pmatrix}
                               0 & 0 & 0 & 0 & 0 \\
                               0 & 0 & 0 & 0 & 0 \\
                               0 & 0 & 1 & 0 & 0 \\
                             \end{pmatrix} \\
                          \begin{pmatrix}
                               0 & 0 & 0 & 1 & 0 \\
                               0 & 0 & 0 & 0 & 0 \\
                               0 & 0 & 0 & 0 & 0 \\
                             \end{pmatrix} & \begin{pmatrix}
                               0 & 0 & 0 & 0 & 0 \\
                               0 & 0 & 0 & 1 & 0 \\
                               0 & 0 & 0 & 0 & 0 \\
                             \end{pmatrix} & \begin{pmatrix}
                               0 & 0 & 0 & 0 & 0 \\
                               0 & 0 & 0 & 0 & 0 \\
                               0 & 0 & 0 & 1 & 0 \\
                             \end{pmatrix} \\
                          \begin{pmatrix}
                               0 & 0 & 0 & 0 & 1 \\
                               0 & 0 & 0 & 0 & 0 \\
                               0 & 0 & 0 & 0 & 0 \\
                             \end{pmatrix} & \begin{pmatrix}
                               0 & 0 & 0 & 0 & 0 \\
                               0 & 0 & 0 & 0 & 1 \\
                               0 & 0 & 0 & 0 & 0 \\
                             \end{pmatrix} & \begin{pmatrix}
                               0 & 0 & 0 & 0 & 0 \\
                               0 & 0 & 0 & 0 & 0 \\
                               0 & 0 & 0 & 0 & 1 \\
                             \end{pmatrix} \\
                        \end{pmatrix}
                      $\\

Consid\'erons les matrices rectangles $I_{n\times
p}=\left(\delta_{j}^{i}\right)_{1\leq i\leq n, 1\leq j
\leq p}$, $I_{p\times
n}=\left(\delta_{j}^{i}\right)_{1\leq i\leq p, 1\leq j
\leq n}$, o\`u $\delta_{j}^{i}$ est le symbole de Kronecker. La matrice $np\times np$
\begin{center}
$I_{p\times n}\otimes I_{n\times
p}=\left(\delta_{j_{1}j_{2}}^{i_{1}i_{2}}\right)\;
=\left(\delta_{j_{1}}^{i_{1}}\delta_{j_{2}}^{i_{2}}\right)$
\end{center}
o\`u,\\
$i_{1}i_{2}=11, 12, \ldots, 1n, 21, 22, \ldots, 2n, \ldots, p1,
 p2, \ldots, pn$\\
 indices de lignes,\\
$j_{1}j_{2}=11, 12, \ldots, 1p, 21, 22, \ldots, 2p, \ldots, n1,
 n2, \ldots, np$\\
 indices de colonnes, \\
 nous sugg\`ere la proposition suivante.

\begin{proposition}
\begin{equation}\label{e61}
\textbf{U}_{n\otimes
p}=\left(U_{j_{1}j_{2}}^{i_{1}i_{2}}\right)=\left(\delta_{j_{2}j_{1}}^{i_{1}i_{2}}\right)
=\left(\delta_{j_{2}}^{i_{1}}\delta_{j_{1}}^{i_{2}}\right)
\end{equation}
o\`u,\\
$i_{1}i_{2}=11, 12, \ldots, 1n, 21, 22, \ldots, 2n, \ldots, p1,
 p2, \ldots, pn$\\
 indices de lignes,\\
$j_{1}j_{2}=11, 12, \ldots, 1p, 21, 22, \ldots, 2p, \ldots, n1,
 n2, \ldots, np$\\
 indices de colonnes, \\
\end{proposition}

\begin{proof} Soient $\textbf{a}=\left(a^{j_{1}}\right)_{1\leq j_{1}\leq
n}\in\mathbb{C}^{n\times1}$,
$\textbf{b}=\left(b^{j_{2}}\right)_{1\leq j_{2}\leq p}\in\mathbb{C}^{p\times1}$\\
\begin{multline*}
\left(\textbf{a}\otimes \textbf{b}\right)^{i_{1}i_{2}}\longrightarrow
\left(\textbf{U}_{n\otimes p}\cdot\left(\textbf{a}\otimes
\textbf{b}\right)\right)^{i_{1}i_{2}}=\delta_{j_{2}}^{i_{1}}\delta_{j_{1}}^{i_{2}}a^{j_{1}}b^{j_{2}}
=\delta_{j_{2}}^{i_{1}}b^{j_{2}}\delta_{j_{1}}^{i_{2}}a^{j_{1}}
=b^{i_{1}}a^{i_{2}}\\
 =\left(\textbf{b}\otimes
\textbf{a}\right)^{i_{1}i_{2}}
\end{multline*}
 \end{proof}
\section{MATRICES DE PERMUTATION TENSORIELLE}

Dans cette section nous suivons l'id\'ee dans \cite{Raoelina86}, en alg\`ebre lin\'eaire et multilin\'eaire, en \'etablissant d'abord les th\'eor\`emes pour les op\'erateurs, de fa\c con intrins\`eque, c'est-\`a-dire ind\'ependamment des bases, puis les th\'eor\`emes correspondant pour les matrices. Ainsi, nous allons d'abord parler d'op\'erateurs de permutation tensorielle (OPT). Nous utiliserons aussi ses notations pour les vecteurs et covecteurs, en surlignant les vecteurs, $\overline{x}$, et en soulignant les covecteurs, $\underline{\varphi}$.

\subsection{Op\'erateurs de permutation tensorielle}
\begin{definition}
Consid\'erons les  $\mathbb{C}$-espaces vectoriels de dimensions finies $\mathcal{E}_1,
\mathcal{E}_2, \ldots, \mathcal{E}_k$ et une permutation $\sigma$ sur
$\left\{1, 2,\ldots , k\right\}$. L'op\'erateur lin\'eaire $U_{\sigma}$ 
   de $\mathcal{E}_1\otimes
\mathcal{E}_2\otimes\ldots \otimes \mathcal{E}_k$ \`a
$\mathcal{E}_{\sigma(1)}\otimes
\mathcal{E}_{\sigma(2)}\otimes\ldots\otimes
\mathcal{E}_{\sigma(k)}$, $U_{\sigma}\in
\mathcal{L}(\mathcal{E}_1\otimes
\mathcal{E}_2\otimes\ldots\otimes \mathcal{E}_k,
\mathcal{E}_{\sigma(1)}\otimes
\mathcal{E}_{\sigma(2)}\otimes\ldots\otimes
\mathcal{E}_{\sigma(k)} )$, d\'efini par
\begin{center}
$U_{\sigma}(\overline{x_1}\otimes
\ldots\otimes\overline{x_k})=
\overline{x_{\sigma(1)}}\otimes\ldots
\otimes\overline{x_{\sigma(k)}}$
\end{center}
pour tous $\overline{x_1}$\;$\in$\;$\mathcal{E}_1$,
$\overline{x_2}$\;$\in$\;$\mathcal{E}_2$, \ldots,
$\overline{x_k}$\;$\in$\;$\mathcal{E}_k$ est appel\'e un $\sigma$-OPT.\\
 Si $n=2$, alors nous disons que $U_{\sigma}$  est un op\'erateur de commutation tensorielle.
\end{definition}
\begin{proposition}
Si $U_{\sigma}$ est un $\sigma$-OPT, alors son transpos\'e $U_{\sigma}^T$ est un ${\sigma}^{-1}$-OPT,  $U_{\sigma^{-1}}=U_{\sigma}^{-1}$.\\
\end{proposition}

\begin{proof} Consid\'erons les $\mathbb{C}$-espaces vectoriels
$\mathcal{E}_1$, $\mathcal{E}_2$, \ldots, $\mathcal{E}_k$ de dimensions finies et ${\mathcal{E}_1}^{\star},
{\mathcal{E}_2}^{\star},\ldots,
{\mathcal{E}_k}^{\star}$ sont leurs espaces duaux. $U_{\sigma}\in
\mathcal{L}(\mathcal{E}_1\;\otimes\;
\mathcal{E}_2\;\otimes\;\ldots\;\otimes \mathcal{E}_k,
\mathcal{E}_{\sigma(1)}\;\otimes\;
\mathcal{E}_{\sigma(2)}\;\otimes\;\ldots\;\otimes\;
\mathcal{E}_{\sigma(k)})$ $\sigma$-OPT.
Alors ${U_{\sigma}}^T\;\in\; \mathcal{L}(
{\mathcal{E}_{\sigma(1)}}^{\star}\;\otimes\;
{\mathcal{E}_{\sigma(2)}}^{\star}\;\otimes\;\ldots\;\otimes\;
{\mathcal{E}_{\sigma(k)}}^{\star},
{\mathcal{E}_1}^{\star}\;\otimes\;
{\mathcal{E}_2}^{\star}\;\otimes\;\ldots\;\otimes
{\mathcal{E}_k}^{\star})$.\\
Soient
$\underline{\varphi^{\sigma(1)}}$\;$\in\;$${\mathcal{E}_{\sigma(1)}}^{\star}$,
 $\underline{\varphi^{\sigma(2)}}$\;$\in\;$${\mathcal{E}_{\sigma(2)}}^{\star}$,
 \ldots$\underline{\varphi^{\sigma(k)}}$\;$\in\;$${\mathcal{E}_{\sigma(k)}}^{\star}$,
 $\overline{x_1}$\;$\in$\;$\mathcal{E}_1$,
$\overline{x_2}$\;$\in$\;$\mathcal{E}_2$, \ldots,
$\overline{x_k}$\;$\in$\;$\mathcal{E}_k$.\\
${U_{\sigma}}^T\cdot\left(
\underline{\varphi^{\sigma(1)}}\;\otimes\;\underline{\varphi^{\sigma(2)}}
\;\otimes\ldots\otimes\;\underline{\varphi^{\sigma(k)}}
\right)$($\overline{x_1}$\;$\otimes$\;$\overline{x_2}$\;
$\otimes$\;\ldots\;$\otimes$\;$\overline{x_k}$)\;\\

=\;$\underline{\varphi^{\sigma(1)}}\;\otimes\;\underline{\varphi^{\sigma(2)}}
\;\otimes\ldots\otimes\;\underline{\varphi^{\sigma(k)}}$\;[$U_{\sigma}$($\overline{x_1}$\;$\otimes$\;$\overline{x_2}$\;
$\otimes$\;\ldots\;$\otimes$\;$\overline{x_k}$)]\\
(par d\'efinition de l'op\'erateur transpos\'e \cite{Raoelina86})\\
=\;$\underline{\varphi^{\sigma(1)}}\;\otimes\;\underline{\varphi^{\sigma(2)}}
\;\otimes\ldots\otimes\;\underline{\varphi^{\sigma(k)}}$($\overline{x_{\sigma(1)}}$\;$\otimes$\;$\overline{x_{\sigma(2)}}$\;$\otimes$\;\ldots
\;$\otimes$\;$\overline{x_{\sigma(k)}}$)\\

=\;$\underline{\varphi^{\sigma(1)}}(\overline{x_{\sigma(1)}})\otimes\underline{\varphi^{\sigma(2)}}(\overline{x_{\sigma(2)}})
\otimes\ldots\otimes\underline{\varphi^{\sigma(k)}}(\overline{x_{\sigma(k)}})$
=\;$\underline{\varphi^{\sigma(1)}}(\overline{x_{\sigma(1)}})\underline{\varphi^{\sigma(2)}}(\overline{x_{\sigma(2)}})
\ldots\underline{\varphi^{\sigma(k)}}(\overline{x_{\sigma(k)}})$\\

(puisque les $\underline{\varphi^{\sigma(i)}}(\overline{x_{\sigma(i)}})$
sont des \'el\'ements de $\mathbb{C}$)\\

=\;$\underline\varphi^1(\overline{x_1})\otimes
\underline\varphi^2(\overline{x_2})
\otimes \ldots\otimes \underline\varphi^k(\overline{x_k})$\\

=\;$(\underline\varphi^1\otimes \underline\varphi^2\otimes
\ldots\otimes
\underline\varphi^k)(\overline{x_1}\;\otimes\;\overline{x_2}\;
\otimes\;\ldots\;\otimes\;\overline{x_k})$\\
Nous avons \\
${U_{\sigma}}^T\left(
\underline{\varphi^{\sigma(1)}}\;\otimes\;\underline{\varphi^{\sigma(2)}}
\;\otimes\ldots\otimes\;\underline{\varphi^{\sigma(k)}}
\right)$($\overline{x_1}$\;$\otimes$\;$\overline{x_2}$\;
$\otimes$\;\ldots\;$\otimes$\;$\overline{x_k}$)\;\\

=\;$(\underline\varphi^1\otimes \underline\varphi^2\otimes
\ldots\otimes
\underline\varphi^k)(\overline{x_1}\;\otimes\;\overline{x_2}\;
\otimes\;\ldots\;\otimes\;\overline{x_k})$\\
pour tous $\overline{x_1}$\;$\in$\;$\mathcal{E}_1$,
$\overline{x_2}$\;$\in$\;$\mathcal{E}_2$, \ldots,
$\overline{x_k}$\;$\in$\;$\mathcal{E}_k$.\\
D'o\`u,\\
${U_{\sigma}}^T\cdot\left(
\underline{\varphi^{\sigma(1)}}\;\otimes\;\underline{\varphi^{\sigma(2)}}
\;\otimes\ldots\otimes\;\underline{\varphi^{\sigma(k)}}
\right)=\underline\varphi^1\otimes \underline\varphi^2\otimes
\ldots\otimes \underline\varphi^k$
\end{proof}
\begin{proposition}\label{tm32}
Consid\'erons les $\mathbb{C}$-espaces vectoriels de dimensions finies $\mathcal{E}_1$,
$\mathcal{E}_2$, \ldots, $\mathcal{E}_k$, $\mathcal{F}_1$,
$\mathcal{F}_2$, \ldots, $\mathcal{F}_k$, une permutation $\sigma$ sur
\{${1, 2,\ldots , k}$\} et un $\sigma$-OPT $U_{\sigma}\in \mathcal{L}\left(\mathcal{F}_1\otimes\ldots\otimes \mathcal{F}_k,
\mathcal{F}_{\sigma(1)}\otimes\ldots\otimes\mathcal{F}_{\sigma(k)}\right)$.
Alors, pour tous $\phi_1\in \mathcal{L}\left(\mathcal{E}_1,\mathcal{F}_1\right)$,
$\phi_2\in \mathcal{L}\left(\mathcal{E}_2,\mathcal{F}_2\right)$, \ldots,
$\phi_k\in \mathcal{L}\left(\mathcal{E}_k,\mathcal{F}_k\right)$
\begin{equation*}
U_{\sigma}\cdot\left(\phi_1 \otimes  \ldots \otimes \phi_k\right)
=\left(\phi_{\sigma(1)}\otimes\ldots\otimes
\phi_{\sigma(k)}\right)\cdot V_{\sigma}
\end{equation*}
o\`u $V_{\sigma}\in \mathcal{L}(\mathcal{E}_1\otimes\mathcal{E}_2\otimes\ldots\otimes \mathcal{E}_k,
\mathcal{E}_{\sigma(1)}\otimes\mathcal{E}_{\sigma(2)}\otimes\ldots\otimes\mathcal{E}_{\sigma(k)})$ est un $\sigma$-OPT.  
\end{proposition}

\begin{proof}
$\phi_1\otimes\ldots\otimes\phi_k
\in\mathcal{L}\left(\mathcal{E}_1\otimes
\ldots\otimes
\mathcal{E}_k,\mathcal{F}_1\otimes\ldots\otimes \mathcal{F}_k,
\right)$, donc $U_{\sigma}\cdot(\phi_1\otimes\phi_2\otimes\ldots\otimes\phi_k)
\in\mathcal{L}\left(\mathcal{E}_1\otimes\ldots\otimes \mathcal{E}_k,
\mathcal{F}_{\sigma(1)}\otimes
\ldots\otimes\mathcal{F}_{\sigma(k)}\right)$.
$\left(\phi_{\sigma(1)}\otimes\phi_{\sigma(2)}\otimes\ldots\otimes\phi_{\sigma(k)}\right)\cdot V_{\sigma}\\
\in\mathcal{L}\left(\mathcal{E}_1\otimes\ldots\otimes \mathcal{E}_k,
\mathcal{F}_{\sigma(1)}\otimes\ldots\otimes\mathcal{F}_{\sigma(k)}\right)$\\
Si $\overline{x_1}\in\mathcal{E}_1$, 
$\overline{x_2}\in\mathcal{E}_2$, \ldots,
$\overline{x_k}\in\mathcal{E}_k$, alors
$U_{\sigma}\cdot\left(\phi_1\otimes\ldots\otimes\phi_k\right)
\left(\overline{x_1}\otimes\ldots\otimes\overline{x_k}\right)
=U_{\sigma}\left[\phi_1\left(\overline{x_1}\right)\otimes\phi_2\left(\overline{x_2}\right)\otimes\ldots\otimes\phi_k\left(\overline{x_k}\right)\right]
=\phi_{\sigma(1)}\left(\overline{x_{\sigma(1)}}\right)\otimes\ldots\otimes\phi_{\sigma(k)}\left(\overline{x_{\sigma(k)}}\right)$\\
(puisque $U_{\sigma}$ est un OPT)\\
$=\left(\phi_{\sigma(1)}\otimes\ldots\otimes\phi_{\sigma(k)}\right)\left(\overline{x_{\sigma(1)}}\otimes\ldots
\otimes\overline{x_{\sigma(k)}}\right)=\left(\phi_{\sigma(1)}\otimes\ldots\otimes\phi_{\sigma(k)}\right)\cdot V_{\sigma}\left(\overline{x_1}\otimes\ldots\otimes\overline{x_k}\right)$.
\end{proof}
\begin{definition} Consid\'erons les $\mathbb{C}$-espaces vectoriels
$\mathcal{E}_1$, $\mathcal{E}_2$, \ldots, $\mathcal{E}_k$ de
dimensions $n_1$, $n_2$, \ldots, $n_k$ et $\sigma$-OPT
$U_{\sigma}$\;$\in\;$ $\mathcal{L}$($\mathcal{E}_1$\;$\otimes$\;
$\mathcal{E}_2$\;$\otimes$\;\ldots\;$\otimes$ $\mathcal{E}_k$,
$\mathcal{E}_{\sigma(1)}$\;$\otimes$\;
$\mathcal{E}_{\sigma(2)}$\;$\otimes$\;\ldots\;$\otimes$\;
$\mathcal{E}_{\sigma(k)}$).
Soit\\
$\mathcal{B}_1=\left(\overline{e_{11}},\overline{e_{12}},\ldots,\overline{e_{1n_1}}\right)$ une base de
$\mathcal{E}_1$;\\
$\mathcal{B}_2=\left(\overline{e_{21}},\overline{e_{22}},\ldots,\overline{e_{2n_2}}\right)$ une base de
$\mathcal{E}_2$;\\\ldots\\
$\mathcal{B}_k=\left(\overline{e_{k1}},\overline{e_{k2}},\ldots,\overline{e_{kn_k}}\right)$ une base de
$\mathcal{E}_k$.\\
$\textbf{U}_{\sigma}$ la matrice de $U_{\sigma}$ par rapport au couple de bases\\
\noindent $\left(\mathcal{B}_1\otimes\mathcal{B}_2\otimes \ldots \otimes \mathcal{B}_k, \mathcal{B}_{\sigma(1)}
\otimes\mathcal{B}_{\sigma(2)}\otimes \ldots\otimes\mathcal{B}_{\sigma(k)}\right)$. La matrice carr\'ee $\textbf{U}_{\sigma}$ de dimension $n_1\times n_2\times\ldots \times n_k$ est ind\'ependante des bases $\mathcal{B}_1$, $\mathcal{B}_2$,\ldots , $\mathcal{B}_k$.
Nous appelons cette matrice la $\sigma$-MPT $n_1\otimes
n_2\otimes\ldots \otimes n_k$.
\end{definition}
\begin{proposition}\label{prop7}
Soient $\textbf{U}_{\sigma}$ la $\sigma$-MPT $n_1\otimes
n_2\otimes\ldots \otimes n_k$ et $\textbf{V}_{\sigma}$ la $\sigma$-MPT $m_1\otimes
m_2\otimes\ldots \otimes m_k$. Alors, pour toutes matrices
$\textbf{A}_1$, $\textbf{A}_2$,\ldots, $\textbf{A}_k$, de dimensions respectives, $m_1\times n_1$,
$m_2\times n_2$, \ldots, $m_k\times n_k$
\begin{equation}\label{eq33}
\textbf{U}_{\sigma}\cdot\left(\textbf{A}_1\otimes 
\ldots\otimes\textbf{A}_k\right)\cdot\textbf{V}_{\sigma}^T=\textbf{A}_{\sigma(1)}\otimes
 \ldots\otimes
\textbf{A}_{\sigma(k)}
\end{equation}
\end{proposition}

\begin{proof} Soient $A_1\in\mathcal{L}\left(\mathcal{E}_1,\mathcal{F}_1\right)$,
$A_2\in\mathcal{L}\left(\mathcal{E}_2,\mathcal{F}_2\right)$, \ldots,
$A_k\in\mathcal{L}\left(\mathcal{E}_k,\mathcal{F}_k\right)$. Leurs matrices par rapport aux couples de bases $\left(\mathcal{B}_1,\mathcal{B}^{'}_1\right)$, $\left(\mathcal{B}_2,\mathcal{B}^{'}_2\right)$,\ldots , $\left(\mathcal{B}_k,\mathcal{B}^{'}_k\right)$
sont respectivement $\textbf{A}_1$, $\textbf{A}_2$,\ldots, $\textbf{A}_k$. Alors,
$\textbf{A}_1\otimes \textbf{A}_2\otimes \ldots\otimes\textbf{A}_k$ est la matrice de
$A_1\otimes A_2\otimes\ldots\;\otimes A_k$ par rapport \`a
$\left(\mathcal{B}_1\otimes\mathcal{B}_2\otimes
\ldots\otimes\mathcal{B}_k,\mathcal{B}{'}_1\otimes\mathcal{B}{'}_2\otimes
\ldots\otimes\mathcal{B}{'}_k\right)$.\\
Cependant,  $A_1\otimes A_2\otimes \ldots\otimes A_k\\
\in\mathcal{L}\left(\mathcal{E}_1\otimes
\mathcal{E}_2\otimes\ldots\otimes
\mathcal{E}_k,\mathcal{F}_1\otimes\mathcal{F}_2\otimes\ldots\otimes \mathcal{F}_k
\right)$ et
$A_{\sigma(1)}\otimes A_{\sigma(2)}\otimes \ldots\otimes
A_{\sigma(k)}\\
\in \mathcal{L}\left(\mathcal{E}_{\sigma(1)}\otimes\;\ldots\otimes
\mathcal{E}_{\sigma(k)},\mathcal{F}_{\sigma(1)}\otimes\ldots\otimes\mathcal{F}_{\sigma(k)}\right)$,\\
 donc\\
  $U_{\sigma}\cdot\left(A_1\otimes  \ldots\otimes A_k\right)$,
$\left(A_{\sigma(1)}\otimes  \ldots\otimes
A_{\sigma(k)}\right)\cdot V_{\sigma}
\in \mathcal{L}\left(\mathcal{E}_1\otimes\ldots\otimes\mathcal{E}_k,
\mathcal{F}_{\sigma(1)}\otimes\ldots\otimes\mathcal{F}_{\sigma(k)}\right)$.\\
$\textbf{A}_{\sigma(1)}\otimes \textbf{A}_{\sigma(2)}\otimes \ldots\otimes
\textbf{A}_{\sigma(k)}$ est la matrice de $A_{\sigma(1)}\otimes
A_{\sigma(2)}\otimes \ldots\otimes A_{\sigma(k)}$ par rapport \`a $\left(\mathcal{B}_{\sigma(1)}\otimes \ldots\otimes\mathcal{B}_{\sigma(k)},\mathcal{B}{'}_{\sigma(1)}\otimes \ldots\otimes\mathcal{B}{'}_{\sigma(k)}\right)$. Par suite $(\textbf{A}_{\sigma(1)}\otimes
 \ldots\otimes
\textbf{A}_{\sigma(k)})\cdot\textbf{V}_{\sigma}$ est celle de $(A_{\sigma(1)}\otimes \ldots\otimes
A_{\sigma(k)})\cdot V_{\sigma}$ par rapport \`a $\left(\mathcal{B}_1\otimes\ldots\otimes\mathcal{B}_k, \mathcal{B}{'}_{\sigma(1)}\otimes
\ldots \otimes \mathcal{B}{'}_{\sigma(k)}\right)$ .\\
 $\textbf{U}_{\sigma}\cdot\left(\textbf{A}_1\otimes \textbf{A}_2\otimes
\ldots\otimes\textbf{A}_k\right)$ est la matrice de $U_{\sigma}\cdot\left(A_1\otimes
A_2\otimes \ldots\otimes A_k\right)$ par rapport au m\^eme couple de bases.\\
D'apr\`es la proposition \ref{tm32}, nous avons \\
$\textbf{U}_{\sigma}\cdot\left(\textbf{A}_1\otimes 
\ldots \otimes\textbf{A}_k\right)
=\left(\textbf{A}_{\sigma(1)}\otimes
\ldots\otimes
\textbf{A}_{\sigma(k)}\right)\cdot\textbf{V}_{\sigma}$
\end{proof}
\begin{proposition}\label{tm41}
La matrice $\textbf{U}_{\sigma}$ est une $\sigma$-MPT $n_1\otimes
n_2\otimes\ldots \otimes n_k$ si, et seulement si, pour tous $\textbf{a}_1\in\mathbb{C}^{n_1\times 1}$,
$\textbf{a}_2\in\mathbb{C}^{n_2\times 1}$,\ldots, $\textbf{a}_k\in\mathbb{C}^{n_k\times 1}$\\
$\textbf{U}_{\sigma}\cdot(\textbf{a}_1\otimes \ldots\otimes \textbf{a}_k)$=$\textbf{a}_{\sigma(1)}\otimes
\ldots\otimes \textbf{a}_{\sigma(k)}$

\end{proposition}

\begin{proof} $"\Longrightarrow"$  C'est \'evident d'apr\`es la proposition \ref{prop7}.\\
$"\Longleftarrow"$ Supposons que pour tous\\
$\textbf{a}_1\in\mathbb{C}^{n_1\times 1}$,
$\textbf{a}_2\in\mathbb{C}^{n_2\times 1}$,\ldots,
$\textbf{a}_k\in\mathbb{C}^{n_k\times 1}$

$\textbf{U}_{\sigma}\cdot(\textbf{a}_1\otimes \ldots\otimes
\textbf{a}_k)$=$\textbf{a}_{\sigma(1)}\otimes
\ldots\otimes \textbf{a}_{\sigma(k)}$

Soient $\overline{a_1}$\;$\in$\;$\mathcal{E}_1$,
$\overline{a_2}$\;$\in$\;$\mathcal{E}_2$, \ldots,
$\overline{a_k}$\;$\in$\;$\mathcal{E}_k$ et $\mathcal{B}_1$,
$\mathcal{B}_2$,\ldots , $\mathcal{B}_k$ bases de
$\mathcal{E}_1$, $\mathcal{E}_2$, \ldots, $\mathcal{E}_k$ par rapport auxquelles les composantes de $\overline{a_1}$, $\overline{a_2}$,\ldots,
$\overline{a_k}$ forment les matrices colonnes $\textbf{a}_1$,
$\textbf{a}_2$,\ldots, $\textbf{a}_k$. Le $\sigma$-OPT 
$U_{\sigma}$\;$\in\;$ $\mathcal{L}$($\mathcal{E}_1$\;$\otimes$\;
$\mathcal{E}_2$\;$\otimes$\;\ldots\;$\otimes$ $\mathcal{E}_k$,
$\mathcal{E}_{\sigma(1)}$\;$\otimes$\;
$\mathcal{E}_{\sigma(2)}$\;$\otimes$\;\ldots\;$\otimes$\;
$\mathcal{E}_{\sigma(k)}$) dont la matrice par rapport \`a ($\mathcal{B}_1$$\otimes$$\mathcal{B}_2$$\otimes$ \ldots$\otimes$
$\mathcal{B}_k$, $\mathcal{B}_{\sigma(1)}$$\otimes$
$\mathcal{B}_{\sigma(2)}$$\otimes$ \ldots$\otimes$
$\mathcal{B}_{\sigma(k)}$ ) est $\textbf{U}_{\sigma}$. Donc\\
$U_{\sigma}$($\overline{a_1}$\;$\otimes$\;\ldots\;$\otimes$\;$\overline{a_k}$)=
$\overline{a_{\sigma(1)}}$\;$\otimes$\;\ldots
\;$\otimes$\;$\overline{a_{\sigma(k)}}$.
Ceci est vraie pour tous $\overline{a_1}$\;$\in$\;$\mathcal{E}_1$,
$\overline{a_2}$\;$\in$\;$\mathcal{E}_2$, \ldots,
$\overline{a_k}$\;$\in$\;$\mathcal{E}_k$.\\
Puisque $U_{\sigma}$ est un $\sigma$-OPT,
$\textbf{U}_{\sigma}$ est une $\sigma$-MPT $n_1\otimes
n_2\otimes\ldots \otimes n_k$ .
\end{proof}
\subsection{Expression d'un \'el\'ement d'une matrice de permutation tensorielle}

Maintenant, nous allons g\'en\'eraliser la formule \eqref{e61}. Consid\'erons les matrices $I_{n_{\sigma(r)}\times n_{r}}
  =\left(\delta_{j_{r}}^{i_{r}}\right)_{1\leq i_{r} \leq n_{\sigma(r)}, 1\leq j_{r} \leq n_{r}}$,
  $r=1, 2, \ldots,k$.

  \begin{center}
$I_{n_{\sigma(1)}\times
n_{1}}\otimes I_{n_{\sigma(2)}\times
n_{2}}\otimes\ldots\otimes I_{n_{\sigma(k)}\times
n_{k}}=\left(\delta_{j_{1}j_{2}\ldots
j_{k}}^{i_{1}i_{2}\ldots i_{k}}\right)
=\left(\delta_{j_{1}}^{i_{1}}\delta_{j_{2}}^{i_{2}}\ldots
\delta_{j_{k}}^{i_{k}}\right)$,
  \end{center}
  est une matrice carr\'ee $n_{1}n_{2}\ldots n_{k}\times n_{\sigma(1)}n_{\sigma(2)}\ldots n_{\sigma(k)}$, qui nous sugg\`ere  la proposition suivante.
\begin{proposition}
\begin{equation*}
\textbf{U}_{n_1\otimes n_2\otimes\ldots\otimes
n_k}(\sigma)=\left(U_{j_{1}j_{2}\ldots j_{k}}^{i_{1}i_{2}\ldots
i_{k}}\right)=\left(\delta_{j_{\sigma(1)}}^{i_{1}}\delta_{j_{\sigma(2)}}^{i_{2}}\ldots
\delta_{j_{\sigma(k)}}^{i_{k}}\right)
\end{equation*}
\end{proposition}

\begin{proof} Pour $\textbf{a}_{r}=\left(a_{r}^{j_{r}}\right)_{1\leq
j_{r} \leq
n_{r}}\in\mathbb{C}^{n_{r}\times1}$,
$r=1, 2, \ldots,k$,\\
\begin{multline*}
\left(\textbf{a}_{1}\otimes \textbf{a}_{2}\otimes\ldots
\otimes\textbf{a}_{k}\right)^{i_{1}i_{2}\ldots
i_{k}}\longrightarrow \left(\textbf{U}_{n_1\otimes n_2\otimes\ldots\otimes
n_k}(\sigma)\cdot\left(\textbf{a}_{1}\otimes
\textbf{a}_{2}\otimes\ldots
\otimes\textbf{a}_{k}\right)\right)^{i_{1}i_{2}\ldots
i_{k}}=\\
=\delta_{j_{\sigma(1)}}^{i_{1}}\delta_{j_{\sigma(2)}}^{i_{2}}\ldots
\delta_{j_{\sigma(k)}}^{i_{k}}a_{1}^{j_{1}}a_{2}^{j_{2}}\ldots
a_{k}^{j_{k}}\\
=\delta_{j_{\sigma(1)}}^{i_{1}}a_{\sigma(1)}^{j_{\sigma(1)}}
\delta_{j_{\sigma(2)}}^{i_{2}}a_{\sigma(2)}^{j_{\sigma(2)}}\ldots
\delta_{j_{\sigma(k)}}^{i_{k}}a_{\sigma(k)}^{j_{\sigma(k)}}\\
=a_{\sigma(1)}^{i_{1}}a_{\sigma(2)}^{i_{2}}\ldots
a_{\sigma(k)}^{i_{k}}\\
=\left(\textbf{a}_{\sigma(1)}\otimes
\textbf{a}_{\sigma(2)}\otimes\ldots
\otimes\textbf{a}_{\sigma(k)}\right)^{i_{1}i_{2}\ldots i_{k}}
\end{multline*}
 \end{proof}

\subsection{D\'ecomposition d'une matrice de permutation tensorielle}\label{soussec323}
\begin{definition}
Pour $k\in\mathbb{N}$, $k>2$ et pour une permutation $\sigma$ sur $\{1,
2,\ldots, k\}$, nous appelons $\sigma$-matrice de transposition  tensorielle $n_{1}\otimes n_{2}\otimes\ldots\otimes n_{k}$ une $\sigma$-MPT $n_{1}\otimes n_{2}\otimes\ldots\otimes n_{k}$ si $\sigma$ est une transposition.\\
\end{definition}

Consid\'erons la $\sigma$-matrice de transposition  tensorielle $n_{1}\otimes
n_{2}\otimes \ldots\otimes n_{k}$, $\textbf{U}_{n_{1}\otimes n_{2}\otimes
\ldots\otimes n_{k}}(\sigma)$ avec $\sigma$ la transposition
$(i\;\;j)$.
\begin{multline*}
\textbf{U}_{n_{1}\otimes n_{2}\otimes \ldots\otimes
n_{k}}(\sigma)\cdot\left(\textbf{a}_{1}\otimes\ldots\otimes \textbf{a}_{i}\otimes
\textbf{a}_{i+1}\otimes\ldots\otimes \textbf{a}_{j}\otimes \textbf{a}_{j+1}\otimes\ldots\otimes
\textbf{a}_{k}\right)\;
\\
=\textbf{a}_{1}\otimes\ldots\otimes \textbf{a}_{i-1}\otimes \textbf{a}_{j}\otimes
\textbf{a}_{i+1}\otimes\ldots\otimes \textbf{a}_{j-1}\otimes \textbf{a}_{i}\otimes
\textbf{a}_{j+1}\otimes\ldots\otimes \textbf{a}_{k}
\end{multline*}
pour tous
$\textbf{a}_{l}\in\mathbb{C}^{n_{l}\times1}$.\\
Si $\left(\textbf{B}_{l_{i}}\right)_{1\leq l_{i}\leq n_{j}n_{i}}$,
$\left(\textbf{B}_{l_{j}}\right)_{1\leq l_{j}\leq n_{i}n_{j}}$ sont
respectivement des bases de $\mathbb{C}^{n_{j}\times
n_{i}}$ et $\mathbb{C}^{n_{i}\times
n_{j}}$, alors la MCT $\textbf{U}_{n_{i}\otimes n_{j}}$ peut \^etre decompos\'ee comme une combinaison lin\'eaire de la base $\left(\textbf{B}_{l_{i}}\otimes \textbf{B}_{l_{j}}\right)_{1\leq
l_{i}\leq n_{j}n_{i},1\leq l_{j}\leq n_{i}n_{j}}$ de $\mathbb{C}^{n_{i}n_{j}\times n_{i}n_{j}}$.
 Nous voulons prouver que $\textbf{U}_{n_{1}\otimes n_{2}\otimes \ldots\otimes n_{k}}(\sigma)$ est une combinaison lin\'eaire de\\

$\left(I_{n_{1}n_{2}\ldots n_{i-1}}\otimes \textbf{B}_{l_{i}}\otimes
I_{n_{i+1}n_{i+2}\ldots n_{j-1}}\otimes \textbf{B}_{l_{j}}\otimes
I_{n_{j+1}n_{j+2}\ldots n_{k}}\right)_{1\leq l_{i}\leq
n_{j}n_{i},1\leq l_{j}\leq n_{i}n_{j}}$. \\
Pour ce faire, il nous suffit de prouver la proposition suivante.

\begin{proposition}\label{thm21}
Supposons $\sigma=\left(
                      \begin{array}{ccc}
                        1 & 2 & 3 \\
                        3 & 2 & 1 \\
                      \end{array}
                    \right)=\left(
                                  \begin{array}{cc}
                                    1 & 3 \\
                                  \end{array}
                                \right)$ permutation sur $\{1, 2,
                                3\}$, $\left(\textbf{B}_{i_{1}}\right)_{1\leq i_{1}\leq
                                N_{3}N_{1}}$, $\left(\textbf{B}_{i_{3}}\right)_{1\leq i_{3}\leq
                                N_{1}N_{3}}$ sont respectivement des bases de $\mathbb{C}^{N_{3}\times
N_{1}}$ et $\mathbb{C}^{N_{1}\times
N_{3}}$. Si $\textbf{U}_{N_{1}\otimes
N_{3}}=\displaystyle\sum_{i_{1}=1}^{N_{1}N_{3}}\displaystyle\sum_{i_{3}=1}^{N_{1}N_{3}}\alpha^{i_{1}i_{3}}\textbf{B}_{i_{1}}\otimes
\textbf{B}_{i_{3}}$, $\alpha^{i_{1}i_{3}}\in\mathbb{C}$, alors\\
$\textbf{U}_{N_{1}\otimes N_{2}\otimes
N_{3}}\left(\sigma\right)=\displaystyle\sum_{i_{1}=1}^{N_{1}N_{3}}\displaystyle\sum_{i_{3}=1}^{N_{1}N_{3}}\alpha^{i_{1}i_{3}}\textbf{B}_{i_{1}}\otimes
I_{N_{2}}\otimes \textbf{B}_{i_{3}}$.
\end{proposition}

\begin{proof} Soient $\textbf{b}_{1}\in\mathbb{C}^{N_{1}\times
1}$, $\textbf{b}_{3}\in\mathbb{C}^{N_{3}\times
1}$.\\
D\'eveloppons d'abord la relation
\begin{equation*}
\textbf{U}_{N_{1}\otimes N_{3}}\cdot\left(\textbf{b}_{1}\otimes
\textbf{b}_{3}\right)=\textbf{b}_{3}\otimes \textbf{b}_{1}
\end{equation*}
\begin{equation*}
\displaystyle\sum_{i_{1}=1}^{N_{1}N_{3}}\displaystyle\sum_{i_{3}=1}^{N_{1}N_{3}}\alpha^{i_{1}i_{3}}\left(\textbf{B}_{i_{1}}\otimes \textbf{B}_{i_{3}}\right).\left(\textbf{b}_{1}\otimes \textbf{b}_{3}\right)=\textbf{b}_{3}\otimes \textbf{b}_{1}
\end{equation*}

\begin{equation*}
\displaystyle\sum_{i_{1}=1}^{N_{1}N_{3}}\displaystyle\sum_{i_{3}=1}^{N_{1}N_{3}}\alpha^{i_{1}i_{3}}\left(\textbf{B}_{i_{1}}.\textbf{b}_{1}\right)\otimes \left(\textbf{B}_{i_{3}}.\textbf{b}_{3}\right)=\textbf{b}_{3}\otimes \textbf{b}_{1}
\end{equation*}

En utilisant la proposition \ref{thm7}, de l'Annexe \ref{AppB}, 
 
 \begin{equation*}
\displaystyle\sum_{i_{1}=1}^{N_{1}N_{3}}\displaystyle\sum_{i_{3}=1}^{N_{1}N_{3}}\alpha^{i_{1}i_{3}}\left(\textbf{B}_{i_{1}}.\textbf{b}_{1}\right)\otimes \textbf{b}_2\otimes\left(\textbf{B}_{i_{3}}.\textbf{b}_{3}\right)=\textbf{b}_{3}\otimes \textbf{b}_2\otimes \textbf{b}_{1}
\end{equation*}
 \begin{equation*}
\displaystyle\sum_{i_{1}=1}^{N_{1}N_{3}}\displaystyle\sum_{i_{3}=1}^{N_{1}N_{3}}\alpha^{i_{1}i_{3}}\left(\textbf{B}_{i_{1}}.\textbf{b}_{1}\right)\otimes \left(I_{N_2}.\textbf{b}_2\right)\otimes\left(\textbf{B}_{i_{3}}.\textbf{b}_{3}\right)=\textbf{b}_{3}\otimes \textbf{b}_2\otimes \textbf{b}_{1}
\end{equation*}
\begin{equation*}
\displaystyle\sum_{i_{1}=1}^{N_{1}N_{3}}\displaystyle\sum_{i_{3}=1}^{N_{1}N_{3}}\alpha^{i_{1}i_{3}}\left(\textbf{B}_{i_{1}}\otimes I_{N_2}\otimes \textbf{B}_{i_{3}}\right).\left(\textbf{b}_{1}\otimes \textbf{b}_2\otimes \textbf{b}_{3}\right)=\textbf{b}_{3}\otimes \textbf{b}_2\otimes \textbf{b}_{1}
\end{equation*}
D'o\`u,
\begin{equation*}
\textbf{U}_{N_{1}\otimes N_{2}\otimes
N_{3}}\left(\sigma\right)=\displaystyle\sum_{i_{1}=1}^{N_{1}N_{3}}\displaystyle\sum_{i_{3}=1}^{N_{1}N_{3}}\alpha^{i_{1}i_{3}}\textbf{B}_{i_{1}}\otimes I_{N_2}\otimes \textbf{B}_{i_{3}} 
\end{equation*}
\end{proof}

\begin{notation}
    Soit $\sigma\in S_{n}$, c'est-\`a-dire $\sigma$ est une permutation sur \{1, 2, \ldots, n\}, $p\in\mathbb{N}$, $p\geq2$, nous notons la MPT $\textbf{U}_{\underbrace{p\otimes p\otimes\ldots\otimes
p}_{n-times}}\left(\sigma\right)$ par $\textbf{U}_{p^{\otimes
n}}\left(\sigma\right)$, \cite{Raoelina86}.
\end{notation}

Nous utilisons les lemmes \cite{Merris97} suivants pour prouver la proposition ci-dessous.
\begin{lemma}Toute permutation $\sigma\in S_{n}$ peut s'\'ecrire comme produit de transpositions. (Cette factorisation en un produit de transpositions n'est pas unique)
\begin{equation*}
\mathcal{C_{\sigma}}=\left(i_{1}\;i_{2}\;i_{3}\;\ldots
\;i_{n-1}\;i_{n}\right)=\left(i_{1}\;i_{2}\right)\left(i_{2}\;i_{3}\right)\ldots\left(i_{n-1}\;i_{n}\right)
\end{equation*}
\end{lemma}
\begin{lemma}\label{thm33}
Soit $\sigma\in S_{n}$, dont le cycle est
\begin{equation*}
\mathcal{C_{\sigma}}=\left(i_{1}\;i_{2}\;i_{3}\;\ldots
\;i_{n-1}\;i_{n}\right)
\end{equation*}
alors
\begin{equation*}
\mathcal{C_{\sigma}}=\left(i_{1}\;i_{2}\;i_{3}\;\ldots\;i_{n-1}\right)\left(i_{n-1}\;i_{n}\right)
\end{equation*}
\end{lemma}
\begin{proposition}\label{thm34}
Pour $n\in\mathbb{N}^{*}$, $n>1$, $\sigma\in S_{n}$ dont le cycle est
\begin{equation*}
\mathcal{C_{\sigma}}=\left(i_{1}\;i_{2}\;i_{3}\;\ldots
\;i_{k-1}\;i_{k}\right)
\end{equation*}
avec $k\in\mathbb{N}^{*}$, $k>2$. Alors,
\begin{equation*}
\textbf{U}_{p^{\otimes n}}\left(\sigma\right)=\textbf{U}_{p^{\otimes
n}}\left(\left(i_{1}\;i_{2}\;\ldots\;i_{k-1}\right)\right)\cdot
\textbf{U}_{p^{\otimes n}}\left(\left(i_{k-1}\;i_{k}\right)\right)
\end{equation*}
ou
\begin{equation*}
\textbf{U}_{p^{\otimes n}}\left(\sigma\right)=\textbf{U}_{p^{\otimes
n}}\left(\left(i_{1}\;i_{2}\right)\right)\cdot \textbf{U}_{p^{\otimes
n}}\left(\left(i_{2}\;i_{3}\right)\right)\cdot\ldots \cdot
\textbf{U}_{p^{\otimes n}}\left(\left(i_{k-1}\;i_{k}\right)\right)
\end{equation*}
avec $p\in\mathbb{N}^{*}$, $p>2$.
\end{proposition}
Ainsi, une MPT peut \^etre exprim\'ee comme un produit de matrices de transposition tensorielle.
\begin{corollary}
Pour $n\in\mathbb{N}^{*}$, $n>2$, $\sigma\in S_{n}$ dont le cycle est
\begin{equation*}
\mathcal{C_{\sigma}}=\left(i_{1}\;i_{2}\;i_{3}\;\ldots
\;i_{n-1}\;n\right)
\end{equation*}
Alors,
\begin{equation*}
\textbf{U}_{p^{\otimes n}}\left(\sigma\right)=\left[\textbf{U}_{p^{\otimes
(n-1)}}\left(\left(i_{1}\;i_{2}\;\ldots\;i_{n-2}\;i_{n-1}\right)\right)\otimes
I_{p}\right]\cdot \textbf{U}_{p^{\otimes
n}}\left(\left(i_{n-1}\;n\right)\right)
\end{equation*}
\end{corollary}
\section{MATRICES DE COMMUTATION TENSORIELLE ET EQUATIONS MATRICIELLES    }
Les \'equations matricielles que nous allons voir dans cette section sont les \'equations matricielles de la forme  $\textbf{A}\textbf{X}=\textbf{B}$, $\textbf{A}\cdot\textbf{X}\cdot\textbf{B}=\textbf{C}$ et $\textbf{A}\cdot\textbf{X}+\textbf{X}\cdot\textbf{B}=\textbf{C}$. Pour la premi\`ere \'equation une \'etude de la m\'ethode de Cholesky sera donn\'ee. La deuxi\`eme et la troisi\`eme \'equations peut \^etre ramen\'ees \`a la premi\`ere, et c'est l\`a que nous introduirons les MCT.

\subsection{Sur la m\'ethode de Cholesky}

\noindent Les m\'ethodes directes pour r\'esoudre un syst\`eme lin\'eaire, m\'ethode d'\'elimination de Gauss, d\'ecomposition  LU et m\'ethode de Cholesky sont bien connues. Nous sommes d'accord avec certains auteurs \cite{DemidovitchMaron79,Nougier87} que la d\'ecomposition LU et la m\'ethode de Cholesky sont tr\`es utiles pour r\'esoudre plusieurs syst\`emes lin\'eaires dont la seule diff\'erence est les termes constants dans les seconds membres.

La m\'ethode d'\'elimination de Gauss avec ou sans choix de pivot peut nous conduire \`a la d\'ecomposition LU. La m\'ethode d'\'elimination de Gauss avec choix de pivot ne peut pas nous conduire \`a la m\'ethode de Cholesky car le choix de pivot peut  d\'etruire la sym\'eetrie. Cependant, quelques fois nous n'avons pas besoin de choisir le pivot.

Comme la m\'ethode d'\'elimination de Gauss peut nous conduire \`a la d\'ecomposition LU, nous pensons que c'est mieux de r\'esoudre d'abord l'une de ces \'equations par la m\'ethode d'\'elimination de Gauss et les autres par la d\'ecomposition LU.  Ainsi, nous serons permis de r\'esoudre d'abord  l'une de ces \'equation par la m\'ethode d'\'elimination de Gauss et les autres par la m\'ethode de Cholesky, dans le cas o\`u la matrice est sym\'etrique et d\'efinie positive, si la m\'ethode d'\'elimination de Gauss peut nous conduire \`a la d\'ecomposition de Cholesky.

Cette sous-section est arrang\'ee de la fa\c cons suivante. Dans le paragraphe ci-dessous nous allons pr\'esenter la d\'ecomposition LU en utilisant la m\'ethode d'\'elimination de Gauss. Dans le paragraphe suivant, nous prouverons avec l'aide de la d\'ecomposition LU que nous pouvons obtenir la d\'ecomposition de Cholesky \`a partir de la m\'ethode d'\'elimination de Gauss sans choix de pivot. Enfin, un exemple pour rendre plus claire la m\'ethode sera pr\'esent\'e.

\paragraph{Elimination de Gauss}

\noindent Consid\'erons le syst\`eme de $n$ \'equations \`a $n$ inconnues suivant:\\

$\left\{\begin{array}{clrrrrrr} %
a_{11}x_1+a_{12}x_2+\ldots+a_{1n}x_n & =  b_1\\
a_{21}x_1+a_{22}x_2+\ldots+a_{2n}x_n & =  b_2\\
\ldots\ldots\ldots\ldots\ldots\ldots\ldots\ldots\ldots\\
a_{i1}x_1+a_{i2}x_2+\ldots+a_{in}x_n & =  b_i\\
\ldots\ldots\ldots\ldots\ldots\ldots\ldots\ldots\ldots\\
a_{n1}x_1+a_{n2}x_2+\ldots+a_{nn}x_n & =  b_n
\end{array}\right.$\\
\noindent qui peut s'\'ecrire sous forme matricielle
\begin{equation}\label{eqn1}
\textbf{A}\textbf{X}=\textbf{B}
\end{equation}\\
o\`u\\
$\textbf{A}=\left(a_{ij}\right)_{1\leq i,j\leq n}$,
$\textbf{X}=\begin{pmatrix}
x_1\\
	x_2\\
	\vdots\\
	x_n
\end{pmatrix}$,
$\textbf{B}=\begin{pmatrix}
	b_1\\
	b_2\\
	\vdots\\
	b_n
\end{pmatrix}$
\\
\noindent Si $a_{11}\neq0$, l'\'elimination de Gauss sur la premi\`ere colonne s'\'ecrit\\

$\left\{\begin{array}{clrrrrrr} %
a_{11}^{(0)}x_1 & + & a_{12}^{(0)}x_2 & + & \ldots & + & a_{1n}^{(0)}x_n & =  b_{1}^{(0)}\\
\ & \ & a_{22}^{(1)}x_2 & + & \ldots & + & a_{2n}^{(1)}x_n & =  b_{2}^{(1)}\\
\ & \  & \ldots & \ldots & \ldots & \ldots & \ldots & \ldots  \ldots\\
\ & \ & a_{i2}^{(1)}x_2 & + & \ldots & + & a_{in}^{(1)}x_n & =  b_{i}^{(1)}\\
\ & \ & \ldots & \ldots & \ldots & \ldots & \ldots & \ldots  \ldots\\
\ & \ & a_{n2}^{(1)}x_2 & + & \ldots & + & a_{nn}^{(1)}x_n & =  b_{n}^{(1)}
\end{array}\right.$

\noindent avec $a_{1j}^{(0)}=a_{1j}$ et $a_{ij}^{(1)}=-\frac{a_{i1}a_{1j}-a_{11}a_{ij}}{a_{11}}$,\\

 \noindent qui peut s'\'ecrire, sous forme matricielle
 \begin{equation}\label{eqn2}
\textbf{A}_1\textbf{X}=\textbf{B}_1
\end{equation}\\
et peut \^etre obtenue en multipliant (\ref{eqn1}) par la matrice triangulaire inf\'erieure \cite{KincaidCheney99,Steven04} \\
\begin{equation}\nonumber
\textbf{G}_1=\left(\begin{array}{clrrrrr} %
	1 & 0 & 0  & \ldots & 0\\
	\  & \  & \   & \  & \ \\
	-\frac{a_{21}^{(0)}}{a_{11}^{(0)}} & 1 & 0 & \ldots & 0\\
	-\frac{a_{31}^{(0)}}{a_{11}^{(0)}} & 0 & 1 & \ldots & 0\\
	\vdots & \vdots & \vdots & \ddots & \vdots\\
-\frac{a_{n1}^{(0)}}{a_{11}^{(0)}} & 0 & 0 & \ldots & 1	
\end{array}\right)
\end{equation}

 $\textbf{A}_1=\textbf{G}_1\textbf{A}$ et $\textbf{B}_1=\textbf{G}_1\textbf{B}$.\\

 L'\'elimination de Gauss sur la deuxi\`eme colonne peut \^etre obtenue par multiplication \`a la relation (\ref{eqn2}) la matrice triangulaire inf\'erieure\\
\begin{equation}\nonumber
 \textbf{G}_2=\left(\begin{array}{clrrrrrr} %
	1 & 0 & 0  & 0 & \ldots & 0\\
	0 & 1 & 0  & 0 & \ldots & 0\\
	\  & \  & \   & \  & \  & \ \\
	  0 & -\frac{a_{32}^{(1)}}{a_{22}^{(1)}} & 1 & 0 & \ldots & 0\\
0 & -\frac{a_{42}^{(1)}}{a_{22}^{(1)}}& 0 & 1 & \ldots & 0\\
	\vdots & \vdots & \vdots & \vdots & \ddots & \vdots\\
0 & -\frac{a_{n2}^{(1)}}{a_{22}^{(1)}} & 0 & 0 & \ldots & 1	
\end{array}\right)
\end{equation}\\

L'\'equation matricielle (\ref{eqn2}) devient
\begin{equation}\nonumber
\textbf{A}_2\textbf{X}=\textbf{B}_2
\end{equation}\\
avec $\textbf{A}_2=\textbf{G}_2\textbf{G}_1\textbf{A}$ et $\textbf{B}_2=\textbf{G}_2\textbf{G}_1\textbf{B}$.\\

En continuant ainsi, nous avons finalement,
\begin{equation}\label{eqn3}
\textbf{U}\textbf{X}=\textbf{B}'
\end{equation}\\
avec $\textbf{U}=\textbf{G}_{n-1}\ldots \textbf{G}_2\textbf{G}_1\textbf{U}$ une matrice triangulaire sup\'erieure et  $\textbf{B}'=\textbf{G}_{n-1}\ldots \textbf{G}_2\textbf{G}_1\textbf{B}$.\\
Nous pouvons v\'erifier facilement que l'inverse de $\textbf{G}_l$ est \\
\begin{equation}\nonumber
\textbf{G}^{-1}_l=\left(\begin{array}{clrrrrrrr} %
	1 & \ & 0 & \ldots   & 0 &  \ & \ldots &  \ & 0\\
	\  & \ & \  & \   & \  & \  & \ & \ & \ \\	
	0 & \ & 1 & \ldots   & 0 & \ & \ldots & \ & 0\\
	\  & \ & \  & \   & \  & \  & \ & \ & \ \\	
		\vdots & \ & \vdots & \ddots  & \vdots & \  & \ & \ & \vdots\\
		\  & \ & \  & \   & \  & \  & \ & \ & \ \\		
	  0 & \ & 0 & \ldots & 1 & \ & \ldots & \ & 0\\
	  \  & \ & \  & \   & \  & \  & \ & \ & \ \\	
0 & \ & 0 & \ldots & \frac{a_{(l+1)l}^{(l-1)}}{a_{ll}^{(l-1)}} & \ & \ldots & \ & 0\\
	\vdots & \ & \vdots & \  & \vdots & \ & \ddots & \ & \vdots\\
0 & \ & 0 & \ldots & \frac{a_{nl}^{(l-1)}}{a_{ll}^{(l-1)}} & \ & \ldots & \ & 1\\
\end{array}\right)
\end{equation}
Donc $\textbf{G}_{n-1}\textbf{G}_{n-2}\cdots \textbf{G}_2\textbf{G}_1$ est une matrice inversible et l'\'equation (\ref{eqn3}) devient
\begin{equation}\nonumber
\textbf{L}\textbf{U}\textbf{X}=\textbf{B}
\end{equation}\\
\noindent avec $\textbf{L}=\textbf{G}^{-1}_1\textbf{G}^{-1}_2\cdots \textbf{G}^{-1}_{n-2}\textbf{G}^{-1}_{n-1}$ et nous avons la d\'ecomposition de $\textbf{A}$ comme produit de matrice triangulaire inf\'erieure par une matrice triangulaire sup\'erieure
\begin{equation}\nonumber
\textbf{A}=\textbf{L}\textbf{U}
\end{equation}\\
On peut v\'erifier facilement que
\begin{equation}\nonumber
\textbf{L}=\left(\begin{array}{clrrrrrrr} %
	1 & \ & 0 & \ldots   & 0 &  \ & \ldots &  \ & 0\\
	\  & \ & \  & \   & \  & \  & \ & \ & \ \\	
	\frac{a_{21}^{(0)}}{a_{11}^{(0)}} & \ & 1 & \ldots   & 0 & \ & \ldots & \ & 0\\
	\  & \ & \  & \   & \  & \  & \ & \ & \ \\	
		\frac{a_{31}^{(0)}}{a_{11}^{(0)}} & \ & \frac{a_{32}^{(1)}}{a_{22}^{(1)}} & \ddots  & \vdots & \  & \ & \ & \vdots\\
		\  & \ & \  & \   & \  & \  & \ & \ & \ \\		
	  \vdots & \ & \vdots & \ldots & 1 & \ & \ldots & \ & 0\\
	  \  & \ & \  & \   & \  & \  & \ & \ & \ \\	
\frac{a_{(l+1)1}^{(0)}}{a_{11}^{(0)}} & \ & \frac{a_{(l+1)2}^{(1)}}{a_{22}^{(1)}} & \ldots & \frac{a_{(l+1)l}^{(l-1)}}{a_{ll}^{(l-1)}} & \ & \ddots & \ & 0\\
	\vdots & \ & \vdots & \  & \vdots & \ & \  & \ & \vdots\\
\frac{a_{n1}^{(0)}}{a_{11}^{(0)}} & \ & \frac{a_{(n)2}^{(1)}}{a_{22}^{(1)}} & \ldots & \frac{a_{nl}^{(l-1)}}{a_{ll}^{(l-1)}} & \ & \ldots & \ & 1\\
\end{array}\right)
\end{equation}\\
\paragraph{M\'ethode de Cholesky}
Maintenant, supposons $\textbf{A}$ est une matrice sym\'etrique et d\'efinie positive. Alors, la m\'ethode de Cholesky consiste \`a d\'ecomposer $\textbf{A}$ comme le produit
\begin{equation}\nonumber
\textbf{A}=\textbf{G}^T\textbf{G}
\end{equation}\\
avec $\textbf{G}$ est une matrice triangulaire sup\'erieure et $\textbf{G}^T$ son transpos\'ee.

Laissons nous d'abord g\'en\'eraliser cette d\'ecomposition.
\begin{definition}
Soit $\textbf{A}=\left(a_{ij}\right)_{1\leq i,j\leq n}$ une matrice complexe sym\'etrique. Nous appelons matrice de Gauss de  $\textbf{A}$ la matrice triangulaire sup\'erieure $\textbf{U}(\textbf{A})$, obtenue en transformant $\textbf{A}$ par l'\'elimination de Gauss ci-dessus.
\end{definition}
\begin{proposition}\label{prop}
Soit $\textbf{A}=\left(a_{ij}\right)_{1\leq i,j\leq n}$ une matrice complexe sym\'etrique, telle que $det(\textbf{A})\neq 0$,
\begin{equation}\nonumber
\textbf{U}(\textbf{A})=\left(\begin{array}{clrrrrr} %
	u_{11} & u_{12} & \ldots & u_{1n}\\
	
	0 & u_{22} & \ldots & u_{2n}\\
	
	\vdots & \vdots & \ddots &  \vdots\\
0 & 0 & \ldots & u_{nn}	
\end{array}\right)
\end{equation}\\
la matrice de Gauss de $\textbf{A}$. Alors $\textbf{A}$ peut \^etre d\'ecompos\'ee comme le produit $\textbf{A}=\textbf{G}^T\textbf{G}$ avec
 \begin{equation}\nonumber
\textbf{G}=\left(\begin{array}{clrrrrr} %
	\frac{u_{11}}{\sqrt{u_{11}}} & \frac{u_{12}}{\sqrt{u_{11}}} & \ldots & \frac{u_{1n}}{\sqrt{u_{11}}}\\
	
	0 & \frac{u_{22}}{\sqrt{u_{22}}} & \ldots & \frac{u_{2n}}{\sqrt{u_{22}}}\\
	
	\vdots & \vdots & \ddots &  \vdots\\
0 & 0 & \ldots & \frac{u_{nn}}{\sqrt{u_{nn}}}	
\end{array}\right)
\end{equation}\\
o\`u $\sqrt{u_{ii}}$ une racine carr\'ee du nombre complexe $u_{ii}$.
\end{proposition}
\begin{proof}
\begin{equation}\nonumber
\textbf{U}(\textbf{A})=\left(\begin{array}{clrrrrr} %
	a^{(0)}_{11} & a^{(0)}_{12} & \ldots & a^{(0)}_{1n}\\
	
	0 & a^{(1)}_{22} & \ldots & a^{(1)}_{2n}\\
	
	\vdots & \vdots & \ddots &  \vdots\\
0 & 0 & \ldots & a^{(n-1)}_{nn}	
\end{array}\right)
\end{equation}\\
\begin{equation}\nonumber
\textbf{L}(\textbf{A})=\left(\begin{array}{clrrrrrrr} %
	1 & \ & 0 & \ldots   & 0 &  \ & \ldots &  \ & 0\\
	\  & \ & \  & \   & \  & \  & \ & \ & \ \\	
	\frac{a_{21}^{(0)}}{a_{11}^{(0)}} & \ & 1 & \ldots   & 0 & \ & \ldots & \ & 0\\
	\  & \ & \  & \   & \  & \  & \ & \ & \ \\	
		\frac{a_{31}^{(0)}}{a_{11}^{(0)}} & \ & \frac{a_{32}^{(1)}}{a_{22}^{(1)}} & \ddots  & \vdots & \  & \ & \ & \vdots\\
		\  & \ & \  & \   & \  & \  & \ & \ & \ \\		
	  \vdots & \ & \vdots & \ldots & 1 & \ & \ldots & \ & 0\\
	  \  & \ & \  & \   & \  & \  & \ & \ & \ \\	
\frac{a_{(l+1)1}^{(0)}}{a_{11}^{(0)}} & \ & \frac{a_{(l+1)2}^{(1)}}{a_{22}^{(1)}} & \ldots & \frac{a_{(l+1)l}^{(l-1)}}{a_{ll}^{(l-1)}} & \ & \ddots & \ & 0\\
	\vdots & \ & \vdots & \  & \vdots & \ & \  & \ & \vdots\\
\frac{a_{n1}^{(0)}}{a_{11}^{(0)}} & \ & \frac{a_{(n)2}^{(1)}}{a_{22}^{(1)}} & \ldots & \frac{a_{nl}^{(l-1)}}{a_{ll}^{(l-1)}} & \ & \ldots & \ & 1\\
\end{array}\right)
\end{equation}\\
\begin{equation}\nonumber
\textbf{A}=\textbf{L}(\textbf{A})\textbf{D}^{-1}\textbf{D}\textbf{U}(\textbf{A})
\end{equation}\\
avec
\begin{equation}\nonumber
\textbf{D}=\left(\begin{array}{clrrrrr} %
	\frac{1}{\sqrt{a^{(0)}_{11}}} & 0 & \ldots & 0\\
	
	0 & \frac{1}{\sqrt{a^{(1)}_{22}}} & \ldots & 0\\
	
	\vdots & \vdots & \ddots &  \vdots\\
0 & 0 & \ldots & \frac{1}{\sqrt{a^{(n-1)}_{nn}}}	
\end{array}\right)
\end{equation}\\
Soit
$\textbf{G}=\textbf{D}\textbf{U}(\textbf{A})$. Comme $A$ est sym\'etrique, d'o\`u $\textbf{L}(\textbf{A})\textbf{D}^{-1}=\textbf{G}^T$.

\end{proof}
\begin{example}
Consid\'erons les deux syst\`emes d'quations lin\'eaires suivants\\

\begin{equation}\nonumber
\left\{\begin{array}{clrrrrrr} %
x_1  & - & x_2 & \ & \ & + & x_4 & =  3\\
-x_1 & + & 5x_2 & + & 2x_3 & - & 3x_4 & =  -5\\
\ & \ & 2x_2 & + & 5x_2 & + & x_4 & =  -7\\
x_1 & - & 3x_2 & + & x_3 & + & 4x_4 & =  2
\end{array}\right.
\end{equation}

\begin{equation}\nonumber
\left\{\begin{array}{clrrrrrr} %
x_1  & - & x_2 & \ & \ & + & x_4 & =  3\\
-x_1 & + & 5x_2 & + & 2x_3 & - & 3x_4 & =  1\\
\ & \ & 2x_2 & + & 5x_2 & + & x_4 & =  2\\
x_1 & - & 3x_2 & + & x_3 & + & 4x_4 & =  2
\end{array}\right.
\end{equation}
dont la seule diff\'erence est les seconds membres et leur matrice est sym\'etrique.
Ainsi, r\'esolvons le premier par la m\'ethode d'\'elimination de Gauss et le second par la d\'ecomposition LU ou la m\'ethode de Cholesky.
\begin{equation}\nonumber
\left(\begin{array}{clrrrr} %
1  & 0  & 0 & 0 \\
1 & 1 & 0 & 0 \\
0 & 0 & 1 & 0\\
-1 & 0 & 0 & 1
\end{array}\right)\left(\begin{array}{clrrrr} %
1  & -1  & 0 & 1 \\
-1 & 5 & 2 & -3 \\
0 & 2 & 5 & 1\\
1 & -3 & 1 & 4
\end{array}\right)\left(\begin{array}{clr} %
x_1   \\
x_2  \\
x_3 \\
x_4
\end{array}\right)=\left(\begin{array}{clrrrr} %
1  & 0  & 0 & 0 \\
1 & 1 & 0 & 0 \\
0 & 0 & 1 & 0\\
-1 & 0 & 0 & 1
\end{array}\right)\left(\begin{array}{clr} %
3   \\
-5  \\
-7 \\
2
\end{array}\right)
\end{equation}
\begin{equation}\nonumber
\left(\begin{array}{clrrrr} %
1  & -1  & 0 & 1 \\
0 & 4 & 2 & -2 \\
0 & 2 & 5 & 1\\
0 & -2 & 1 & 3
\end{array}\right)\left(\begin{array}{clr} %
x_1   \\
x_2  \\
x_3 \\
x_4
\end{array}\right)=\left(\begin{array}{clr} %
3   \\
-2  \\
-7 \\
-1
\end{array}\right)
\end{equation}
\begin{equation}\nonumber
\left(\begin{array}{clrrrr} %
1  & 0  & 0 & 0 \\
0 & 1 & 0 & 0 \\
0 & -\frac{1}{2} & 1 & 0\\
0 & \frac{1}{2} & 0 & 1
\end{array}\right)\left(\begin{array}{clrrrr} %
1  & -1  & 0 & 1 \\
0 & 4 & 2 & -2 \\
0 & 2 & 5 & 1\\
0 & -2 & 1 & 3
\end{array}\right)\left(\begin{array}{clr} %
x_1   \\
x_2  \\
x_3 \\
x_4
\end{array}\right)=\left(\begin{array}{clrrrr} %
1  & 0  & 0 & 0 \\
0 & 1 & 0 & 0 \\
0 & -\frac{1}{2} & 1 & 0\\
0 & \frac{1}{2} & 0 & 1
\end{array}\right)\left(\begin{array}{clr} %
3   \\
-2  \\
-7 \\
-1
\end{array}\right)
\end{equation}
\begin{equation}\nonumber
\left(\begin{array}{clrrrr} %
1  & -1  & 0 & 1 \\
0 & 4 & 2 & -2 \\
0 & 0 & 4 & 2\\
0 & 0 & 2 & 2
\end{array}\right)\left(\begin{array}{clr} %
x_1   \\
x_2  \\
x_3 \\
x_4
\end{array}\right)=\left(\begin{array}{clr} %
3   \\
-2  \\
-6 \\
-2
\end{array}\right)
\end{equation}
\begin{equation}\nonumber
\left(\begin{array}{clrrrr} %
1  & 0  & 0 & 0 \\
0 & 1 & 0 & 0 \\
0 & 0 & 1 & 0\\
0 & 0 & -\frac{1}{2} & 1
\end{array}\right)\left(\begin{array}{clrrrr} %
1  & -1  & 0 & 1 \\
0 & 4 & 2 & -2 \\
0 & 0 & 4 & 2\\
0 & 0 & 2 & 2
\end{array}\right)\left(\begin{array}{clr} %
x_1   \\
x_2  \\
x_3 \\
x_4
\end{array}\right)=\left(\begin{array}{clrrrr} %
1  & 0  & 0 & 0 \\
0 & 1 & 0 & 0 \\
0 & 0 & 1 & 0\\
0 & 0 & -\frac{1}{2} & 1
\end{array}\right)\left(\begin{array}{clr} %
3   \\
-2  \\
-6 \\
-2
\end{array}\right)
\end{equation}
\begin{equation}\nonumber
\left(\begin{array}{clrrrr} %
1  & -1  & 0 & 1 \\
0 & 4 & 2 & -2 \\
0 & 0 & 4 & 2\\
0 & 0 & 0 & 1
\end{array}\right)\left(\begin{array}{clr} %
x_1   \\
x_2  \\
x_3 \\
x_4
\end{array}\right)=\left(\begin{array}{clr} %
3   \\
-2  \\
-6 \\
1
\end{array}\right)
\end{equation}
$x_4=1$, $4x_3+2=-6$, $x_3=-2$, $4x_2-4-2=-2$, $x_2=1$, $x_1-1+1=3$, $x_1=3$

Maintenant, passons au deuxi\`eme syst\`eme, qui peut s'\'ecrire, d'apr\`es la Proposition \ref{prop}, $\textbf{G}^T\textbf{G}\textbf{X}=\textbf{B}$ et \^etre r\'esolu en \'ecrivant $\left\{\begin{array}{clrr} %
\textbf{G}^T\textbf{Y} &=\textbf{B} \\
\textbf{G}\textbf{X} &=\textbf{Y}
\end{array}\right.$

\begin{equation}\nonumber
\left(\begin{array}{clrrrr} %
1  & 0  & 0 & 0 \\
-1 & 2 & 0 & 0 \\
0 & 1 & 2 & 0\\
1 & -1 & 1 & 1
\end{array}\right)\left(\begin{array}{clr} %
y_1   \\
y_2  \\
y_3 \\
y_4
\end{array}\right)=\left(\begin{array}{clr} %
3   \\
1  \\
2 \\
2
\end{array}\right)
\end{equation}
$y_1=3$, $3+2y_2=1$, $y_2=2$, $2+2y_3=2$, $y_3=0$, $3-2+y_4=2$, $y_4=1$
\begin{equation}\nonumber
\left(\begin{array}{clrrrr} %
1  & -1  & 0 & 1 \\
0 & 2 & 1 & -1 \\
0 & 0 & 2 & 1\\
0 & 0 & 0 & 1
\end{array}\right)\left(\begin{array}{clr} %
x_1   \\
x_2  \\
x_3 \\
x_4
\end{array}\right)=\left(\begin{array}{clr} %
3   \\
2  \\
0 \\
1
\end{array}\right)
\end{equation}
$x_4=1$, $2x_3+1=0$, $x_3=-\frac{1}{2}$, $2x_2-\frac{1}{2}-1=2$, $x_2=\frac{7}{4}$, $x_1-\frac{7}{4}+1=3$, $x_1=\frac{15}{4}$
\end{example}
\subsection{Matrices de commutation tensorielle et transform\'ees d'une \'equation matricielle}
Dans cette sous-section, les transform\'ees de certaines \'equations matricielles, aux \'equations matricielles de la forme (\ref{eqn1}),  se transforment l'une \`a l'autre \`a aide d'une MCT. 

\paragraph{$\textbf{A}\cdot\textbf{X}\cdot\textbf{B}=\textbf{C}$}
Soient $\textbf{A}$, $\textbf{B}$, et $\textbf{C}$, matrices $m\times n$, $p\times q$ et $m\times q$, respectivement. Consid\'erons l'\'equation matricielle $\textbf{A}\cdot\textbf{X}\cdot\textbf{B}=\textbf{C}$, par rapport \`a $\textbf{X}$,  matrice $n\times q$. Cette \'equation peut se transformer au syst\`eme d'\'equations lin\'eaires dont l'\'equations matricielle est \cite{Ikramov77}
\begin{equation}\label{equ1}
\left(\textbf{A}\otimes \textbf{B}^T\right)\cdot L\left(\textbf{X}\right)=L\left(\textbf{C}\right)
\end{equation}
 ou 
 \begin{equation}\label{equ2}
  \left(\textbf{B}^T\otimes \textbf{A}\right)\cdot L\left(\textbf{X}^T\right)=L\left(\textbf{C}^T\right)
\end{equation}  
   L'\'equation $(\ref{equ2})$ s'obtient en multipliant l'\'equation $(\ref{equ1})$ par la MCT $\textbf{U}_{m\otimes q}$ et en utilisant la remarque $\ref{rmk}$.\\
    R\'eciproquement, l'\'equation $(\ref{equ1})$ s'obtient en multipliant l'\'equation $(\ref{equ2})$ par la MCT $\textbf{U}_{q\otimes m}$.
 
 \paragraph{$\textbf{A}\cdot\textbf{X}+\textbf{X}\cdot\textbf{B}=\textbf{C}$}
 L'\'equation matricielle $\textbf{A}\cdot\textbf{X}+\textbf{X}\cdot\textbf{B}=\textbf{C}$, o\`u $\textbf{A}$ est une matrice $m\times m$, $\textbf{B}$ est une matrice $n\times n$ et $\textbf{C}$ est une matrice $m\times n$, peut se transformer au syst\`eme d'\'equations lin\'eaires dont l'\'equations matricielle est \cite{Ikramov77}
 \begin{equation}\label{equ3}
 \left(\textbf{A}\otimes \textbf{I}_n+\textbf{I}_m\otimes\textbf{B}^T\right)\cdot L\left(\textbf{X}\right)=L\left(\textbf{C}\right)
\end{equation}
o\`u $\textbf{I}_n$ est la matrice unit\'e $n\times n$, ou 
\begin{equation}\label{equ4}
\left(\textbf{I}_n\otimes\textbf{A} +\textbf{B}^T\otimes\textbf{I}_m\right)\cdot L\left(\textbf{X}^T\right)=L\left(\textbf{C}^T\right)
\end{equation}
L'\'equation $(\ref{equ4})$ s'obtient en multipliant l'\'equation $(\ref{equ3})$ par la MCT $\textbf{U}_{m\otimes n}$.\\
R\'eciproquement, l'\'equation $(\ref{equ3})$ s'obtient en multipliant l'\'equation $(\ref{equ4})$ par la MCT $\textbf{U}_{n\otimes m}$.
 \chapter{VERS UNE APPLICATION EN PHYSIQUE DES PARTICULES}
Ce chapitre est bas\'e sur Refs. \cite{Rakotonirina07,Rakotonirina08,Rakotonirina12,Rakotonirina13}    

\section{MATRICES DE PERMUTATION TENSORIELLE ET MATRICES DE GELL-MANN    }

  \subsection{MCT en termes de matrices de Gell-Mann g\'en\'eralis\'ees}
Soit $n\in\mathbb{N}$, $n\geq2$. Les matrices de Gell-Mann g\'en\'eralis\'ees ou matrices de Gell-Mann $n\times n$ sont des matrices hermitiennes et de trace nulles $\Lambda_{1}$,$\Lambda_{2}$,\ldots,$\Lambda_{n^{2}-1}$ qui satisfont la relation de commutation (Cf. par exemple \cite{Narison89})
\begin{equation}\label{e41}
\left[\Lambda_{a},\Lambda_{b}\right]=2i\sum_{a=1}^{n^{2}-1}f_{abc}\Lambda_{c}
\end{equation}
o\`u $f_{abc}$ sont les constantes de structure qui sont r\'eelles et totalement antisym\'etriques, et
\begin{equation*}
Tr\left(\Lambda_{a},\Lambda_{b}\right)=2\delta_{ab}
\end{equation*}
avec $\delta_{ab}$ le symbole de Kronecker.\\
Pour $n=2$, les matrices de Gell-Mann $2\times2$ sont les matrices de Pauli habituelles. Pour $n=3$, elles correspondent aux huit matrices de Gell-Mann $3\times3$, qui se construisent de la fa\c con suivante: les trois premi\`eres sont des matrices $3\times3$ obtenues en ajoutant aux trois matrices de Pauli troisi\`eme ligne et troisi\`eme colonne form\'ees de 0, \`a savoir\\
 $\lambda_{1}=\left(
                         \begin{array}{ccc}
                           0 & 1 & 0 \\
                           1 & 0 & 0 \\
                           0 & 0 & 0 \\
                         \end{array}
                       \right)$
, $\lambda_{2}=\left(
                         \begin{array}{ccc}
                           0 & -i & 0 \\
                           i & 0 & 0 \\
                           0 & 0 & 0 \\
                         \end{array}
                       \right)$
, $\lambda_{3}=\left(
                         \begin{array}{ccc}
                           1 & 0 & 0 \\
                           0 & -1 & 0 \\
                           0 & 0 & 0 \\
                         \end{array}
                       \right)$\\
Les deux secondes matrices $3\times3$  sont obtenues en ajoutant aux deux matrices Pauli non-diagonales deuxi\`eme ligne et deuxi\`eme colonne form\'ees de 0, \`a savoir\\
$\lambda_{4}=\left(
                         \begin{array}{ccc}
                           0 & 0 & 1 \\
                           0 & 0 & 0 \\
                           1 & 0 & 0 \\
                         \end{array}
                       \right)$,
$\lambda_{5}=\left(
                         \begin{array}{ccc}
                           0 & 0 & -i \\
                           0 & 0 & 0 \\
                           i & 0 & 0 \\
                         \end{array}
                       \right)$\\
Les deux troisi\`eme matrices $3\times3$  sont obtenues en ajoutant aux deux matrices Pauli non-diagonales premi\`ere ligne et premi\`ere colonne form\'ees de 0, \`a savoir\\
$\lambda_{6}=\left(
                         \begin{array}{ccc}
                           0 & 0 & 0 \\
                           0 & 0 & 1 \\
                           0 & 1 & 0 \\
                         \end{array}
                       \right)$,
$\lambda_{7}=\left(
                         \begin{array}{ccc}
                           0 & 0 & 0 \\
                           0 & 0 & -i \\
                           0 & i & 0 \\
                         \end{array}
                       \right)$\\
Et finalement, la d\'erni\`ere matrice est une matrice diagonale qui est hermitienne, trace nulle avec\\
\begin{equation*}
Tr\left({\lambda_{8}}^{2}\right)=2
\end{equation*}
\`a savoir\\
$\lambda_{8}=\frac{1}{\sqrt{3}}\left(
                                 \begin{array}{ccc}
                                   1 & 0 & 0 \\
                                   0 & 1 & 0 \\
                                   0 & 0 & -2 \\
                                 \end{array}
                               \right)$\\
De fa\c con analogue, nous pouvons construire \`a partir des matrices de Gell-Mann $(n-1)\times(n-1)$ les matrices de Gell-Mann $n\times n$. Les premi\`eres $[(n-1)^{2}-1]$ matrices de Gell-Mann $n\times n$ sont obtenues en ajoutant $n$-i\`eme ligne et $n$-i\`eme colonne form\'ees de 0 \`a chaque matrices de Gell-Mann $(n-1)\times(n-1)$. Les $(2n-2)$ matrices de Gell-Mann 
$n\times n$ sont les matrices sym\'etriques et les matrices antisym\'etriques non-diagonales suivantes \\
$\Lambda_{(n-1)^{2}}=\left(
  \begin{array}{cccccc}
    0 & 0 & \ldots & \ldots & 0 & 1 \\
    0 & 0 & &  &  & 0 \\
    \vdots &  & \ddots & &  & \vdots \\
    \vdots &  &  & \ddots &  & \vdots \\
    0 &  &  &  &  & 0 \\
    1 & 0 & \ldots & \ldots & 0 & 0 \\
  \end{array}
\right)$, $\Lambda_{(n-1)^{2}+1}=\left(
  \begin{array}{cccccc}
    0 & 0 & \ldots & \ldots & 0 & -i \\
    0 & 0 & &  &  & 0 \\
    \vdots &  & \ddots & &  & \vdots \\
    \vdots &  &  & \ddots &  & \vdots \\
    0 &  &  &  &  & 0 \\
    i & 0 & \ldots & \ldots & 0 & 0 \\
  \end{array}
\right)$, $\Lambda_{(n-1)^{2}+2}=\left(
  \begin{array}{cccccc}
    0 & 0 & \ldots & \ldots & 0 & 0 \\
    0 & 0 & &  &  & 1 \\
    \vdots &  & \ddots & &  & 0 \\
    0 &  &  & \ddots &  & \vdots \\
    0 &  &  &  &  & 0 \\
    0 & 1 & \ldots & \ldots & 0 & 0 \\
  \end{array}
\right)$, $\Lambda_{(n-1)^{2}+3}=\left(
  \begin{array}{cccccc}
    0 & 0 & \ldots & \ldots & 0 & 0 \\
    0 & 0 & &  &  & -i \\
    \vdots &  & \ddots & &  & 0 \\
    0 &  &  & \ddots &  & \vdots \\
    0 &  &  &  &  & 0 \\
    0 & i & \ldots & \ldots & 0 & 0 \\
  \end{array}
\right)$,\\

 ...,\\

 $\Lambda_{n^{2}-3}=\left(
  \begin{array}{cccccc}
    0 & 0 & \ldots & \ldots & 0 & 0 \\
    \vdots & 0 & &  &  & \vdots \\
    \vdots &  & \ddots & &  & \vdots \\
    \vdots &  &  & \ddots &  & 0 \\
     &  &  &  & 0 & 1 \\
    0 & \ldots & \ldots & \ldots & 1 & 0 \\
  \end{array}
\right)$, $\Lambda_{n^{2}-2}=\left(
  \begin{array}{cccccc}
    0 & 0 & \ldots & \ldots & 0 & 0 \\
    \vdots & 0 & &  &  & \vdots \\
    \vdots &  & \ddots & &  & \vdots \\
    \vdots &  &  & \ddots &  & 0 \\
     &  &  &  & 0 & -i \\
    0 & \ldots & \ldots & \ldots & i & 0 \\
  \end{array}
\right)$\\
et finalement, la d\'erni\`ere matrice est une matrice diagonale, hermitienne et de trace nulle avec
\begin{equation*}
Tr\left({\Lambda_{n^{2}-1}}^{2}\right)=2
\end{equation*}
\`a savoir\\
$\Lambda_{n^{2}-1}=\frac{1}{\sqrt{C_{n}^{2}}}\left(
                                               \begin{array}{cccccc}
                                                 1 & 0 & \ldots & \ldots & \ldots & 0 \\
                                                 0 & 1 &  &  &  &  \\
                                                 \vdots &  & \ddots &  &  & \vdots \\
                                                 \vdots &  &  & \ddots &  & \vdots \\
                                                  &  &  &  & 1 & 0 \\
                                                 0 & \ldots & \ldots & \ldots & 0 & -(n-1) \\
                                               \end{array}
                                             \right)$\\

Elles satisfont aussi la relation d'anticommutation (Cf. par exemple\cite{Narison89})
\begin{equation}\label{e42}
\left\{\Lambda_{a},
\Lambda_{b}\right\}=\frac{4}{n}\delta_{ab}I_{n}+2\displaystyle\sum_{c=1}^{n^{2}-1}d_{abc}\Lambda_{c}
\end{equation}
o\`u les constantes $d_{abc}$ sont r\'eelles et totalement sym\'etriques, et en utilisant les relations (\ref{e41}) et (\ref{e42}), nous avons
\begin{equation}\label{e43}
\Lambda_{a}\Lambda_{b}=\frac{2}{n}\delta_{ab}+\displaystyle\sum_{c=1}^{n^{2}-1}d_{abc}\Lambda_{c}
+i\displaystyle\sum_{c=1}^{n^{2}-1}f_{abc}\Lambda_{c}
\end{equation}
Les constantes de structure satisfont la relation (Cf. par exemple\cite{Narison89})
\begin{equation}\label{e44}
\displaystyle\sum_{e=1}^{n^{2}-1}f_{abe}f_{cde}=\frac{2}{n}\left(\delta_{ac}\delta_{bd}-\delta_{ad}\delta_{bc}\right)
+\displaystyle\sum_{e=1}^{n^{2}-1}d_{ace}d_{dbe}-\displaystyle\sum_{e=1}^{n^{2}-1}d_{ade}d_{bce}
\end{equation}\\
Pour la d\'emonstration du Th\'eor\`eme \ref{thm32} ci-dessous, notons, pour $1\leq i < j\leq n$, les $C_{n}^{2}\;=\;\frac{n!}{2!(n-2)!}$ matrices de Gell-Mann $n\times n$ qui sont sym\'etriques avec des \'el\'ements tous 0 sauf le $i$-i\`eme ligne $j$-i\`eme colonne et le $j$-i\`eme ligne  $i$-i\`eme colonne qui sont \'egaux \`a 1, par $\Lambda^{(ij)}$, les $C_{n}^{2}\;=\;\frac{n!}{2!(n-2)!}$ matrices de Gell-Mann $n\times n$ qui sont antisym\'etriques avec des \'el\'ements tous 0 sauf le $i$-i\`eme ligne $j$-i\`eme colonne qui est \'egal \`a $-i$ et le $j$-i\`eme ligne $i$-i\`eme colonne qui est \'egal \`a $i$  , par $\Lambda^{[ij]}$ et enfin, par $\Lambda^{(d)}$,$1 \leq d \leq n-1$, les $(n - 1)$ matrices de Gell-Mann $n\times n$ suivantes, qui sont diagonales:\\

$\Lambda^{(1)}\;=\;\left(
                     \begin{array}{cccccc}
                       1 & 0 &  & \ldots &  & 0 \\
                       0 & -1 &  &  &  &  \\
                        &  & 0 &  &  & \vdots \\
                       \vdots &  &  & \ddots &  &  \\
                        &  &  &  & \ddots &  \\
                       0 &  & \ldots &  &  & 0 \\
                     \end{array}
                   \right)$, $\Lambda^{(2)}\;=\;\frac{1}{\sqrt{3}}\left(
                               \begin{array}{cccccc}
                                 1 & 0 &  & \ldots &  & 0 \\
                                 0 & 1 &  &  &  &  \\
                                  &  & -2 &  &  & \vdots \\
                                 \vdots &  &  & 0 &  &  \\
                                  &  &  &  & \ddots &  \\
                                 0 &  & \ldots &  &  & 0 \\
                               \end{array}
                             \right)$,\\

                             \ldots,

                             $\Lambda^{(n-1)}\;=\;\frac{1}{\sqrt{C_{n}^{2}}}\left(
                                                                              \begin{array}{cccccc}
                                                                                1 & 0 &  & \ldots &  & 0 \\
                                                                                0 & 1 &  &  &  &  \\
                                                                                 &  & 1 &  &  & \vdots \\
                                                                                \vdots &  &  & \ddots &  &  \\
                                                                                 &  &  &  & 1 &  \\
                                                                                0 &  & \ldots &  &  & -(n-1) \\
                                                                              \end{array}
                                                                            \right)$\\

\begin{theorem}\label{thm32}
Nous avons

\begin{equation}\label{e44prime}
\textbf{U}_{n\otimes n}\;=\;\frac{1}{n}I_{n}\otimes
I_{n}+\frac{1}{2}\displaystyle\sum_{i=1}^{n^{2}-1}\Lambda_{i}\otimes\Lambda_{i}
\end{equation}
\end{theorem}
\begin{proof}
\begin{equation*}
I_{n}\otimes
I_{n}\;=\;\left(\delta_{j_{1}j_{2}}^{i_{1}i_{2}}\right)\;=\;
\left(\delta_{j_{1}}^{i_{1}}\delta_{j_{2}}^{i_{2}}\right)
\end{equation*}
\numberwithin{equation}{section}
\begin{equation}\label{e31}
\textbf{U}_{n\otimes n}\;=\;
\left(\delta_{j_{2}}^{i_{1}}\delta_{j_{1}}^{i_{2}}\right)
\end{equation}
o\`u,\\
  $i_1i_2$ sont des indices de ligne\\
  $j_1j_2$ sont des indices de colonne \cite{Fujii01}.\\
  Consid\'erons d'abord les $C_{n}^{2}$  matrices de Gell-Mann $n\times n$ sym\'etriques, qui peuvent s'\'ecrire \\
  \begin{equation*}
  \begin{split}
  \Lambda^{(ij)}\;&=
  \;\left({\Lambda^{(ij)}}_{k}^{l}\right)_{1\leq l\leq n, 1\leq k\leq n}
  \;\\
  &=  \;\left(\delta^{il}\delta^{j}_{k}\right)_{1\leq l\leq n, 1\leq k\leq n}+
  \left(\delta^{jl}\delta^{i}_{k}\right)_{1\leq l\leq n, 1\leq k\leq
  n}\;\\
  &=\;\left(\delta^{il}\delta^{j}_{k}+
  \delta^{jl}\delta^{i}_{k}\right)_{1\leq l\leq n, 1\leq k\leq n}
  \end{split}
  \end{equation*}
  Alors
  \begin{equation*}
\Lambda^{(ij)}\otimes
\Lambda^{(ij)}\;=\;\left({\left(\Lambda^{(ij)}\otimes
\Lambda^{(ij)}\right)}_{k_{1}k_{2}}^{l_{1}l_{2}}\right)\;=\;\left(\delta^{il_{1}}\delta^{j}_{k_{1}}+
  \delta^{jl_{1}}\delta^{i}_{k_{1}}\right)\left(\delta^{il_{2}}\delta^{j}_{k_{2}}+
  \delta^{jl_{2}}\delta^{i}_{k_{2}}\right)
\end{equation*}
$l_{1}l_{2}$ indices de ligne, $k_{1}k_{2}$ indices de colonne.\\
C'est-\`a-dire
\begin{equation*}
 {\left(\Lambda^{(ij)}\otimes
\Lambda^{(ij)}\right)}_{k_{1}k_{2}}^{l_{1}l_{2}}\;=\;\delta^{il_{1}}\delta^{j}_{k_{1}}
\delta^{il_{2}}\delta^{j}_{k_{2}}+\delta^{il_{1}}\delta^{j}_{k_{1}}\delta^{jl_{2}}\delta^{i}_{k_{2}}
+\delta^{jl_{1}}\delta^{i}_{k_{1}}\delta^{il_{2}}\delta^{j}_{k_{2}}
+\delta^{jl_{1}}\delta^{i}_{k_{1}}\delta^{jl_{2}}\delta^{i}_{k_{2}}
\end{equation*}
  Les $C_{n}^{2}$  matrices de Gell-Mann $n\times n$, antisym\'etriques peuvent s'\'ecrire\\
\begin{equation*}
\Lambda^{[ij]}\;=
  \;\left({\Lambda^{[ij]}}_{k}^{l}\right)_{1\leq l\leq n, 1\leq k\leq n}
  \;=\;\left(-i\delta^{il}\delta^{j}_{k}+
  i\delta^{jl}\delta^{i}_{k}\right)_{1\leq l\leq n, 1\leq k\leq n}
\end{equation*}
 Alors
  \begin{equation*}
\Lambda^{[ij]}\otimes
\Lambda^{[ij]}\;=\;\left({\left(\Lambda^{[ij]}\otimes
\Lambda^{[ij]}\right)}_{k_{1}k_{2}}^{l_{1}l_{2}}\right)
\end{equation*}
\begin{equation*}
{\left(\Lambda^{[ij]}\otimes
\Lambda^{[ij]}\right)}_{k_{1}k_{2}}^{l_{1}l_{2}}\;=\;-\delta^{il_{1}}\delta^{j}_{k_{1}}
\delta^{il_{2}}\delta^{j}_{k_{2}}+\delta^{il_{1}}\delta^{j}_{k_{1}}\delta^{jl_{2}}\delta^{i}_{k_{2}}
+\delta^{jl_{1}}\delta^{i}_{k_{1}}\delta^{il_{2}}\delta^{j}_{k_{2}}
-\delta^{jl_{1}}\delta^{i}_{k_{1}}\delta^{jl_{2}}\delta^{i}_{k_{2}}
\end{equation*}
et
\begin{multline*}
  \displaystyle\sum_{1\leq i< j\leq n}{\left(\Lambda^{(ij)}\otimes
\Lambda^{(ij)}\right)}_{k_{1}k_{2}}^{l_{1}l_{2}}+\displaystyle\sum_{1\leq
i< j\leq n}{\left(\Lambda^{[ij]}\otimes
\Lambda^{[ij]}\right)}_{k_{1}k_{2}}^{l_{1}l_{2}}\\
=\;2\displaystyle\sum_{1\leq i< j\leq
n}\left(\delta^{il_{1}}\delta^{j}_{k_{1}}\delta^{jl_{2}}\delta^{i}_{k_{2}}
+\delta^{jl_{1}}\delta^{i}_{k_{1}}\delta^{il_{2}}\delta^{j}_{k_{2}}\right)\;\\
=\;2\displaystyle\sum_{i\neq
j}\delta^{il_{1}}\delta^{j}_{k_{1}}\delta^{jl_{2}}\delta^{i}_{k_{2}}
  \end{multline*}
  la $l_{1}l_{2}$-i\`eme ligne $k_{1}k_{2}$-i\`eme colonne de la matrice\\
   $\displaystyle\sum_{1\leq i< j\leq
  n}\Lambda^{(ij)}\otimes
\Lambda^{(ij)}+\displaystyle\sum_{1\leq i< j\leq
  n}\Lambda^{[ij]}\otimes
\Lambda^{[ij]}$.\\

Maintenant, consid\'erons les  matrices de Gell-Mann $n\times n$, diagonales. Soit $d\in\mathbb{N}$, $1 \leq d \leq n-1$,
\begin{center}
$\Lambda^{(d)}\;=\;\frac{1}{\sqrt{C_{d+1}^{2}}}\left(\delta_{k}^{l}
\displaystyle\sum_{p=1}^{d}\delta_{k}^{p}-d\delta_{k}^{l}\delta_{k}^{d+1}\right)$
\end{center}
et la $l_{1}l_{2}$-i\`eme ligne $k_{1}k_{2}$-i\`eme colonne de la matrice
$\Lambda^{(d)}\otimes \Lambda^{(d)}$ est\\
\begin{equation*}
\begin{split}
\left(\Lambda^{(d)}\otimes
\Lambda^{(d)}\right)_{k_{1}k_{2}}^{l_{1}l_{2}}
&=\frac{1}{C_{d+1}^{2}}\delta_{k_{1}}^{l_{1}}\delta_{k_{2}}^{l_{2}}
\left(\displaystyle\sum_{q=1}^{d}\displaystyle\sum_{p=1}^{d}\delta_{k_{1}}^{q}\delta_{k_{2}}^{p}\right)\\
&\quad-\frac{1}{C_{d+1}^{2}}\delta_{k_{1}}^{l_{1}}\delta_{k_{2}}^{l_{2}}
\left(d\delta_{k_{2}}^{d+1}\displaystyle\sum_{p=1}^{d}\delta_{k_{1}}^{p}\right)\\
&\quad-\frac{1}{C_{d+1}^{2}}\delta_{k_{1}}^{l_{1}}\delta_{k_{2}}^{l_{2}}
\left(d\delta_{k_{1}}^{d+1}\displaystyle\sum_{p=1}^{d}\delta_{k_{2}}^{p}\right)\\
&\quad+\frac{1}{C_{d+1}^{2}}\delta_{k_{1}}^{l_{1}}\delta_{k_{2}}^{l_{2}}
\left(d^{2}\delta_{k_{1}}^{d+1}\delta_{k_{2}}^{d+1}\right)
\end{split}
\end{equation*}
$\Lambda^{(d)}\otimes \Lambda^{(d)}$ est une matrice diagonale, donc tout ce que nous devons \`a faire est de calculer les \'el\'ements sur la diagonale, o\`u $l_{1} \;=\; k_{1}$ et $l_{2} \;=\; k_{2}$. Alors,
\begin{equation*}
\begin{split}
\sum_{d=1}^{n-1}\left(\Lambda^{(d)}\otimes
\Lambda^{(d)}\right)_{k_{1}k_{2}}^{l_{1}l_{2}}
&=\sum_{d=1}^{n-1}\frac{1}{C_{d+1}^{2}}\left(\displaystyle\sum_{q=1}^{d}\delta_{k_{1}}^{q}\right)
\left(\displaystyle\sum_{p=1}^{d}\delta_{k_{2}}^{p}\right)\\
&\quad-\sum_{d=1}^{n-1}\frac{1}{C_{d+1}^{2}}d\delta_{k_{2}}^{d+1}\displaystyle\sum_{p=1}^{d}\delta_{k_{1}}^{p}\\
&\quad-\sum_{d=1}^{n-1}\frac{1}{C_{d+1}^{2}}d\delta_{k_{1}}^{d+1}\displaystyle\sum_{p=1}^{d}\delta_{k_{2}}^{p}\\
&\quad+\sum_{d=1}^{n-1}\frac{1}{C_{d+1}^{2}}d^{2}\delta_{k_{1}}^{d+1}\delta_{k_{2}}^{d+1}
\end{split}
\end{equation*}
la $l_{1}l_{2}$-i\`eme ligne $k_{1}k_{2}$-i\`eme colonne de la 
matrice diagonale $\displaystyle\sum_{d=1}^{n-1}\Lambda^{(d)}\otimes
\Lambda^{(d)}$ avec $l_{1}
\;=\; k_{1}$ et $l_{2} \;=\; k_{2}$.\\
Laissons nous distinguer deux cas.\\
 \underline{$1^{er}$ cas}: $k_{1}\neq1$ or $k_{2}\neq1$\\
 \indent cas1: $k_{1}\neq k_{2}$

\text{Si $k_{1}<k_{2}$},
\begin{equation*}
\begin{split}
\sum_{d=1}^{n-1}\left(\Lambda^{(d)}\otimes
\Lambda^{(d)}\right)_{k_{1}k_{2}}^{l_{1}l_{2}}\;&=\;\sum_{d=k_{2}}^{n-1}\frac{1}{C_{d+1}^{2}}-\frac{k_{2}-1}{C_{k_{2}}^{2}}\;\\
&\quad=\;2\left[\sum_{d=k_{2}}^{n-1}\left(\frac{1}{d}-\frac{1}{d+1}\right)-\frac{1}{k_{2}}\right]\;\\
&\quad=\;-\frac{2}{n}
\end{split}
\end{equation*}
 \text{De fa\c con Similaire, si $k_{1}>k_{2}$},
$\displaystyle\sum_{d=1}^{n-1}\left(\Lambda^{(d)}\otimes
\Lambda^{(d)}\right)_{k_{1}k_{2}}^{l_{1}l_{2}}\; =\;-\frac{2}{n}$\\
\indent cas2: $k_{1}\;=\; k_{2}\;\neq\;1$
\begin{equation*}
\begin{split}
\sum_{d=1}^{n-1}\left(\Lambda^{(d)}\otimes
\Lambda^{(d)}\right)_{k_{1}k_{2}}^{l_{1}l_{2}}\;
&=\;\sum_{d=k_{2}}^{n-1}\frac{1}{C_{d+1}^{2}}+\frac{\left(k_{2}-1\right)^{2}}{C_{k_{2}}^{2}}\;\\
&\quad=\;\frac{2}{k_{2}}-\frac{2}{n}+\frac{\left(k_{2}-1\right)^{2}}{C_{k_{2}}^{2}}\;\\
 &\quad=\;2-\frac{2}{n}
\end{split}
\end{equation*}
\underline{$2^{e}$ cas}: $k_{1}= k_{2}=1$
\begin{equation*}
\sum_{d=1}^{n-1}\left(\Lambda^{(d)}\otimes
\Lambda^{(d)}\right)_{k_{1}k_{2}}^{l_{1}l_{2}}\;=\;\sum_{d=1}^{n-1}\frac{1}{C_{d+1}^{2}}\;=\;2-\frac{2}{n}
\end{equation*}
Nous pouvons compacter ces cas en une seule formule
\begin{equation*}
\sum_{d=1}^{n-1}\left(\Lambda^{(d)}\otimes
\Lambda^{(d)}\right)_{k_{1}k_{2}}^{l_{1}l_{2}}\;
=\;-\frac{2}{n}\delta_{k_{1}}^{l_{1}}\delta_{k_{2}}^{l_{2}}
+2\sum_{i=1}^{n}\delta^{il_{1}}\delta_{k_{1}}^{i}\delta^{il_{2}}\delta_{k_{2}}^{i}
\end{equation*}
qui donne la diagonale de la matrice diagonale
$\displaystyle\sum_{d=1}^{n-1}\Lambda^{(d)}\otimes \Lambda^{(d)}$.\\

Pour toutes les matrices de Gell-Mann $n\times n$ nous avons \\
\begin{multline*}
\sum_{1\leq i< j\leq n}{\left(\Lambda^{(ij)}\otimes
\Lambda^{(ij)}\right)}_{k_{1}k_{2}}^{l_{1}l_{2}}+\sum_{1\leq i<
j\leq n}{\left(\Lambda^{[ij]}\otimes
\Lambda^{[ij]}\right)}_{k_{1}k_{2}}^{l_{1}l_{2}}+\sum_{d=1}^{n-1}\left(\Lambda^{(d)}\otimes
\Lambda^{(d)}\right)_{k_{1}k_{2}}^{l_{1}l_{2}}\;\\
=\;-\frac{2}{n}\delta_{k_{1}}^{l_{1}}\delta_{k_{2}}^{l_{2}}
+2\sum_{i=1}^{n}\delta^{il_{1}}\delta_{k_{1}}^{i}\delta^{il_{2}}\delta_{k_{2}}^{i}
+2\sum_{i\neq
j}\delta^{il_{1}}\delta_{k_{1}}^{j}\delta^{jl_{2}}\delta_{k_{2}}^{i}\\
=\;-\frac{2}{n}\delta_{k_{1}}^{l_{1}}\delta_{k_{2}}^{l_{2}}
+2\sum_{j=1}^{n}\sum_{i=1}^{n}\delta^{il_{1}}\delta_{k_{1}}^{j}\delta^{jl_{2}}\delta_{k_{2}}^{i}\\
=\;-\frac{2}{n}\delta_{k_{1}}^{l_{1}}\delta_{k_{2}}^{l_{2}}+2\delta_{k_{2}}^{l_{1}}\delta_{k_{1}}^{l_{2}}
\end{multline*}
pour tous $l_{1}$, $l_{2}$, $k_{1}$,$k_{2}\in\{1, 2, \ldots, n\}$.\\
D'o\`u, en utilisant \eqref{e31}
\begin{equation*}
\displaystyle\sum_{i=1}^{n^{2}-1}\Lambda_{i}\otimes\Lambda_{i}\;=\;-\frac{2}{n}I_{n}\otimes
I_{n}+2\textbf{U}_{n\otimes n}
\end{equation*}
et le th\'eor\`eme est prouv\'e.
 \end{proof}
 \subsection{Matrices de Gell-Mann rectangles}
 Dans cette sous-section nous voulons g\'en\'eraliser la formule (\ref{e44prime}) \`a l'expression de $\textbf{U}_{n \otimes p}$ avec $n\neq p$.\\
 \noindent Nous avions essay\'e d'exprimer les MCT $\textbf{U}_{2 \otimes 3}$ et $\textbf{U}_{3 \otimes 2}$ comme 
combinaisons lin\'eaire des  produits tensoriels des matrices de Pauli avec les matrices de Gell-Mann $3 \times 3$, en esp\'erant d'avoir des expressions qui m\`enent \`a la g\'en\'eralisation de (\ref{e44prime}) \`a l'expression de $\textbf{U}_{n \otimes p}$ avec $n\neq p$. Cependant, les expressions obtenues ne sont pas assez inter\'essant pour la g\'en\'eralisation voulue. Nous avons remarqu\'e que pour g\'en\'eraliser (\ref{e44prime}) \`a l'expression de $\textbf{U}_{n \otimes p}$, $n \neq p$, nous devons utiliser des matrices rectangles au lieu de matrices carr\'ees. Nous appelons matrices de Gell-Mann rectangles des telles matrices rectangles.\\

D'abord, laissons nous consid\'erer quelques cas particuliers.

\begin{itemize}

\item Matrices de Gell-Mann $2 \times 3$

Etant inspir\'e par une fa\c con de construire les matrices de Gell-Mann $n \times
n$ \`a partir des matrices de Gell-Mann $( n - 1 ) \times ( n - 1)$, nous ajoutons aux matrices de Pauli et $I_2$ troisi\`eme  colonne form\'ee de zeros. Alors, nous avons un syst\`eme form\'e par
\begin{equation*}
I_{2 \times 3} = \left( \begin{array}{ccc}
  1 & 0 & 0\\
  0 & 1 & 0
\end{array} \right), \Lambda_1 = \left( \begin{array}{ccc}
  0 & 1 & 0\\
  1 & 0 & 0
\end{array} \right), \Lambda_2 = \left( \begin{array}{ccc}
  0 & - i & 0\\
  i & 0 & 0
\end{array} \right),
 \Lambda_3 = \left( \begin{array}{ccc}
  1 & 0 & 0\\
  0 & - 1 & 0
\end{array} \right)
\end{equation*}
 Et pour obtenir une base de $\mathbb{C}^{2 \times 3}$, nous ajoutons au syst\`eme les matrices\\
 $\Lambda_4 = \left(
\begin{array}{ccc}
  0 & 0 & \sqrt{2}\\
  0 & 0 & 0
\end{array} \right), \Lambda_5 = \left( \begin{array}{ccc}
  0 & 0 & 0\\
  0 & 0 & \sqrt{2}
\end{array} \right)$. Nous pouvons v\'erifier facilement que 
\[ \textbf{U}_{2 \otimes 3} = \frac{1}{2} I_{2 \times 3}^+ \otimes I_{2 \times 3} +
   \frac{1}{2} \sum_{a = 1}^5 \Lambda_a^+ \otimes \Lambda_a \]
o\`u $\Lambda_a^+$ est le conjugu\'e hermitien de $\Lambda_a$.

\item  Matrices de Gell-Mann $3 \times 2$

De fa\c con analogue, mais cette fois nous ajoutons aux matrices de Pauli et $I_2$ une troisi\`eme ligne form\'e de zeros, au lieu de colonne. Alors, nous avons un syst\`eme $( \Lambda_a )_{1 \leqslant a
\leqslant 5}$ de matrices $3 \times 2$ qui satisfont
\[ \textbf{U}_{3 \otimes 2} = \frac{1}{2} I_{3 \times 2}^+ \otimes I_{3 \times 2} +
   \frac{1}{2} \sum_{a = 1}^5 \Lambda_a^+ \otimes \Lambda_a \]
En effet, $\textbf{U}_{n \otimes p} = \textbf{U}_{n \otimes p}^T = \textbf{U}_{n \otimes p}^+$,
pour tous $n, p \in \mathbb{N}$, $n, p \geqslant 2$.

\item Matrices de Gell-Mann $2 \times 4$ 

Utilisant encore la m\^eme mani\`ere, mais pour ce cas nous ajoutons aux matrices de Pauli et $I_2$ trois\`eme et quatri\`eme 
colonnes form\'ees de zeros. Alors, nous avons un syst\`eme form\'e par quatre matrices $2\times 4$  $I_{2 \times 4}, \Lambda_1, \Lambda_2, \Lambda_3$. Et pour obtenir une base de $\mathbb{C}^{2 \times 4}$, nous ajoutons au syst\`eme les matrices\\
\noindent $\Lambda_4 = \left(
\begin{array}{cccc}
  0 & 0 & \sqrt{2} & 0\\
  0 & 0 & 0 & 0
\end{array} \right)$, $\Lambda_5 = \left( \begin{array}{cccc}
  0 & 0 & 0 & 0\\
  0 & 0 & \sqrt{2} & 0
\end{array} \right)$, $\Lambda_6 = \left( \begin{array}{cccc}
  0 & 0 & 0 & \sqrt{2}\\
  0 & 0 & 0 & 0
\end{array} \right)$, $\Lambda_7 = \left( \begin{array}{cccc}
  0 & 0 & 0 & 0\\
  0 & 0 & 0 & \sqrt{2}
\end{array} \right)$. Le syst\`eme satisfait la relation
\[ \textbf{U}_{2 \otimes 4} = \frac{1}{2} I_{2 \times 4}^+ \otimes I_{2 \times 4} +
   \frac{1}{2} \sum_{a = 1}^7 \Lambda_a^+ \otimes \Lambda_a \]
\end{itemize}
\begin{definition}
  Soient $n, p \in \mathbb{N}$, $p \geqslant n \geqslant 2$. Nous appelons matrices de Gell-Mann $n\times p$ les matrices $n$ lignes et $p$ colonnes $\Lambda_1$, $\Lambda_2$,
  ..., $\Lambda_{n^2 - 1}$, $\Lambda_{n^2}$, $\Lambda_{n^2 + 1}$, ...,
  $\Lambda_{np - 1}$ telles que:

  $\Lambda_1$, $\Lambda_2$, ..., $\Lambda_{n^2 - 1}$ sont obtenues en ajoutant aux matrices de Gell-Mann $n\times n$, $( n + 1 )$-i\`eme, $( n + 2 )$-i\`eme,
  ..., $p$-i\`eme colonnes, form\'ees de zeros;

  $\Lambda_{n^2} = \sqrt{2} \textbf{E}_{n \times p}^{( 1, n + 1 )}$, $\Lambda_{n^2 + 1}
  = \sqrt{2} \textbf{E}_{n \times p}^{( 2, n + 1 )},$..., $\Lambda_{n^2 + n - 1} =
  \sqrt{2} \textbf{E}_{n \times p}^{( n, n + 1 )}$,

  $\Lambda_{n^2 + n} = \sqrt{2} \textbf{E}_{n \times p}^{( 1, n + 2 )}$, $\Lambda_{n^2
  + n + 1} = \sqrt{2} \textbf{E}_{n \times p}^{( 2, n + 2 )}$, ..., $\Lambda_{n^2 + 2 n
  - 1} = \sqrt{2} \textbf{E}_{n \times p}^{( 1, n + 2 )}$,

  ....................................................................................................................,

  $\Lambda_{n^{} ( p - 1 )} = \sqrt{2} \textbf{E}_{n \times p}^{( 1, p )}$,
  $\Lambda_{n^{} ( p - 1 ) + 1} = \sqrt{2} \textbf{E}_{n \times p}^{( 2, p )}$, ...,
  $\Lambda_{n^{} p - 1} = \sqrt{2} \textbf{E}_{n \times p}^{( n, p )}$.
\end{definition}
  Alors, nous d\'efinissons les matrices de Gell-Mann $p \times n$ comme les matrices obtenues en prenant les conjugu\'es  hermitiens  des matrices de Gell-Mann $n \times p$.

\begin{proposition}
  Pour $n, p \in \mathbb{N}$, $p, n \geqslant 2$, consid\'erons le syst\`eme de matrices de Gell-Mann $n\times p$, $\Lambda_1$, $\Lambda_2$, ..., $\Lambda_{np -
  1}$. Alors,
  \begin{equation}
    \label{eq2} \textbf{U}_{n \otimes p} = \frac{1}{\inf( n, p )} I_{n \times
    p}^+ \otimes I_{n \times p} + \frac{1}{2} \sum_{a = 1}^{np - 1}
    \Lambda_a^+ \otimes \Lambda_a
  \end{equation}
\end{proposition}

\begin{proof}
  Supposons $p \geqslant n$.
  \begin{equation}
    \label{eq3} \frac{1}{2} \sum_{a = n^2}^{np - 1} \Lambda_a^+ \otimes
    \Lambda_a = \sum_{( j, l ) = ( 1, n + 1 )}^{( n, p )} \textbf{E}_{n \times p}^{( j,
    l )^T} \otimes \textbf{E}_{n \times p}^{( j, l )^{}}
  \end{equation}
  En utilisant la proposition \ref{prop23} et la formule (\ref{e44prime}) nous avons 
  \begin{equation}
    \label{eq4} \sum_{( j, l ) = ( 1, 1 )}^{( n, n )} \textbf{E}_{n \times n}^{( j, l
    )^T} \otimes \textbf{E}_{n \times n}^{( j, l )^{}} = \frac{1}{n} I_n \otimes I_n +
    \frac{1}{2} \sum _{a = 1}^{n^2 - 1}\Lambda_a^{( n )} \otimes \Lambda_a^{( n )}
  \end{equation}
  En ajoutant, dans (\ref{eq4}), aux termes \`a gauche de $\otimes$'s $p - n$ lignes, $( n + 1 )$-i\`eme, $( n + 2 )$-i\`eme, ..., $p$-i\`eme lignes, et \`a droite $p
  - n$ colonnes, $( n + 1 )$-i\`eme, $( n + 2 )$-i\`eme, ..., $p$-i\`eme colonnes, form\'ees de zeros nous avons 
  \[ \sum_{( j, l ) = ( 1, 1 )}^{( n, p )} \textbf{E}_{n \times p}^{( j, l )^T} \otimes
     \textbf{E}_{n \times p}^{( j, l )^{}} - \sum_{( j, l ) = ( 1, n + 1 )}^{( n, p )}
     \textbf{E}_{n \times p}^{( j, l )^T} \otimes \textbf{E}_{n \times p}^{( j, l )^{}} =
     \frac{1}{n} I_{n \times
    p}^+ \otimes I_{n \times p} + \frac{1}{2} \sum_{a = 1}^{n^2
     - 1} \Lambda_a^+ \otimes \Lambda_a \]
  En utilisant la proposition \ref{prop23} et (\ref{eq3}) nous avons (\ref{eq2}).
\end{proof}

Maintenant, nous allons donner quelques proprit\'es des matrices de Gell-Mann rectangles.

\begin{proposition}
  Pour $n, p \in \mathbb{N}$, $p, n \geqslant 2$, soit $(\Lambda_a)_{1\leq a\leq np-1}$ un syst\`eme de matrices de Gell-Mann
 $n\times p$. Alors,  $Tr (\Lambda_a^+ \Lambda_b) = 2 \delta_{ab}$.
  
\end{proposition}

\begin{proposition}
Pour $n, p \in \mathbb{N}$, $p\geqslant n \geqslant 2$, soit
$(\Lambda_a)_{1\leq a\leq np-1}$ un syst\`eme de matrices de Gell-Mann
 $n\times p$. Alors,
\begin{equation*}
     \Lambda_a^{} \Lambda^+_b - \Lambda_b^{} \Lambda^+_a = i\sum_{c=1}^{n^{2}-1}f_{abc} \Lambda^{(
     n )}_c
     \end{equation*}
     o\`u les $f_{abc}$ sont les composantes d'un tenseur totalement antisym\'etrique, avec
     $f_{abc}=0$ si au moins un de $a$, $b$, $c$ est dans $\{n^{2}, n^{2}+1, ...,
     np-1\}$.
  \end{proposition}

\section*{Conclusion}

Etant inspir\'e par une fa\c con de construire les matrices de Gell-Mann $n \times
n$ \`a partir des matrices de Gell-Mann $( n - 1 ) \times ( n - 1)$, nous pouvons construire une base de $\mathbb{C}^{n \times p}$, dont les \'el\'ements font la g\'en\'eralisation de l'expression de $\textbf{U}_{n \otimes n}$ en termes de matrices de Gell-Mann $n \times n$ \`a l'expression de $\textbf{U}_{n \otimes p}$.
\subsection{Exprimer une matrice de permutation tensorielle $p^{\otimes n}$ en termes de matrices de Gell-Mann g\'en\'eralis\'ees}
Des th\'eor\`emes de la sous-section \ref{soussec323} et des relations sur les matrices de Gell-Mann g\'en\'eralis\'ees sont dont nous avons besoin pour exprimer une matrice de permutation tensorielle en termes de matrices de Gell-Mann g\'en\'eralis\'ees. Dans cette sous-section, nous traitons quelques exemples.
\subsection*{$\textbf{U}_{n^{\otimes3}}(\sigma)$}

1) $\sigma=(1\;2\;3)$\\
En utilisant le Lemme \ref{thm33}, $\sigma=(1\;2)(2\;3)$, et par utilisation de la proposition \ref{thm34}
\begin{equation}\label{e51}
\textbf{U}_{n^{\otimes3}}\left((1\;2\;3)\right)=\textbf{U}_{n^{\otimes3}}\left((1\;2)\right)\cdot
\textbf{U}_{n^{\otimes3}}\left((2\;3)\right)
\end{equation}
 Nous pouvons v\'erifier facilement que 
\begin{equation*}
 \textbf{U}_{n^{\otimes3}}\left((1\;2)\right)=\;\frac{1}{n}I_{n}\otimes
I_{n}\otimes
I_{n}+\frac{1}{2}\displaystyle\sum_{a=1}^{n^{2}-1}\Lambda_{a}\otimes\Lambda_{a}\otimes
I_{n}
\end{equation*}
et
\begin{equation*}
 \textbf{U}_{n^{\otimes3}}\left((2\;3)\right)=\;\frac{1}{n}I_{n}\otimes
I_{n}\otimes
I_{n}+\frac{1}{2}\displaystyle\sum_{a=1}^{n^{2}-1}I_{n}\otimes\Lambda_{a}\otimes\Lambda_{a}
\end{equation*}
Ainsi, (\ref{e51}) devient
\begin{multline*}
\textbf{U}_{n^{\otimes3}}\left((1\;2\;3)\right)=\frac{1}{n^{2}}I_{n}\otimes
I_{n}\otimes
I_{n}+\frac{1}{2n}\displaystyle\sum_{a=1}^{n^{2}-1}I_{n}\otimes\Lambda_{a}\otimes\Lambda_{a}\\
+\frac{1}{2n}\displaystyle\sum_{a=1}^{n^{2}-1}\Lambda_{a}\otimes\Lambda_{a}\otimes
I_{n}+\frac{1}{4}\displaystyle\sum_{a=1}^{n^{2}-1}\displaystyle\sum_{b=1}^{n^{2}-1}
\Lambda_{a}\otimes\Lambda_{a}\Lambda_{b}\otimes\Lambda_{b}
\end{multline*}
D'o\`u, \`a l'aide de la relation (\ref{e43})
\begin{multline}\label{e52}
\textbf{U}_{n^{\otimes3}}\left((1\;2\;3)\right)=\frac{1}{n^{2}}I_{n}\otimes
I_{n}\otimes I_{n}+\frac{1}{2n}\displaystyle\sum_{a=1}^{n^{2}-1}I_{n}\otimes\Lambda_{a}\otimes\Lambda_{a}\\
+\frac{1}{2n}\displaystyle\sum_{a=1}^{n^{2}-1}\Lambda_{a}\otimes\Lambda_{a}\otimes
I_{n}+\frac{1}{2n}\displaystyle\sum_{a=1}^{n^{2}-1}\Lambda_{a}\otimes
I_{n}\otimes\Lambda_{a}\\
-\frac{i}{4}\displaystyle\sum_{a=1}^{n^{2}-1}\displaystyle\sum_{b=1}^{n^{2}-1}
\displaystyle\sum_{c=1}^{n^{2}-1}f_{abc}\Lambda_{a}\otimes\Lambda_{b}\otimes\Lambda_{c}
+\frac{1}{4}\displaystyle\sum_{a=1}^{n^{2}-1}\displaystyle\sum_{b=1}^{n^{2}-1}
\displaystyle\sum_{c=1}^{n^{2}-1}d_{abc}\Lambda_{a}\otimes\Lambda_{b}\otimes\Lambda_{c}
\end{multline}
2) $\sigma=(1\;3\;2)$\\
A l'aide du Lemme \ref{thm33}, $\sigma=(1\;3)(3\;2)$, et de la proposition \ref{thm21}, nous avons 
\begin{equation*}
\textbf{U}_{n^{\otimes3}}\left((1\;3)\right)=\frac{1}{n}I_{n}\otimes
I_{n}\otimes
I_{n}+\frac{1}{2}\displaystyle\sum_{a=1}^{n^{2}-1}\Lambda_{a}\otimes
I_{n}\otimes\Lambda_{a}
\end{equation*}
et par utilisation de la m\^eme m\'ethode 

\begin{multline*}
\textbf{U}_{n^{\otimes3}}\left((1\;3\;2)\right)=\frac{1}{n^{2}}I_{n}\otimes
I_{n}\otimes I_{n}+\frac{1}{2n}\displaystyle\sum_{a=1}^{n^{2}-1}I_{n}\otimes\Lambda_{a}\otimes\Lambda_{a}\\
+\frac{1}{2n}\displaystyle\sum_{a=1}^{n^{2}-1}\Lambda_{a}\otimes\Lambda_{a}\otimes
I_{n}+\frac{1}{2n}\displaystyle\sum_{a=1}^{n^{2}-1}\Lambda_{a}\otimes
I_{n}\otimes\Lambda_{a}\\
+\frac{i}{4}\displaystyle\sum_{a=1}^{n^{2}-1}\displaystyle\sum_{b=1}^{n^{2}-1}
\displaystyle\sum_{c=1}^{n^{2}-1}f_{abc}\Lambda_{a}\otimes\Lambda_{b}\otimes\Lambda_{c}
+\frac{1}{4}\displaystyle\sum_{a=1}^{n^{2}-1}\displaystyle\sum_{b=1}^{n^{2}-1}
\displaystyle\sum_{c=1}^{n^{2}-1}d_{abc}\Lambda_{a}\otimes\Lambda_{b}\otimes\Lambda_{c}
\end{multline*}
La diff\'erence entre $\textbf{U}_{n^{\otimes3}}\left((1\;2\;3)\right)$ et $\textbf{U}_{n^{\otimes3}}\left((1\;3\;2)\right)$ est le signe moins devant le cinqui\`eme terme.
\subsection*{$\textbf{U}_{n^{\otimes4}}(\sigma)$, $\sigma=(1\;2\;3\;4)$}
A l'aide du Lemme \ref{thm33}, $\sigma=(1\;2\;3)(3\;4)$, de la formule (\ref{e52}), de la proposition \ref{thm34} et des relations (\ref{e43}) et (\ref{e44}), nous avons

\begin{equation*}
\begin{split}
\textbf{U}_{n^{\otimes4}}(\sigma)&=\frac{1}{n^{3}}I_{n}\otimes I_{n}\otimes
I_{n}\otimes I_{n}
+\frac{1}{2n^{2}}\displaystyle\sum_{a=1}^{n^{2}-1}I_{n}\otimes\Lambda_{a}\otimes\Lambda_{a}\otimes
I_{n}
+\frac{1}{2n^{2}}\displaystyle\sum_{a=1}^{n^{2}-1}\Lambda_{a}\otimes\Lambda_{a}\otimes
I_{n}\otimes I_{n}\\
&+\frac{1}{2n^{2}}\displaystyle\sum_{a=1}^{n^{2}-1}\Lambda_{a}\otimes
I_{n}\otimes\Lambda_{a}\otimes
I_{n}+\frac{1}{2n^{2}}\displaystyle\sum_{a=1}^{n^{2}-1}I_{n}\otimes\ I_{n}\otimes\Lambda_{a}\otimes\Lambda_{a}\\
&+\frac{1}{2n^{2}}\displaystyle\sum_{a=1}^{n^{2}-1}I_{n}\otimes\Lambda_{a}\otimes
I_{n}\otimes\Lambda_{a}+\frac{1}{2n^{2}}\displaystyle\sum_{a=1}^{n^{2}-1}\Lambda_{a}\otimes
I_{n}\otimes I_{n}\otimes\Lambda_{a}\\
&+\frac{1}{4n}\displaystyle\sum_{a=1}^{n^{2}-1}\displaystyle\sum_{b=1}^{n^{2}-1}
\Lambda_{a}\otimes\Lambda_{a}\otimes\Lambda_{b}\otimes\Lambda_{b}+\frac{1}{4n}\displaystyle\sum_{a=1}^{n^{2}-1}\displaystyle\sum_{b=1}^{n^{2}-1}
\Lambda_{a}\otimes\Lambda_{b}\otimes\Lambda_{b}\otimes\Lambda_{a}\\
&-\frac{1}{4n}\displaystyle\sum_{a=1}^{n^{2}-1}\displaystyle\sum_{b=1}^{n^{2}-1}
\Lambda_{a}\otimes\Lambda_{b}\otimes\Lambda_{a}\otimes\Lambda_{b}\\
&+\frac{1}{4n}\displaystyle\sum_{a=1}^{n^{2}-1}\displaystyle\sum_{b=1}^{n^{2}-1}
\displaystyle\sum_{c=1}^{n^{2}-1}d_{abc}I_{n}\otimes\Lambda_{a}\otimes\Lambda_{b}\otimes\Lambda_{c}
-\frac{i}{4n}\displaystyle\sum_{a=1}^{n^{2}-1}\displaystyle\sum_{b=1}^{n^{2}-1}
\displaystyle\sum_{c=1}^{n^{2}-1}f_{abc}I_{n}\otimes\Lambda_{a}\otimes\Lambda_{b}\otimes\Lambda_{c}\\
&+\frac{1}{4n}\displaystyle\sum_{a=1}^{n^{2}-1}\displaystyle\sum_{b=1}^{n^{2}-1}
\displaystyle\sum_{c=1}^{n^{2}-1}d_{abc}\Lambda_{a}\otimes
I_{n}\otimes\Lambda_{b}\otimes\Lambda_{c}
-\frac{i}{4n}\displaystyle\sum_{a=1}^{n^{2}-1}\displaystyle\sum_{b=1}^{n^{2}-1}
\displaystyle\sum_{c=1}^{n^{2}-1}f_{abc}\Lambda_{a}\otimes I_{n}\otimes\Lambda_{b}\otimes\Lambda_{c}\\
&+\frac{1}{4n}\displaystyle\sum_{a=1}^{n^{2}-1}\displaystyle\sum_{b=1}^{n^{2}-1}
\displaystyle\sum_{c=1}^{n^{2}-1}d_{abc}\Lambda_{a}\otimes\Lambda_{b}\otimes\Lambda_{c}\otimes
I_{n}
-\frac{i}{4n}\displaystyle\sum_{a=1}^{n^{2}-1}\displaystyle\sum_{b=1}^{n^{2}-1}
\displaystyle\sum_{c=1}^{n^{2}-1}f_{abc}\Lambda_{a}\otimes\Lambda_{b}\otimes\Lambda_{c}\otimes
I_{n}\\
&+\frac{1}{4n}\displaystyle\sum_{a=1}^{n^{2}-1}\displaystyle\sum_{b=1}^{n^{2}-1}
\displaystyle\sum_{c=1}^{n^{2}-1}d_{abc}\Lambda_{a}\otimes\Lambda_{b}\otimes
I_{n}\otimes\Lambda_{c}
-\frac{i}{4n}\displaystyle\sum_{a=1}^{n^{2}-1}\displaystyle\sum_{b=1}^{n^{2}-1}
\displaystyle\sum_{c=1}^{n^{2}-1}f_{abc}\Lambda_{a}\otimes\Lambda_{b}\otimes
I_{n}\otimes\Lambda_{c}\\
&+\frac{1}{8}\displaystyle\sum_{a=1}^{n^{2}-1}\displaystyle\sum_{b=1}^{n^{2}-1}\displaystyle\sum_{c=1}^{n^{2}-1}
\displaystyle\sum_{e=1}^{n^{2}-1}\displaystyle\sum_{g=1}^{n^{2}-1}\left(-if_{abc}d_{ceg}+id_{abc}f_{ceg}+d_{aec}d_{bgc}
-d_{agc}d_{cbe}+d_{abc}d_{ceg}\right)\\
&\Lambda_{a}\otimes\Lambda_{b}\otimes\Lambda_{g}\otimes\Lambda_{e}\\
\end{split}
\end{equation*}

\subsection*{$\textbf{U}_{2^{\otimes3}}(\sigma)$, $\sigma\in S_{3}$ }
Maintenant, nous donnons la formule donnant $\textbf{U}_{2^{\otimes3}}(\sigma)$, en termes des matrices de Pauli, bien attendu.
 En utilisant la relation (Cf. par exemple \cite{Raoelina72})
\begin{equation*}
\sigma_{l}\sigma_{k}=\delta_{lk}I_{2}+i\displaystyle\sum_{m=1}^{3}\varepsilon_{lkm}\sigma_{m}
\end{equation*}
o\`u $\varepsilon_{ijk}$ est totalement antisym\'etrique, qui est \'egal \`a 1 si $(i\;\;j\;\;k)=(1\;\;2\;\;3)$, nous avons 
\begin{equation*}
\begin{split}
\textbf{U}_{2^{\otimes3}}(1\;\;2\;\;3)&=\frac{1}{4}I_{2}\otimes I_{2}\otimes
I_{2}
+\frac{1}{4}\displaystyle\sum_{l=1}^{3}I_{2}\otimes\sigma_{l}\otimes\sigma_{l}
+\frac{1}{4}\displaystyle\sum_{l=1}^{3}\sigma_{l}\otimes
I_{2}\otimes\sigma_{l}\\
&+\frac{1}{4}\displaystyle\sum_{l=1}^{3}\sigma_{l}\otimes\sigma_{l}\otimes
I_{2}-\frac{i}{4}\displaystyle\sum_{i=1}^{3}\displaystyle\sum_{j=1}^{3}\displaystyle\sum_{k=1}^{3}
\varepsilon_{ijk}\sigma_{i}\otimes\sigma_{j}\otimes\sigma_{k}
\end{split}
\end{equation*}
et
\begin{equation*}
\begin{split}
\textbf{U}_{2^{\otimes3}}(1\;\;3\;\;2)&=\frac{1}{4}I_{2}\otimes I_{2}\otimes
I_{2}
+\frac{1}{4}\displaystyle\sum_{l=1}^{3}I_{2}\otimes\sigma_{l}\otimes\sigma_{l}
+\frac{1}{4}\displaystyle\sum_{l=1}^{3}\sigma_{l}\otimes
I_{2}\otimes\sigma_{l}\\
&+\frac{1}{4}\displaystyle\sum_{l=1}^{3}\sigma_{l}\otimes\sigma_{l}\otimes
I_{2}+\frac{i}{4}\displaystyle\sum_{i=1}^{3}\displaystyle\sum_{j=1}^{3}\displaystyle\sum_{k=1}^{3}
\varepsilon_{ijk}\sigma_{i}\otimes\sigma_{j}\otimes\sigma_{k}
\end{split}
\end{equation*}
\section*{Conclusion}
Partant du fait qu'une MPT est un produit de matrices de transposition tensorielle, le th\'eor\`eme \ref{thm21} et avec
l'aide de l'expression d'une MCT en termes des matrices de Gell-Mann g\'en\'eralis\'ees, nous pouvons exprimer une MPT comme combinaison linéaire des produits tensoriels des matrices de Gell-Mann g\'en\'eralis\'ees.\\ 
\indent Nous n'avons pas l'intention de chercher une formule g\'en\'erale. Cependant, nous avons montr\'e que toute MPT peut \^etre exprim\'e en termes de matrices de Gell-Mann g\'en\'eralis\'ees et puis l'expression peut \^etre simplifi\'ee en utilisant les relations entre ces matrices.
\section{MCT ET CHARGES ELECTRIQUES DES FERMIONS}
Les fermions ont les nombres quantiques $I_3$, l'isospin et $Y$, l'hypercharge. La charge \'electrique $Q$ d'un fermion est donn\'ee par la relation de Gell-Mann-Nishijima 

\begin{equation}
Q=I_3+\frac{Y}{2}
\end{equation}

Pour les fermions du mod\`ele standard (MS), ces nombres quantiques sont donn\'es par le tableau suivant.

\begin{center}	
	\begin{tabular}{lr|r|rrr}
		  &   &   & $Q$ & $I_3$ & $Y$\\
		\hline
		 &  Leptons Neutres & $\nu_{eL}$, $\nu_{\mu L}$, $\nu_{\tau L}$ & $0$ & $1/2$ & $-1$\\
		 
		\hline
	 &  Leptons Charg\'es & $e_L$, $\mu_L$, $\tau_L$ & $-1$ & $-1/2$ & $-1$\\
		 &  & $e_R$, $\mu_R$, $\tau_R$ & $-1$ & $0$ & $-2$ \\
		\hline
	 & Quarks $u$, $c$, $t$ & $u^{r}_L$, $u^{b}_L$, $u^{g}_L$, $c^{r}_L$, $c^{b}_L$, $c^{g}_L$, $t^{r}_L$, $t^{b}_L$, $t^{g}_L$ & $2/3$ & $1/2$ & $1/3$\\
		 &  & $u^{r}_R$, $u^{b}_R$, $u^{g}_R$, $c^{r}_R$, $c^{b}_R$, $c^{g}_R$, $t^{r}_R$, $t^{b}_R$, $t^{g}_R$ & $2/3$ & $0$ & $4/3$\\
		
	 \hline
	 & Quarks $d$, $s$, $b$ & $d^{r}_L$, $d^{b}_L$, $d^{g}_L$, $s^{r}_L$, $s^{b}_L$, $s^{g}_L$, $b^{r}_L$, $b^{b}_L$, $b^{g}_L$ & $-1/3$ & $-1/2$ & $1/3$\\
		 &  & $d^{r}_R$, $d^{b}_R$, $d^{g}_R$, $s^{r}_R$, $s^{b}_R$, $s^{g}_R$, $b^{r}_R$, $b^{b}_R$, $b^{g}_R$ & $-1/3$ & $0$ & $-2/3$\\
	\hline
						
		\end{tabular}						
\end{center}
 
Une relation matricielle de Gell-Mann-Nishijima  pour huit leptons et quarks du MS de la m\^eme generation a \'et\'e propos\'ee par \cite{Zenczykowski07}, dans la formulation dans l'espace des phases. Selon la formule (\ref{eq5prime}) il est facile de remarquer que cette relation matricielle de Gell-Mann-Nishijima peut \^etre exprim\'ee en termes de $\textbf{U}_{2\otimes2}$. Dans cette section, nous allons \'ecrire cette relation en certaines formes dont une interp\'etation physique de l'action de MCT $\textbf{U}_{2\otimes2}$ sera possible. Alors, les valeur propres et les vecteur propres de $\textbf{U}_{2\otimes2}$ prendront des sens physiques. Dans la sous-section \ref{subsec423}, pour inclure plus de fermions du MS nous \'ecrirons une formule matricielle donnant les charges \'electriques en termes de la MCT $3\otimes3$,

\begin{equation*}
\textbf{U}_{3\otimes3}=
\begin{pmatrix}
  1 & 0 & 0 & 0 & 0 & 0 & 0 & 0 & 0 \\
  0 & 0 & 0 & 1 & 0 & 0 & 0 & 0 & 0 \\
  0 & 0 & 0 & 0 & 0 & 0 & 1 & 0 & 0 \\
  0 & 1 & 0 & 0 & 0 & 0 & 0 & 0 & 0 \\
  0 & 0 & 0 & 0 & 1 & 0 & 0 & 0 & 0 \\
  0 & 0 & 0 & 0 & 0 & 0 & 0 & 1 & 0 \\
  0 & 0 & 1 & 0 & 0 & 0 & 0 & 0 & 0 \\
  0 & 0 & 0 & 0 & 0 & 1 & 0 & 0 & 0 \\
  0 & 0 & 0 & 0 & 0 & 0 & 0 & 0 & 1 \\
\end{pmatrix}
\end{equation*}
Alors le sens physique donn\'e aux valeurs propres de la MCT $\textbf{U}_{2\otimes2}$ est maintenu. Dans la sous-section \ref{subsec424}, pour inclure tous les fermions du MS nous \'ecrirons la matrice des charges \'electriques en termes de  MCT $\textbf{U}_{4\otimes4}$. Pour les calculs nous avons utilis\'e SCILAB, un logiciel libre pour l'analyse numerique.

\subsection{Relation de Gell-Mann-Nishijima dans la formulation dans l'espace des phases}
Nous re\'ecrivons ici la relation de Gell-Mann-Nishijima, qui donne l'OCE des huit fermions, deux leptons et six quarks color\'es, d'une seule g\'en\'eration du MS, par exemple $e_L$, $\nu_{eL}$, $u^{r}_L$, $u^{b}_L$, $u^{g}_L$, $d^{r}_L$, $d^{b}_L$, et $d^{g}_L$, dans l'approche dans l'espace des phases \cite{Zenczykowski07}. 

\begin{equation}\label{eq6}
\textbf{Q}=\textbf{I}_3+\frac{\textbf{Y}}{2}
\end{equation} 
o\'u 
\begin{equation*}
\textbf{I}_3=\frac{1}{2}\sigma_0\otimes\sigma_0\otimes\sigma_3 
\end{equation*} 
l'op\'erateur isospin , 
\begin{equation*}
\textbf{Y}=\left(\frac{1}{3}\sum_{i=1}^3\sigma_i\otimes \sigma_i\right)\otimes \sigma_0
\end{equation*} 
l'op\'erateur hypercharge.\\

Nous remarquons que les op\'erateurs $\textbf{I}_3$ et $\textbf{Y}$ agissent ind\'ependamment sur le champ de vecteur puisque dans l'expression de $\textbf{I}_3$, $\sigma_3$ est au c\^ot\'e droit du produit tensoriel et dans l'expression de $\textbf{Y}$,  $\left(\frac{1}{3}\displaystyle\sum_{i=1}^3\sigma_i\otimes \sigma_i\right)$ est au c\^ot\'e gauche.

\subsection{Op\'erateur charges \'electriques  en termes de la MCT $2\otimes 2$}

Selon la formule (\ref{eq5prime}) nous pouvons introduire l'op\'erateur $\textbf{U}_{2\otimes2}$ dans l'expression de l'op\'erateur hypercharge $\textbf{Y}$, et alors dans l'expression de l'OCE $\textbf{Q}$.

\begin{equation*}
\textbf{Y}=\frac{2}{3}\textbf{U}_{2\otimes2}\otimes \sigma_0-\frac{1}{3}\sigma_0\otimes\sigma_0\otimes\sigma_0
\end{equation*}
D'o\`u
\begin{equation}\label{eq10}
\textbf{Q}=\frac{1}{2}\sigma_0\otimes\sigma_0\otimes\sigma_3+\frac{1}{3}\left(\textbf{U}_{2\otimes2}\otimes \sigma_0-\frac{1}{2}\sigma_0\otimes\sigma_0\otimes\sigma_0\right)
\end{equation}
ou
\begin{equation}
\textbf{Q}=\sigma_0\otimes\sigma_0\otimes\left(\frac{\sigma_3}{2}+\alpha\frac{\sigma_0}{6}\right)+\frac{1}{3}\left(\textbf{U}_{2\otimes2} -\frac{1+\alpha}{2}\sigma_0\otimes\sigma_0\right)\otimes\sigma_0
\end{equation}
avec $\alpha$ un param\`etre r\'eel.

Si $\alpha=0$, nous avons la relation (\ref{eq10}), c'est-\`a-dire (\ref{eq6}).\\

Si $\alpha=-3$,
\begin{equation}\label{eq12}
\textbf{Q}=\sigma_0\otimes\sigma_0\otimes \textbf{Q}_L+\frac{1}{3}\left(\textbf{U}_{2\otimes2}+\sigma_0\otimes\sigma_0\right)\otimes\sigma_0
\end{equation}
 
\noindent avec $\textbf{Q}_L=\begin{pmatrix}
  0 & 0 \\
  0 & -1 \\
\end{pmatrix}$ dont la diagonale est form\'ee par les charges \'electriques de $\nu_{eL}$ et de $e_L$.\\

Si $\alpha=1$, 
\begin{equation}\label{eq13}
\textbf{Q}=\sigma_0\otimes\sigma_0\otimes \textbf{Q}_Q+\frac{1}{3}\left(\textbf{U}_{2\otimes2}-\sigma_0\otimes\sigma_0\right)\otimes\sigma_0
\end{equation} 
avec $\textbf{Q}_Q=\begin{pmatrix}
  2/3 & 0 \\
  0 & -1/3 \\
\end{pmatrix}$ dont la diagonale est form\'ee par les charges \'electriques d'un quark $u$ et d'un quark $d$.\\

Les quatre valeurs propres de $\textbf{U}_{2\otimes2}$ sont une fois $-1$ et trois fois $+1$.\\

$\frac{1}{\sqrt{2}}\begin{pmatrix}
  0 \\
  -1 \\
  1\\
  0
\end{pmatrix}$ un vecteur propre de $\textbf{U}_{2\otimes2}$ associ\'e \`a la valeur propre $-1$, qui est, selon (\ref{eq12}), associ\'e aux leptons. \\
$\begin{pmatrix}
  1 \\
  0 \\
  0\\
  0
\end{pmatrix}$, $\frac{1}{\sqrt{2}}\begin{pmatrix}
  0 \\
  1 \\
  1\\
  0
\end{pmatrix}$, $\begin{pmatrix}
  0 \\
  0 \\
  0\\
  1
\end{pmatrix}$ sont les vecturs propres associ\'es \`a $+1$, qui est, selon (\ref{eq13}), associ\'es aux trois quarks color\'es.\\
La diagonale de $\textbf{Q}_L=\begin{pmatrix}
  0 & 0 \\
  0 & -1 \\
\end{pmatrix}$ sont form\'ee par les charges de $e_L$ et $\nu_{eL}$. La diagonale de $\textbf{Q}_Q=\begin{pmatrix}
  2/3 & 0 \\
  0 & -1/3 \\
\end{pmatrix}$ sont form\'ee par les charges de quark $u$ (up) et de quark $d$ (down). Le nombre de valeur propre $-1$ de  $\textbf{U}_{2\otimes2}$ est le nombre de g\'en\'eration de leptons. $+1$ trois fois valeur propre de $\textbf{U}_{2\otimes2}$, c'est-\`a-dire le nombre de couleurs. D'o\`u, les valeurs propres de $\textbf{Q}$ sont les charges des huit fermions mention\'es ci-dessus.  

\subsection{Op\'erateur charges \'electriques  en termes de MCT $3\otimes3$}\label{subsec423}
Soit, par exemple, $\textbf{Q}_L=\begin{pmatrix}
  0 & 0 & 0  \\
  0 & -1 & 0  \\
  0 & 0 & -1  \\
\end{pmatrix}$ dont la diagonale est form\'ee par la charge \'electrique d'un neutrino et les charges \'electriques de deux leptons charg\'es  et $\textbf{Q}_Q=\begin{pmatrix}
  2/3 & 0 & 0  \\
  0 & -1/3 & 0  \\
  0 & 0 & -1/3  \\
\end{pmatrix}$ 
dont la diagonale est form\'ee par la charge \'electrique d'un quark u (un quark c (charm) ou d'un quark t (top)) et les charges \'electriques d'un quarks d, d'un quarks s (strange) ou quarks b (bottom).
\begin{equation}\label{eq13prime}
\textbf{Q}_Q-\textbf{Q}_L=\frac{2}{3}\lambda_0
\end{equation}
o\`u $\lambda_0$ est la matrice unitaire $3\times 3$.\\
D'o\`u, 
\begin{equation*}
\lambda_0\otimes \lambda_0\otimes \textbf{Q}_Q + \frac{1}{3}(\textbf{U}_{3\otimes 3}-\lambda_0\otimes \lambda_0)\otimes \lambda_0=\lambda_0\otimes \lambda_0\otimes \textbf{Q}_L + \frac{1}{3}(\textbf{U}_{3\otimes 3}+\lambda_0\otimes \lambda_0)\otimes \lambda_0 
\end{equation*}
Nous notons cette expression $\textbf{Q}$, comme un OCE. \\
Les valeurs propres de la MCT $\textbf{U}_{3\otimes 3}$ sont $-1$ trois fois et $+1$ six fois. D'apr\`es l'\'equation ci-dessus les valeurs propres $-1$ sont associ\'ees aux leptons tandis que les valeurs propres $+1$ sont associ\'ees aux quarks. Les trois valeurs propres $-1$ sont associ\'ees aux trois g\'en\'erations de leptons, tandis que les trois valeurs propres $+1$ sont associ\'ees aux trois couleurs de quarks gauchers et les trois autres sont associ\'ees aux trois couleurs de quarks droitiers. La diagonale de $\textbf{Q}_L$ sont form\'es par les charges des leptons du MS dans une m\^eme g\'en\'eration, par exemple $\nu_{eL}$, $e_L$ et $e_R$. La diagonale de $\textbf{Q}_Q$ sont form\'es par les charges d'un quark $u$, d'un quark $d$ et d'un quark $s$. D'o\`u, les vingt sept valeurs propres de $\textbf{Q}$, $-1$ six fois, $0$ trois fois, $-1/3$ douze fois et $+2/3$ six fois, peuvent \^etre les charges des fermions du MS suivants: $\nu_{eL}$, $\nu_{\mu L}$, $\nu_{\tau L}$,  $e_L$, $\mu_L$, $\tau_L$, $e_R$, $\mu_R$, $\tau_R$,  $u^{r}_L$, $u^{b}_L$, $u^{g}_L$, $u^{r}_R$, $u^{b}_R$, $u^{g}_R$,  $d^{r}_L$, $d^{b}_L$, $d^{g}_L$, $d^{r}_R$, $d^{b}_R$, $d^{g}_R$, $s^{r}_L$, $s^{b}_L$, $s^{g}_L$, $s^{r}_R$, $s^{b}_R$, $s^{g}_R$. 
\begin{equation*}
\textbf{Q}_Q=\frac{1}{2}\lambda_3+\frac{1}{2\sqrt{3}}\lambda_8 
\end{equation*}

D'apr\`es la relation (\ref{e44prime}) pour $n=3$, la MCT $\textbf{U}_{3\otimes 3}$ peut s'\'ecrire en termes de matrices de Gell-Mann de la fa\c con suivante 
\begin{equation*}
\textbf{U}_{3\otimes 3}=\frac{1}{3}\lambda_0\otimes \lambda_0+\frac{1}{2}\sum_{i=1}^8\lambda_i\otimes\lambda_i 
\end{equation*}

Ainsi,
\begin{equation*}
\textbf{Q}=\lambda_0\otimes \lambda_0\otimes \left(-\frac{2}{9}\lambda_0+\frac{1}{2}\lambda_3+\frac{1}{2\sqrt{3}}\lambda_8\right)+\frac{1}{6}\left(\sum_{i=1}^8\lambda_i\otimes\lambda_i\right)\otimes \lambda_0
\end{equation*}

$\textbf{Q}$ peut s'\'ecrire sous la forme de la relation (\ref{eq6}), o\`u 
\begin{equation*}
 \textbf{I}_3=\frac{1}{2}\left(\lambda_0\otimes\lambda_0\otimes\lambda_3-\tau_1\otimes\tau_1\otimes\tau_1-\tau_2\otimes\tau_2\otimes\tau_1-\tau_3\otimes\tau_3\otimes\tau_1\right)
 \end{equation*}

\begin{equation*}
 \textbf{Y}=\tau_1\otimes\tau_1\otimes\tau_1+\tau_2\otimes\tau_2\otimes\tau_1+\tau_3\otimes\tau_3\otimes\tau_1+\frac{1}{\sqrt{3}}\lambda_0\otimes\lambda_0\otimes\lambda_8+\frac{2}{3}\left(U_{3\otimes 3}-\lambda_0\otimes\lambda_0\right)\otimes\lambda_0
 \end{equation*}
 avec $\tau_1=\begin{pmatrix}
  1 & 0 & 0  \\
  0 & 0 & 0  \\
  0 & 0 & 0  \\
\end{pmatrix}$, $\tau_2=\begin{pmatrix}
  0 & 0 & 0  \\
  0 & 1 & 0  \\
  0 & 0 & 0  \\
\end{pmatrix}$ et $\tau_3=\begin{pmatrix}
  0 & 0 & 0  \\
  0 & 0 & 0  \\
  0 & 0 & 1  \\
\end{pmatrix}$.\\

 D'apr\`es (\ref{eq33}), $\textbf{I}_3$ et $\textbf{Y}$ commutent, donc ils sont simultan\'ement diagonalisables.\\
 
 Si nous prenons $\textbf{Q}_L=\begin{pmatrix}
  0 & 0 & 0  \\
  0 & 0 & 0  \\
  0 & 0 & -1  \\
\end{pmatrix}$ et $\textbf{Q}_Q=\begin{pmatrix}
  2/3 & 0 & 0  \\
  0 & 2/3 & 0  \\
  0 & 0 & -1/3  \\
\end{pmatrix}$ 
la relation ci-dessus entre $\textbf{Q}_L$ et $\textbf{Q}_Q$ sera maintenue. Les vingt sept valeurs propres de l'OCE $\textbf{Q}$ sont $-1$ trois fois, $0$ six fois, $-1/3$ six fois et $+2/3$ douze fois. Ces valeurs propres peuvent \^etre les charges des fermions suivants: $\nu_{eL}$, $\nu_{\mu L}$, $\nu_{\tau L}$,  $\nu_{eR}$, $\nu_{\mu R}$, $\nu_{\tau R}$, $e_L$, $\mu_L$, $\tau_L$, $c^{r}_L$, $c^{b}_L$, $c^{g}_L$, $c^{r}_R$, $c^{b}_R$, $c^{g}_R$,  $t^{r}_L$, $t^{b}_L$, $t^{g}_L$, $t^{r}_R$, $t^{b}_R$, $t^{g}_R$, $b^{r}_L$, $b^{b}_L$, $b^{g}_L$, $b^{r}_R$, $b^{b}_R$, $b^{g}_R$. Les neutrinos gauchers  $\nu_{eR}$, $\nu_{\mu R}$, $\nu_{\tau R}$ dont la charge est $0$ ne sont pas des fermions du MS.

\subsection{Inclure tous les fermions du mod\`ele standard}\label{subsec424}
Pour inclure tous les fermions du MS, nous allons construire un OCE en termes de la MCT $\textbf{U}_{4\otimes 4}$.\\
Les valeurs propres de la MCT $\textbf{U}_{4\otimes 4}$ sont $-1$ six fois et $+1$ dix fois. 
Donc, l'OCE   
\begin{equation*}
\textbf{Q}=\Lambda_0\otimes \Lambda_0\otimes \textbf{Q}_Q + \frac{1}{3}(U_{4\otimes 4}-\Lambda_0\otimes \Lambda_0)\otimes \Lambda_0=\Lambda_0\otimes \Lambda_0\otimes \textbf{Q}_L + \frac{1}{3}(U_{4\otimes 4}+\Lambda_0\otimes \Lambda_0)\otimes \Lambda_0 
\end{equation*}
o\`u $\Lambda_0$ est la matrice unit\'e $4\times 4$ et $\textbf{Q}_Q=\begin{pmatrix} 
  -1/3 & 0 & 0 & 0 \\
  0 & 2/3 & 0 & 0 \\
  0 & 0 & -1/3 & 0 \\
  0 & 0 & 0 & 2/3 \\
\end{pmatrix}$, $\textbf{Q}_L=\begin{pmatrix} 
  -1 & 0 & 0 & 0 \\
  0 & 0 & 0 & 0 \\
  0 & 0 & -1 & 0 \\
  0 & 0 & 0 & 0 \\
\end{pmatrix}$. $\textbf{Q}$ a soixante quatre valeurs propres.\\
La diagonale de $\textbf{Q}_L$ est form\'ee par les charges des quatre leptons d'une m\^eme g\'en\'eration, par exemple $e_L$, $\nu_{eL}$, $e_R$ et $\nu_{eR}$, tandis que la diagonale de $\textbf{Q}_Q$ est form\'ee par les charges d'un quark $u$, d'un quark $d$, d'un quark $s$ (ou d'un quark $b$) et quark $c$ (ou d'un quark $t$). Pour la MCT $\textbf{U}_{4\otimes 4}$, les quatre valeurs propres $-1$ sont associ\'ees aux quatre g\'en\'erations de leptons, les trois valeurs propres $+1$ et une valeur propre $-1$, sont respectivement associ\'ees aux trois couleurs de quarks gauchers et le lepton qui forment ensemble un leptoquark gaucher du mod\`ele de Pati-Salam \cite{ref5}

 \begin{equation*}
\begin{pmatrix}
  u_L^r & u_L^b & u_L^g & \nu_{eL} \\
  d_L^r & d_L^b & d_L^g & e_L \\
  \end{pmatrix}=\begin{pmatrix}
  u_L^r & u_L^b & u_L^g & u_L^w \\
  d_L^r & d_L^b & d_L^g & d_L^w \\
  \end{pmatrix},
  \begin{pmatrix}
  s_L^r & s_L^b & s_L^g & \nu_{\mu L} \\
  c_L^r & c_L^b & c_L^g & \mu_L \\
  \end{pmatrix}=\begin{pmatrix}
  s_L^r & s_L^b & s_L^g & s_L^w \\
  c_L^r & c_L^b & c_L^g & c_L^w \\
  \end{pmatrix} 
\end{equation*}
tandis que les autres trois valeurs propres $+1$ et une valeur propre $-1$, sont respectivement associ\'ees aux trois couleurs de quarks droitiers et le lepton qui forment ensemble un leptoquark droitier du mod\`ele de Pati-Salam

\begin{equation*}
\begin{pmatrix}
  u_R^r & u_R^b & u_R^g & \nu_{eR} \\
  d_R^r & d_R^b & d_R^g & e_R \\
  \end{pmatrix}=\begin{pmatrix}
  u_R^r & u_R^b & u_R^g & u_R^w \\
  d_R^r & d_R^b & d_R^g & d_R^w \\
  \end{pmatrix},
  \begin{pmatrix}
  s_R^r & s_R^b & s_R^g & \nu_{\mu R} \\
  c_R^r & c_R^b & c_R^g & \mu_R \\
  \end{pmatrix}=\begin{pmatrix}
  s_R^r & s_R^b & s_R^g & s_R^w \\
  c_R^r & c_R^b & c_R^g & c_R^w \\
  \end{pmatrix} 
\end{equation*}
  
Ainsi, selon \cite{ref6} nous avons consid\'er\'e les  leptons dans les leptoquarks comme des quarks de couleur blanche.

Finalement, les derni\`eres quatre valeurs propres $+1$ sont associ\'ees aux quatre couleurs de quarks de la troisi\`eme g\'en\'eration, plus un quark de couleur jaune, \`a savoir
\begin{equation*}
\begin{pmatrix}
  t_L^r & t_L^b & t_L^g & t_L^y \\
  b_L^r & b_L^b & b_L^g & b_L^y \\
  \end{pmatrix},  \begin{pmatrix}
  t_R^r & t_R^b & t_R^g & t_R^y \\
  b_R^r & b_R^b & b_R^g & b_R^y \\
  \end{pmatrix}, 
\end{equation*}
et \c ca termine la liste des soixante quatre fermions fondamentaux, avec des neutrinos droitiers y inclus.\\
$\Lambda_1=\begin{pmatrix}
  0 & 1 & 0 & 0  \\
  1 & 0 & 0 & 0  \\
  0 & 0 & 0 & 0  \\
  0 & 0 & 0 & 0  \\
\end{pmatrix}$, $\Lambda_2=\begin{pmatrix}
  0 & -i & 0 & 0  \\
  i & 0 & 0 & 0  \\
  0 & 0 & 0 & 0  \\
  0 & 0 & 0 & 0  \\
\end{pmatrix}$, $\Lambda_3=\begin{pmatrix}
  1 & 0 & 0 & 0  \\
  0 & -1 & 0 & 0  \\
  0 & 0 & 0 & 0  \\
  0 & 0 & 0 & 0  \\
\end{pmatrix}$,\\
 $\Lambda_4=\begin{pmatrix}
  0 & 0 & 1 & 0  \\
  0 & 0 & 0 & 0  \\
  1 & 0 & 0 & 0  \\
  0 & 0 & 0 & 0  \\
\end{pmatrix}$, 
$\Lambda_5=\begin{pmatrix}
  0 & 0 & -i & 0  \\
  0 & 0 & 0 & 0  \\
  i & 0 & 0 & 0  \\
  0 & 0 & 0 & 0  \\
\end{pmatrix}$, 
$\Lambda_6=\begin{pmatrix}
  0 & 0 & 0 & 0  \\
  0 & 0 & 1 & 0  \\
  0 & 1 & 0 & 0  \\
  0 & 0 & 0 & 0  \\
\end{pmatrix}$,\\
 $\Lambda_7=\begin{pmatrix}
  0 & 0 & 0 & 0  \\
  0 & 0 & -i & 0  \\
  0 & i & 0 & 0  \\
  0 & 0 & 0 & 0  \\
\end{pmatrix}$, 
$\Lambda_8=\frac{1}{\sqrt{3}}\begin{pmatrix}
  1 & 0 & 0 & 0  \\
  0 & 1 & 0 & 0  \\
  0 & 0 & -2 & 0  \\
  0 & 0 & 0 & 0  \\
\end{pmatrix}$ 
$\Lambda_9=\begin{pmatrix}
  0 & 0 & 0 & 1  \\
  0 & 0 & 0 & 0  \\
  0 & 0 & 0 & 0  \\
  1 & 0 & 0 & 0  \\
\end{pmatrix}$,\\ 
$\Lambda_{10}=\begin{pmatrix}
  0 & 0 & 0 & -i  \\
  0 & 0 & 0 & 0  \\
  0 & 0 & 0 & 0  \\
  i & 0 & 0 & 0  \\
\end{pmatrix}$, 
$\Lambda_{11}=\begin{pmatrix}
  0 & 0 & 0 & 0  \\
  0 & 0 & 0 & 1  \\
  0 & 0 & 0 & 0  \\
  0 & 1 & 0 & 0  \\
\end{pmatrix}$, 
$\Lambda_{12}=\begin{pmatrix}
  0 & 0 & 0 & 0  \\
  0 & 0 & 0 & -i  \\
  0 & 0 & 0 & 0  \\
  0 & i & 0 & 0  \\
\end{pmatrix}$, \\
$\Lambda_{13}=\begin{pmatrix}
  0 & 0 & 0 & 0  \\
  0 & 0 & 0 & 0  \\
  0 & 0 & 0 & 1  \\
  0 & 0 & 1 & 0  \\
\end{pmatrix}$, 
$\Lambda_{14}=\begin{pmatrix}
  0 & 0 & 0 & 0  \\
  0 & 0 & 0 & 0  \\
  0 & 0 & 0 & -i  \\
  0 & 0 & i & 0  \\
\end{pmatrix}$, 
$\Lambda_{15}=\frac{1}{\sqrt{6}}\begin{pmatrix}
  1 & 0 & 0 & 0  \\
  0 & 1 & 0 & 0  \\
  0 & 0 & 1 & 0  \\
  0 & 0 & 0 & -3  \\
\end{pmatrix}$ \\
sont les matrices de Gell-Mann $4\times 4$.\\
La formule (\ref{e44prime}) pour $n=4$,  

\begin{equation*}
\textbf{U}_{4\otimes 4}=\frac{1}{4}\Lambda_0\otimes \Lambda_0+\frac{1}{2}\sum_{i=1}^{15}\Lambda_i\otimes\Lambda_i 
\end{equation*}
donne
\begin{equation*}
\textbf{Q}=\Lambda_0\otimes \Lambda_0\otimes \left(-\frac{1}{12}\Lambda_0-\frac{1}{36}\Lambda_3+\frac{\sqrt{3}}{2}\Lambda_8-\frac{1}{6\sqrt{6}}\Lambda_{15}\right)+\frac{1}{6}\left(\sum_{i=1}^{15}\Lambda_i\otimes\Lambda_i\right)\otimes \Lambda_0
\end{equation*}

\subsection*{Conclusion}
Gr\^ace \`a \cite{Zenczykowski07}, nous avons un OCE pour deux leptons et six quarks color\'es  du MS d'une seule g\'en\'eration. Cet OCE peut \^etre exprim\'e en termes de la MCT $\textbf{U}_{2\otimes 2}$. Un OCE pour plus de fermions du MS en trois g\'en\'erations a \'et\'e obtenu en termes de la MCT $\textbf{U}_{3\otimes 3}$. \\
L'expression de ces OCE, ainsi que celui qui est propos\'e par \cite{Zenczykowski07}, peut \^etre obtenue \`a partir de la  relation (\ref{eq13prime}) entre les charges \'electriques des leptons et quarks. Ces expressions permettent de dire que les valeurs propres $-1$ d'une MCT sont associ\'ees aux leptons tandis que les valeurs propres $+1$ sont associ\'ees aux quarks.\\ 
D'apr\`es le sens que prend une valeur propre d'une MCT, pour obtenir un OCE pour tous les fermions du MS en termes de la MCT $\textbf{U}_{4\otimes 4}$, nous sommes oblig\'es d'introduire la quatri\`eme g\'en\'eration de leptons, le mod\`ele de  leptoquark de Pati-Salam et les quarks de couleur jaune.    

\chapter*{Conclusion et perspectives}
\noindent La MCT $2\otimes 2$ nous est apparue lorsque nous avons \'etudi\'e l'\'equation de Dirac et sa modification, l'\'equation de Dirac-Sidharth. Alors, nous avons construit deux ensembles dont chacun contient six repr\'esentations de l'\'equation de Dirac. L'un de ces ensembles contient la repr\'esentation de Dirac et celle de Weyl. Ils se transforment l'un en l'autre par application de la MCT $2\otimes 2$.\\
Dans le chapitre \ref{chap3}, nous avons g\'en\'eralis\'e les MCT aux MPT. Une fa\c con de d\'ecomposer les MPT $\textbf{U}_{p^{\otimes n}}\left(\sigma\right)$ en produit avec des matrices de transposition tensorielle y a \'et\'e construite. Les MCT peuvent \^etre utilis\'ees pour obtenir diff\'erentes transform\'ees de certaines \'equations matricielles en \'equations matricielles de la forme $AX=B$.\\
La g\'en\'eralisation de la relation (\ref{eq5prime}) en termes de matrices de Gell-Mann g\'en\'eralis\'ees, qui sont des matrices de la physique des particules, et qui constituent une g\'en\'eralisation des matrices de Pauli, nous permet d'esp\'erer l'utilisation des MCT dans ce domaine de la physique. A l'aide de la d\'ecomposition des MPT $\textbf{U}_{p^{\otimes n}}\left(\sigma\right)$ en produit avec des matrices de transposition tensorielle, nous avons aussi obtenu une fa\c con d'exprimer ces MPT en termes des matrices de Gell-Mann g\'en\'eralis\'ees.\\ Afin de g\'en\'eraliser la relation (\ref{eq5prime}) aux MCT $\textbf{U}_{n\otimes p}$, avec $n\neq p$, nous avons introduit ce que nous appelons matrices de Gell-Mann rectangles. L'article \cite{pushpa12} utilise les matrices \'el\'ementaires $3\times 2$. Nous nous demandons ce que nous pourrons obtenir si nous utilisons des matrices de Gell-Mann $3\times 2$ au lieu de matrices \'el\'ementaires $3\times 2$.\\
Il existe d'autres g\'en\'eralisations des matrices de Pauli, entre autres les matrices de Kibler \cite{kibler09}, les nonions 
\cite{volkov10} et la g\'en\'eralisation par produits tensoriels des matrices de Pauli elles-m\^emes \cite{rigetti04,saniga06}. Notre tentative d'exprimer la MCT $3\otimes 3$ en termes de matrices de Kibler et de nonions nous fait penser qu'il doit y avoir une autre g\'en\'eralisation des matrices de Pauli $\left(\Pi_{i}\right)_{1\leq i\leq n^2-1}$ satisfaisant 
\begin{equation*}
\textbf{U}_{n\otimes n}=\frac{1}{n}\Pi_0\otimes \Pi_0+\frac{1}{n}\sum_{i=1}^{n^2-1}\Pi_i\otimes\Pi_i
\end{equation*}
avec $\Pi_0$ la matrice unit\'e $n\times n$. Pour $n=2^p$, $p$ entier naturel non nul, cette relation est satisfaite par la g\'en\'eralisation par produits tensoriels des matrices de Pauli \cite{Christian13}. \\
 Nous avons exprim\'e l'OCE, pour huit leptons et quarks du MS d'une m\^eme g\'en\'eration, propos\'e par Zenczykowski dans la formulation dans l'espace des phases, en termes de la MCT $\textbf{U}_{2\otimes 2}$, sous une forme dont des interpr\`etations physiques seraient possibles. Nous avons remarqu\'e que cet OCE peut \^etre obtenu \`a l'aide de la relation (\ref{eq13prime}) entre les charges des leptons et des quarks. A l'aide de cette m\^eme relation nous avons construit, en termes de la MCT $\textbf{U}_{4\otimes 4}$ un OCE incluant les fermions du MS. \c Ca nous oblige \`a introduire la quatri\`eme g\'en\'eration de leptons, le mod\`ele de  leptoquark de Pati-Salam et les quarks de couleur jaune. Les expressions de ces OCE en termes de MCT nous permettent de remarquer que les valeurs propres $-1$ d'une MCT peuvent \^etre associ\'ees aux leptons tandis que les valeurs propres $+1$ peuvent \^etre associ\'ees aux quarks.

\appendix

\chapter{Matrices. Une g\'en\'eralisation}\label{appA}

Si les \'el\'ements d'une matrice sont consid\'er\'es comme les composantes d'un tenseur du second ordre, nous adoptons la   notation habituelle pour une matrice, sans crochets \`a l'int\'erieur. Tandis que si les \'el\'ements sont, par exemple, consid\'er\'es comme les composantes d'un tenseur du sixi\`eme ordre, trois fois covariant et trois fois contravariant, alors nous repr\'esentons la matrice de la mani\`ere suivante, par exemple\\
 \begin{equation}
 \textbf{M}=\left[%
\begin{array}{cc}
  \left[%
\begin{array}{cc}
  \left[%
\begin{array}{cc}
  1 & 0 \\
  1 & 1 \\
\end{array}%
\right] & \left[%
\begin{array}{cc}
  1 & 1 \\
  3 & 2 \\
\end{array}%
\right] \\
  \left[%
\begin{array}{cc}
  0 & 0 \\
  0 & 0 \\
\end{array}%
\right] & \left[%
\begin{array}{cc}
  1 & 1 \\
  1 & 1 \\
\end{array}%
\right] \\
\end{array}%
\right] & \left[%
\begin{array}{cc}
  \left[%
\begin{array}{cc}
  1 & 0 \\
  1 & 2 \\
\end{array}%
\right] & \left[%
\begin{array}{cc}
  7 & 8 \\
  9 & 0 \\
\end{array}%
\right] \\
  \left[%
\begin{array}{cc}
  3 & 4 \\
  5 & 6 \\
\end{array}%
\right] & \left[%
\begin{array}{cc}
  9 & 8 \\
  7 & 6 \\
\end{array}%
\right] \\
\end{array}%
\right] \\
  \left[%
\begin{array}{cc}
  \left[%
\begin{array}{cc}
  1 & 1 \\
  1 & 1 \\
\end{array}%
\right] & \left[%
\begin{array}{cc}
  0 & 0 \\
  3 & 2 \\
\end{array}%
\right] \\
  \left[%
\begin{array}{cc}
  4 & 5 \\
  1 & 6 \\
\end{array}%
\right] & \left[%
\begin{array}{cc}
  1 & 7 \\
  8 & 9 \\
\end{array}%
\right] \\
\end{array}%
\right] & \left[%
\begin{array}{cc}
  \left[%
\begin{array}{cc}
  5 & 4 \\
  3 & 2 \\
\end{array}%
\right] & \left[%
\begin{array}{cc}
  1 & 0 \\
  1 & 2 \\
\end{array}%
\right] \\
  \left[%
\begin{array}{cc}
  3 & 4 \\
  5 & 6 \\
\end{array}%
\right] & \left[%
\begin{array}{cc}
  7 & 8 \\
  9 & 0 \\
\end{array}%
\right] \\
\end{array}%
\right] \\
  \left[%
\begin{array}{cc}
  \left[%
\begin{array}{cc}
  1 & 2 \\
  3 & 4 \\
\end{array}%
\right] & \left[%
\begin{array}{cc}
  9 & 8 \\
  7 & 6 \\
\end{array}%
\right] \\
  \left[%
\begin{array}{cc}
  5 & 6 \\
  7 & 8 \\
\end{array}%
\right] & \left[%
\begin{array}{cc}
  5 & 4 \\
  3 & 2 \\
\end{array}%
\right] \\
\end{array}%
\right] & \left[%
\begin{array}{cc}
  \left[%
\begin{array}{cc}
  9 & 8 \\
  7 & 6 \\
\end{array}%
\right] & \left[%
\begin{array}{cc}
  5 & 4 \\
  3 & 2 \\
\end{array}%
\right] \\
  \left[%
\begin{array}{cc}
  1 & 0 \\
  1 & 2 \\
\end{array}%
\right] & \left[%
\begin{array}{cc}
  3 & 4 \\
  5 & 6 \\
\end{array}%
\right] \\
\end{array}%
\right] \\
\end{array}%
\right]\linebreak
 \end{equation}

$\textbf{M}=\left(%
  M_{j_1j_2j_3}^{i_1i_2i_3} \\
\right)$\\

$i_1i_2i_3=111,\; 112,\; 121,\; 122,\; 211,\; 212,\; 221,\; 222,\;
311,\;312,\;
321,\; 322$\\
indices de ligne\\

 $j_1j_2j_3=111,\; 112,\; 121,\; 122,\; 211,\; 212,\; 221,\; 222$\\
 indices de colonne\\
 Les premiers indices $i_1$  et $j_1$ sont les indices des crochets ext\'erieurs que nous appelons crochet du premier ordre; les deuxi\`eme indices $i_2$  et $j_2$ sont les indices des crochets suivants que nous appelons les crochets du deuxi\`emes ordre; les troisi\`emes indices $i_3$  et $j_3$  sont les indices des crochets les plus int\'erieurs, de cet exemple, que nous appelons les crochets du troisi\`emes ordre. Ainsi, par exemple, $M_{121}^{321}=5$. \\
 Si nous supprimons les crochets du troisi\`emes ordre, alors les \'el\'ements de la matrice $\textbf{M}$ sont consid\'er\'es comme les composantes d'un tenseurs d'ordre quatre, deux fois contravariant et deux fois covariant.\\
  \indent
  Consid\'erons un cas plus g\'en\'eral, $\textbf{M}=\left(%
  M_{j_1j_2\ldots j_k}^{i_1i_2\ldots i_k} \\
\right)$ o\`u les \'el\'ements de $M$ sont consid\'er\'es comme les composantes d'un tenseur d'ordre $2k$, $k$ fois contravariant et $k$ fois covariant. Les crochets du premier ordre sont les crochets d'une matrice $n_1\times m_1$; les crochets du deuxi\`emes ordre sont les crochets d'une matrice $n_2\times m_2$ ; $\cdots $; les crochets du $k$-i\`eme ordre sont les crochets d'une matrice $n_k\times m_k$. $\textbf{M}=\left(%
  \gamma_t^s \\
\right)_{1\;\leq\;s\;\leq\;n_1n_2\ldots
n_k,\;1\;\leq\;t\;\leq\;m_1m_2\ldots m_k}$ si les \'el\'ements de $\textbf{M}$
sont consid\'er\'es comme les composantes d'un tenseur du second ordre, une fois contravariant et une covariant. Alors, \cite{rao77}
\begin{equation}\label{e21}
s=n_k\ldots n_3n_2(i_1-1)+n_kn_{k-1}\ldots n_3(i_2-1)+\ldots
+n_k(i_{k-1}-1)+i_k
\end{equation}
\begin{equation}\label{e22}
t=m_k\ldots m_3m_2(j_1-1)+m_km_{k-1}\ldots m_3(j_2-1)+\ldots
+m_k(j_{k-1}+j_k
\end{equation}
 Les \'el\'ements de la matrice $\textbf{N}=\left(N_k^{ij}\right)=\left[%
\begin{array}{c}
  \left[%
\begin{array}{cc}
  1 & 1 \\
  1 & 1 \\
\end{array}%
\right] \\
  \left[%
\begin{array}{cc}
  0 & 1 \\
  2 & 3 \\
\end{array}%
\right] \\
  \left[%
\begin{array}{cc}
  4 & 5 \\
  6 & 7 \\
\end{array}%
\right] \\
  \left[%
\begin{array}{cc}
  8 & 9 \\
  0 & 1 \\
\end{array}%
\right] \\
\end{array}%
\right]$   ,
 avec crochets int\'erieurs, peut \^etre consid\'er\'es comme les composantes d'un tenseur d'ordre trois, deux fois  contravariant et une fois covariant. Alors, par exemple, $\left(N_2^{12}\right)=1$.

\chapter{Produit Tensoriel de matrices}\label{AppB}

\begin{definition}
Consid\'erons $\textbf{A}=(A^i_j)\in\mathbb{C}^{m\times
n}$, $\textbf{B}=(B^i_j)\in\mathbb{C}^{p\times
r}$. La matrice d\'efinie par
\begin{center}
$\textbf{A}\otimes \textbf{B}=
\begin{pmatrix}
  A^1_1\textbf{B} & \ldots & A^1_j\textbf{B} & \ldots & A^1_n\textbf{B} \\
  \vdots &  & \vdots &  & \vdots \\
  A^i_1\textbf{B} & \ldots & A^i_j\textbf{B} & \ldots & A^i_n\textbf{B} \\
  \vdots &  & \vdots &  & \vdots \\
  A^m_1\textbf{B} & \ldots & A^m_j\textbf{B} & \ldots & A^m_n\textbf{B} 
\end{pmatrix}%
$
\end{center}
obtenue apr\`es les multiplications par un scalar, $A^i_j\textbf{B}$, est appel\'ee le produit tensoriel de la matrice
$\textbf{A}$ par la matrice $\textbf{B}$.
\begin{center}
$\textbf{A}\otimes \textbf{B}\in\mathbb{C}^{mp\times
nr}$
\end{center}
\end{definition}
\begin{property}
Le produit tensoriel de matrices est associatif.
\end{property}
\begin{property}
Le produit tensoriel de matrices est distributif par rapport \`a l'addition.
\end{property}

\begin{property}
$\left(%
\textbf{B}_1\cdot\textbf{A}_1%
\right) \otimes \left(%
\textbf{B}_2\cdot\textbf{A}_2
\right)=\left(%
\textbf{B}_1 \otimes \textbf{B}_2
\right)\cdot\left(%
\textbf{A}_1 \otimes \textbf{A}_2
\right)$ pour toutes matrices $\textbf{B}_1$, $\textbf{A}_1$, $\textbf{B}_2$, $\textbf{B}_2$ si les produits habituels de matrices $\textbf{B}_1\cdot\textbf{A}_1$ et $\textbf{B}_2\cdot\textbf{A}_2$ sont d\'efinis.
\end{property}

\begin{proposition}
If $\textbf{A} \otimes \textbf{B}=\textbf{O}$, then  $\textbf{A}
=\textbf{O}$  or $\textbf{B}=\textbf{O}$, pour toutes matrices $\textbf{A}, \textbf{B}$.
\end{proposition}

\begin{proposition}
$I_n \otimes I_m$\;=\;$I_{nm}$
\end{proposition}

\begin{proposition}\label{thm12}
Consid\'erons $\left(%
 \textbf{A}_i%
\right)_{1\leq i\leq n\times m}$ une base de $\mathbb{C}^{n\times m}$, $\left(%
\textbf{B}_j %
\right)_{1\leq j\leq p\times r}$ une base de $\mathbb{C}^{p\times r}$. Alors, $\left(%
 \textbf{A}_i\otimes \textbf{B}_j%
\right)_{1\leq i\leq n\times m, 1\leq j\leq p\times r}$ est une base de $\mathbb{C}^{np\times mr}$.
\end{proposition}
\begin{property}\label{thm7}
Consid\'erons $\left(%
 \textbf{A}_i%
\right)_{1\leq i\leq n}$ un syst\`eme d'\'el\'ements de $\mathbb{C}^{p\times r}$, $\left(%
\textbf{B}_i %
\right)_{1\leq i\leq n}$ un syst\`eme d'\'el\'ements de
$\mathbb{C}^{l\times m}$, $\textbf{A}\in\mathbb{C}^{p\times r}$ et $\textbf{B}\in\mathbb{C}^{l\times m}$. Si 
\begin{equation}\label{e12}
\textbf{A}\otimes \textbf{B}=\displaystyle\sum_{i=1}^{n}\textbf{A}_i\otimes \textbf{B}_i
\end{equation}
 alors, pour toute matrice $\textbf{K}$, 
 \begin{equation*}
 \textbf{A}\otimes \textbf{K}\otimes \textbf{B}=\displaystyle\sum_{i=1}^{n}\textbf{A}_i\otimes \textbf{K}\otimes \textbf{B}_i   
\end{equation*}
\end{property}
\begin{proof}
Prenons $K=\left(K_{j_2}^{j_1}\right)\in\mathbb{C}^{q\times s}$, $A_i=\left(A_{(i)k_2}^{\;\;\;\;k_1}\right)$, $A=\left(A_{k_2}^{k_1}\right)$, $B_i=\left(B_{(i)l_2}^{\;\;\;\;l_1}\right)$, $B=\left(B_{l_2}^{l_1}\right)$. Alors, $A_{k_2}^{k_1}K_{j_2}^{j_1}B_{l_2}^{l_1}=K_{j_2}^{j_1}A_{k_2}^{k_1}B_{l_2}^{l_1}$ et $\displaystyle\sum_{i=1}^{n}A_{(i)k_2}^{\;\;\;\;k_1}K_{j_2}^{j_1}B_{(i)l_2}^{\;\;\;\;l_1}=K_{j_2}^{j_1}\displaystyle\sum_{i=1}^{n}A_{(i)k_2}^{\;\;\;\;k_1}B_{(i)l_2}^{\;\;\;\;l_1}$ sont respectivement les \'el\'ements de $A\otimes K\otimes B$ et $\displaystyle\sum_{i=1}^{n}A_i\otimes K\otimes B_i$ sur la $k_1j_1l_1$-i\`eme ligne, $k_2j_2l_2$-i\`eme colonne. En utilisant  \eqref{e12}, $A_{k_2}^{k_1}B_{l_2}^{l_1}=\displaystyle\sum_{i=1}^{n}A_{(i)k_2}^{\;\;\;\;k_1}B_{(i)l_2}^{\;\;\;\;l_1}$. Donc, $A_{k_2}^{k_1}K_{j_2}^{j_1}B_{l_2}^{l_1}=\displaystyle\sum_{i=1}^{n}A_{(i)k_2}^{\;\;\;\;k_1}K_{j_2}^{j_1}B_{(i)l_2}^{\;\;\;\;l_1}$.
 Ceci est vraie pour tous indices $k_1$, $j_1$, $l_1$, $k_2$, $j_2$, et $l_2$. D'o\`u, $A\otimes K\otimes B= \displaystyle\sum_{i=1}^{n}A_i\otimes K\otimes B_i$.   
\end{proof}

\backmatter


\begin{thebibliography}{99}
\bibitem {Lin12} Lin M.M., \textit{On the T-Stein equation $X = AX^{T}B + C$}, arXiv: 1210.5731v1
\bibitem {Zhang11}Zhang L. and Wu J., \textit{A Survey of Dynamical Matrices Theory}, arXiv: 1009.2210v5
\bibitem {Gilchrist11} Gilchrist A., Terno D.R. and Wood C.J., \textit{Vectorization of Quantum Operations and its Use}, arXiv: 0911.2539v2

\bibitem {Fujii01}Fujii K.: \textit{Introduction to Coherent States and Quantum Information Theory}, arXiv: quant-ph/0112090, prepared for 10th Numazu Meeting on Integral System, Noncommutative Geometry and Quantum theory, Numazu,
Shizuoka, Japan, 7-9 Mai 2002(2002).
\bibitem {Chefles00} Chefles A., Gilson C.R. and Barnett S.M., \textit{Entanglement, Information and Multiparticle Quantum Operations}, arXiv: quant-ph/0006106v2
\bibitem {Wilmott08} Wilmott C. and Wild P., \textit{On Interchanging the States of a Pair of Qudits}, arXiv: 0811.1545v1
\bibitem{Wang01}  Wang R.P., \textit{Varieties of Dirac Equation and Flavors of Leptons and Quarks}, arXiv: hep-ph/0107184v2.
\bibitem {Rakotonirina03}Rakotonirina, C. (2003). \textit{Produit Tensoriel de Matrices en Th\'{e}orie de Dirac}, Th\`{e}se de
Doctorat de Troisi\`{e}me Cycle, Universit\'{e} d'Antananarivo, Antananarivo, Madagascar,
(2003).
\bibitem {Zhou12} Zhou B., Lam J. and Duan G-R., \textit{Toward Solution of Matrix Equation $X = Af(X)B + C$}, arXiv: 1211.0346v1
\bibitem {Li13} Li N., Wang Q.-W., and Jiang J., \textit{An Efficient Algorithm for the Reflexive Solution of the
Quaternion Matrix Equation $AXB+CX^{H}D=F$}, Journal of Applied Mathematics, vol. 2013, Article ID
217540, 14 pages, 2013.  
\bibitem{Zenczykowski07}  Zenczykowski P., \textit{Space, Phase Space and Quantum Numbers of Elementary Particles}, Acta Phys. Pol. B \textbf{38}, 2053 (2007).
\bibitem {Raoelina11}Raoelina Andriambololona, Rakotonirina C. \textit{A Study of the Dirac-Sidharth Equation}, Electronic Journal of Theoretical Physics, EJTP 8, No. 25 (2011) 177-182

\bibitem {Einstein1905}Einstein, A. \textit{On the Electrodynamics of Moving Bodies}, Ann. der Phys. 17, 891-921(1905).               \bibitem {Snyder472} Snyder, H.S., \textit{The Electromagnetic Field in Quantized Space-Time}, Phys.Rev., Vol.72, No.1, July 1, 68-71(1947).
\bibitem {Snyder471} Snyder, H.S., \textit{Quantized Space-time}, Phys.Rev., Vol.71, No.1, January 1, 38-41(1947).
 \bibitem{Sidharth04}Sidharth, B.G., \textit{Discrete Space-Time and Lorentz Symmetry}, Int.J.Th.Phys., Vol.43, No.9, September, 1857-1861(2004).
  \bibitem {Glinka101}Glinka, L.A. \textit{CP violation, massive neutrinos, and its chiral condensate: new results from Snyder noncommutative geometry}. Apeiron 17 (4), 2010, 223-242(2010).
\bibitem {Glinka102}Glinka, L.A. \textit{Energy renormalization and integrability within the massive neutrinos model}, Apeiron 17 (4), 243-271(2010).
\bibitem {Glinka11}Glinka, L.A. \textit{\AE thereal Multiverse: Selected Problems of Lorentz Symmetry Violation, Quantum Cosmology, and Quantum Gravity}, arXiv:1102.5002v2(2011).
\bibitem {Sidharth05}Sidharth, B.G., \textit{A Note on Non-commutativity and Mass Generation},  Int.J.Mod.Phys.E., Vol.14, No.6, ,
    927-929(2005).
\bibitem {Sakurai67}Sakurai, J.J., \textit{Advanced Quantum Mechanics},
    Addison Wesley Publishing Company, 308-311(1967).
\bibitem {BjorkenDrell64}  Bjorken J.D. and  Drell S.D., \textit{Relativistic
    Quantum Mechanics}, Mc-Graw Hill, New York, , 39(1964).
    \bibitem {Sidharth09}Sidharth, B.G. \textit{The Mass of the Neutrinos}, arXiv: 0904.3639v2(2009).
\bibitem {messiah58}Messiah A., \textit{M\'ecanique Quantique}, Dunod, Paris (1958).
		
\bibitem {Faddev95}Faddev L.D.:\textit{Algebraic Aspects of the Bethe Ansantz}, Int.J.Mod.Phys.A, 10, No 13, May, 1848(1995).
\bibitem {Verstraete02}Verstraete F.: \textit{A Study of Entanglement in Quantum Information Theory}, Th\`{e}se de Doctorat, Katholieke Universiteit Leuven, 90-93(2002).
\bibitem{Rakotonirina05} Rakotonirina C., \textit{Tensor Permutation Matrices in Finite Dimensions}, arXiv: math.GM/0508053 (2005)
\bibitem{Rakotonirina07} Rakotonirina C., Hanitriarivo Rakotoson, \textit{Expressing a Tensor Permutation Matrix} $p^{\otimes n}$ \textit{in Terms of the Generalized Gell-Mann Matrices}, HEP-MAD-2007-317, http://www.slac.stanford.edu/econf/C0709107/pdf/317.pdf (2007) 
		
\bibitem{Rakotonirina09} Rakotonirina C., \textit{On the Cholesky method}, Journal of Interdisciplinary Mathematics, Vol. 12(2009), No.6 , pp. 875-882


	\bibitem{Rakotonirina11} Rakotonirina C., \textit{On the Tensor Permutation Matrices}, arXiv: 1101.0910 (2011)
	
\bibitem {Raoelina86}Raoelina Andriambololona: \textit{Alg\`{e}bre lin\'{e}aire et Multilin\'{e}aire. Applications}, tome 1, Collection LIRA, Madagascar, (1986).

\bibitem {Merris97}Merris.R, \textit{Multilinear Algebra}, Gordon and Breach Sciences Publishers,(1997).
\bibitem {Raoelina72}Raoelina Andriambololona, \textit{Etude des Matrices Carr\'ees de Dimension 2 et quelques Applications}, Ann.Univ.Madagascar, S\'{e}rie Sc. Nature et Math, No 9, (1972).
\bibitem{DemidovitchMaron79} D\'{e}midovitch B. and  Maron I., \textit{El\'{e}ments de Calcul Num\'{e}rique}. Editions Mir,
(1979).
\bibitem {Nougier87} Nougier J.P., \textit{M\'{e}thodes de Calcul Num\'{e}rique}. Masson. Paris, (1987).
\bibitem{KincaidCheney99} Kincaid D. and  Cheney W., \textit{Numerical Mathematics and Computing}.
Brooks/Cole Publishing Co., Pacific Grove, CA, fourth edition, (1999). 
\bibitem{Steven04} Steven E.P., \textit{Numerical Methods Course Notes}.
http://scicomp.ucsd.edu/ spav/pub/numas.pdf
\bibitem{Ikramov77} Ikramov H., \textit{Recueil de Probl\`{e}mes d'Alg\`{e}bre lin\'{e}aire}, Edition Mir, Moscou,
(1977).
\bibitem{Rakotonirina08} Rakotonirina C., \textit{Expression of a Tensor Commutation Matrix in Terms of
the Generalized Gell-Mann Matrices}, International Journal of Mathematics and Mathematical Sciences,
Volume 2007, Article ID 20672
\bibitem{Rakotonirina12} Rakotonirina C., \textit{Rectangle Gell-Mann Matrices}, International Mathematical Forum, Vol. 6, 2011, no. 2, 57 - 62
	\bibitem{Rakotonirina13} Rakotonirina C., Ratiarison A.A., \textit{Swap Operators and Electric Charges of Fermions}, International Journal of Theoretical and Applied Physics (IJTAP), Vol.3, No. I (June 2013), pp. 15-24.
	
\bibitem {Narison89} Narison S., \textit{Spectral Sum Rules}, vol. 26 of World Scientific Lecture Notes in Physics, World Scientific, Singapore, (1989).
\bibitem{ref5}  Pati J. and  Salam A., \textit{Lepton Number as the Fourth Color}, Phys. Rev. D10, 275, (1974).                  
\bibitem{ref6}  Baez J. and  Huerta J., \textit{The Algebra of Grand Unified Theories}, Bull. Am. Math. Soc. 47 (2010) 483-552 arXiv:0904.1556 [hep-th] (2009). 
\bibitem{rao77} Raoelina Andriambololona, \textit{Etude Intrins\`eque et Repr\'esentation Matricielle des Produits Kroneckeriens et des Puissances Kroneckeriennes d'Op\'erateurs Lin\'eaires - Etude G\'en\'erale}, Ann.Univ.Madagascar, S\'{e}rie Sc. Nature et Math, No 14, (1977).
\bibitem{pushpa12} Pushpa, Bisht P. S., Tianjun Li and Negi O. P. S., \textit{Quaternion Octonion Reformulation of Grand
Unified Theories}, arXiv: 1205.4617v1, (2012).
\bibitem{kibler09} Kibler M.R., \textit{An angular momentum approach to quadratic Fourier transform, Hadamard matrices, Gauss sums, mutually unbiased bases, the unitary group and the Pauli group}, J. Phys. A: Math. Theor. 42 353001 (28pp) (2009).
\bibitem {volkov10} Volkov G., \textit{Ternary Quaternions and Ternary $TU(3)$ algebra}, arXiv: 1006.5627v1 (2010).
\bibitem {rigetti04} Rigetti C., Mosseri R. and Devoret M., \textit{Geometric Approach to Digital Quantum Information}, Quantum Information Processing, Vol. 3, No. 6, December 2004 (2004)  
\bibitem {saniga06} Saniga M., Planat M. and Pracna P., \textit{Projective Ring Line Encompassing Two-Qubits}, hal-00111733, version 5 - 28 Dec 2006, (2006).
\bibitem {Christian13} Rakotonirina C., Rakotondramavonirina J., \textit{Tensor Commutation Matrices and Some Generalizations of the Pauli Matrices}, en pr\'eparation, (2013).

 
\end{thebibliography}
\end{document}